\documentclass[a4paper,onecolumn,superscriptaddress,11pt,accepted=2021-04-08]{quantumarticle}
\pdfoutput=1
\usepackage[utf8]{inputenc}

\usepackage[english]{babel}
\usepackage[T1]{fontenc}
\usepackage{amsmath}
\usepackage[unicode=true,bookmarks=true,bookmarksnumbered=false,bookmarksopen=false,breaklinks=false,pdfborder={0 0 1}, backref=false,colorlinks=true]{hyperref}

\usepackage[dvipsnames]{xcolor}
\definecolor{myred}{rgb}{0.8,0.3,0.2}
\definecolor{myredfill}{rgb}{1,0.93,0.93}
\definecolor{myblue}{rgb}{0.05,0.25,0.7}
\definecolor{mybluefill}{rgb}{0.92,0.98,1}
\definecolor{mygreen}{rgb}{0.3,0.6,0.4}
\definecolor{mygreenfill}{rgb}{0.93,1,0.97}
\definecolor{mygrey}{rgb}{0.4,0.4,0.4}
\definecolor{mygreyfill}{rgb}{0.95,0.95,0.95}
\definecolor{mypurple}{rgb}{0.6,0.2,0.6}
\definecolor{mypurplefill}{rgb}{0.99,0.94,0.99}
\definecolor{myyellow}{rgb}{9,0.9,0.3}
\definecolor{myyellowfill}{rgb}{1,1,0.92}
\definecolor{myorange}{rgb}{9,0.7,0.3}
\definecolor{myorangefill}{rgb}{1,0.95,0.92}

\definecolor{linkred}{rgb}{0.6,0.1,0.1}
\definecolor{citeblue}{rgb}{0.1,0.2,0.8}
\definecolor{urlblue}{rgb}{0.2,0.45,0.65}

\hypersetup{
     linkcolor = linkred,
     anchorcolor = linkred,
     citecolor = citeblue,
     filecolor = linkred,
     urlcolor = urlblue
     }

\usepackage{bbm}
\usepackage{stmaryrd}
\usepackage{tabu,multirow}
\usepackage{float} 
\usepackage{lipsum}
\usepackage[autostyle]{csquotes}
\usepackage{tikz-cd}
\usepackage{amsmath,amsthm,amssymb,latexsym}
\usepackage{mathrsfs}
\usepackage{todonotes}
\usepackage{hhline}
\usepackage[caption=false]{subfig}
\usetikzlibrary{backgrounds}
\usepackage{graphicx}

\newcommand{\EB}[0]{{\mathscr{B}}}
\newcommand{\Q}[0]{{\mathscr{Q}}}
\newcommand{\s}[0]{{{s}}}
\newcommand{\e}[0]{{{e}}}

\newcommand{\an}[0]{{{a}}}

\newcommand{\se}[0]{{{se}}}
\newcommand{\sa}[0]{{{sa}}}
\newcommand{\ea}[0]{{{ea}}}

\newcommand{\appp}[0]{{{a'}}}

\newcommand{\sezspap}[0]{{{sezs'a'}}}
\newcommand{\spap}[0]{{{s'a'}}}
\newcommand{\szsp}[0]{{{szs'}}}

\newcommand{\zs}[0]{{{sz}}}
\newcommand{\zsp}[0]{{{zs'}}}
\newcommand{\z}[0]{{{z}}}
\newcommand{\sezsp}[0]{{{sezs'}}}

\DeclareMathOperator*{\functioncomposition}{\bigcirc}
\DeclareMathOperator*{\tensorcomposition}{\bigotimes}

\usepackage[numbers,sort&compress]{natbib}

\newtheorem{theorem}{Theorem}
\newtheorem{definition}{Definition}
\newtheorem{lemma}{Lemma}

\begin{document}

\title{Resource theories of multi-time processes: A window into quantum non-Markovianity}
\date{April 8, 2021}
\author{Graeme D.\ Berk}
\email{graeme.berk@monash.edu}
\affiliation{School of Physics and Astronomy, Monash University, Clayton, Victoria 3800, Australia}
\author{Andrew J.\ P.\ Garner}
\affiliation{Institute for Quantum Optics and Quantum Information, Austrian Academy of Sciences, Boltzmanngasse 3, A-1090 Vienna, Austria}
\affiliation{School of Physical and Mathematical Sciences, Nanyang Technological University, 21 Nanyang Link, 637371, Singapore}
\author{Benjamin Yadin}
\affiliation{School of Mathematical Sciences, University of Nottingham,
University Park, Nottingham NG7 2RD, United Kingdom.}
\author{Kavan Modi}
\email{kavan.modi@monash.edu}
\affiliation{School of Physics and Astronomy, Monash University, Clayton, Victoria 3800, Australia}
\author{Felix A. Pollock}
\email{felix.pollock@monash.edu}
\affiliation{School of Physics and Astronomy, Monash University, Clayton, Victoria 3800, Australia}

\begin{abstract}
We investigate the conditions under which an uncontrolled background process may be harnessed by an agent to perform a task that would otherwise be impossible within their operational framework. This situation can be understood from the perspective of resource theory: rather than harnessing `useful' quantum states to perform tasks, we propose a resource theory of quantum processes across multiple points in time. Uncontrolled background processes fulfil the role of resources, and a new set of objects called \textit{superprocesses}, corresponding to operationally implementable control of the system undergoing the process, constitute the transformations between them. After formally introducing a framework for deriving resource theories of multi-time processes, we present a hierarchy of examples induced by restricting quantum or classical communication within the superprocess -- corresponding to a client-server scenario. The resulting nine resource theories have different notions of quantum or classical memory as the determinant of their utility. Furthermore, one of these theories has a strict correspondence between non-useful processes and those that are Markovian and, therefore, could be said to be a true `quantum resource theory of non-Markovianity'.
\end{abstract}

\maketitle

\section{Introduction} \label{introduction}
Before the invention of motors powered by hydrocarbon fuel, our ancestors were forced to rely on less tangible energy sources to power their voyages across the oceans. Sailing is a form of propulsion which works by actively harnessing the energy of an uncontrolled background process, namely the wind. Fast forward a few hundred years, and we are at the cusp of another technological revolution, which will be based on the logic of quantum mechanics. As quantum technology matures, understanding the scope of experimental control has become an area of significant focus. Most efforts to improve control over quantum systems have focused on reducing the amount of influence which the environment can exert over the system; these efforts include error correction~\cite{quantumerrorcorrection}, decoupling~\cite{dyndec}, or simply engineering cleaner quantum systems. Our approach is entirely distinct from these, as we choose to make the most of the environment which will inevitably be present.

To see what `sailing' through Hilbert space might look like, consider the following scenario. Several agents act in sequence on a quantum system with some goal in mind, be it to extract work from the system, prepare it in a particular state, or to send messages to each other. Between actions, the system is subject to an uncontrolled noisy process -- through interactions with its surroundings -- which may be temporally correlated. Given that the agents may be limited in the actions they can perform, and the degree to which they can communicate to one another, how can we quantify their ability to achieve their goals given the background process? 

An illustrative example of a (classical) process is given in Fig.~\ref{fig:equations}. Here an agent wishes to turn on a light by converting a background process into useful work. The agent usually outsources this task to a contractor, who can use a wind turbine or a solar panel to generate electricity depending on the weather, stormy or sunny respectively. On the other hand, a wind turbine (solar panel) is useless on windless sunny (stormy) day. This example begs the question: under what conditions is it possible to extract useful work or information out of an uncontrolled background quantum process?

\begin{figure}[t] 
\begin{minipage}[t]{0.45\textwidth}
        \centering
        \includegraphics[width=\linewidth]{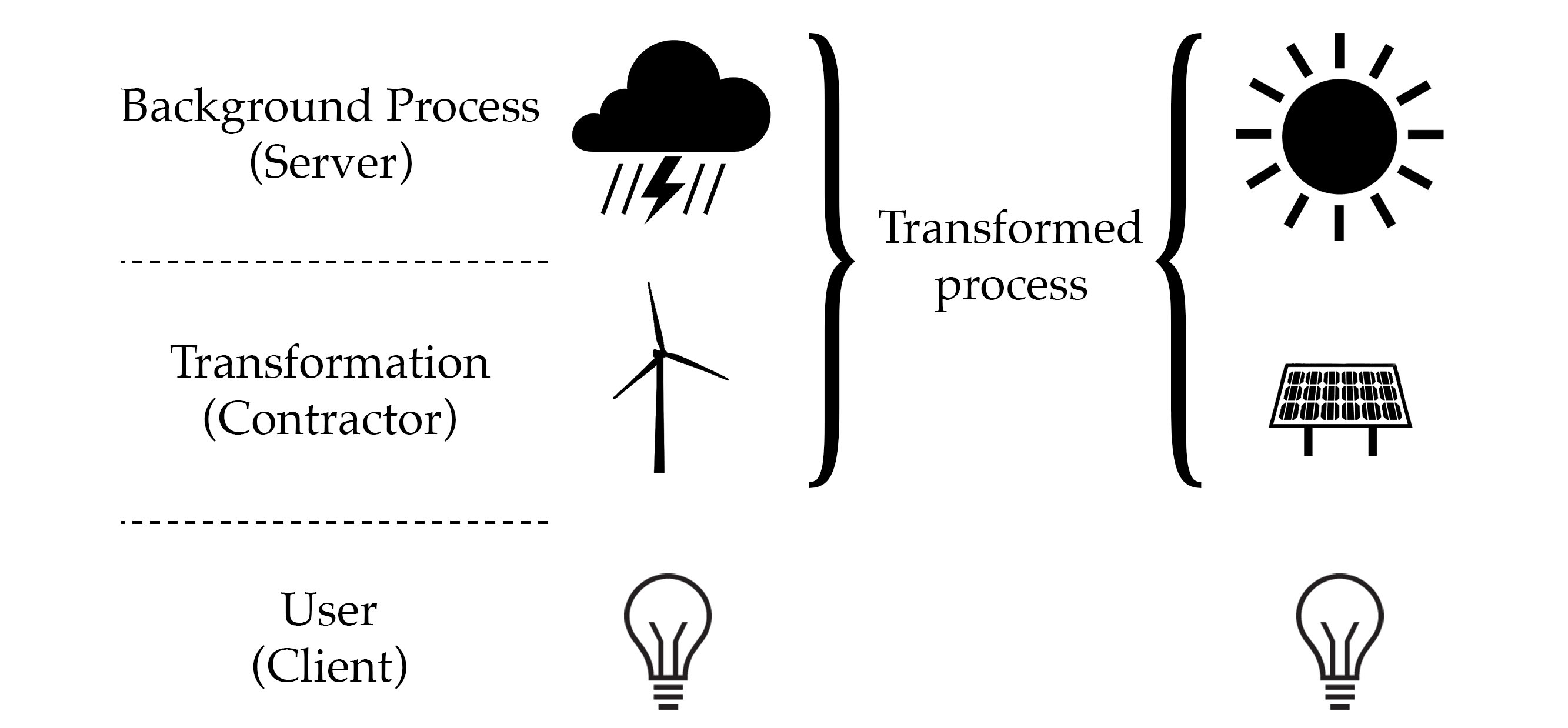}
        (a)
    \end{minipage}
    \hspace{5mm}
    \begin{minipage}[t]{0.3\textwidth}
        \centering
        \includegraphics[width=\linewidth]{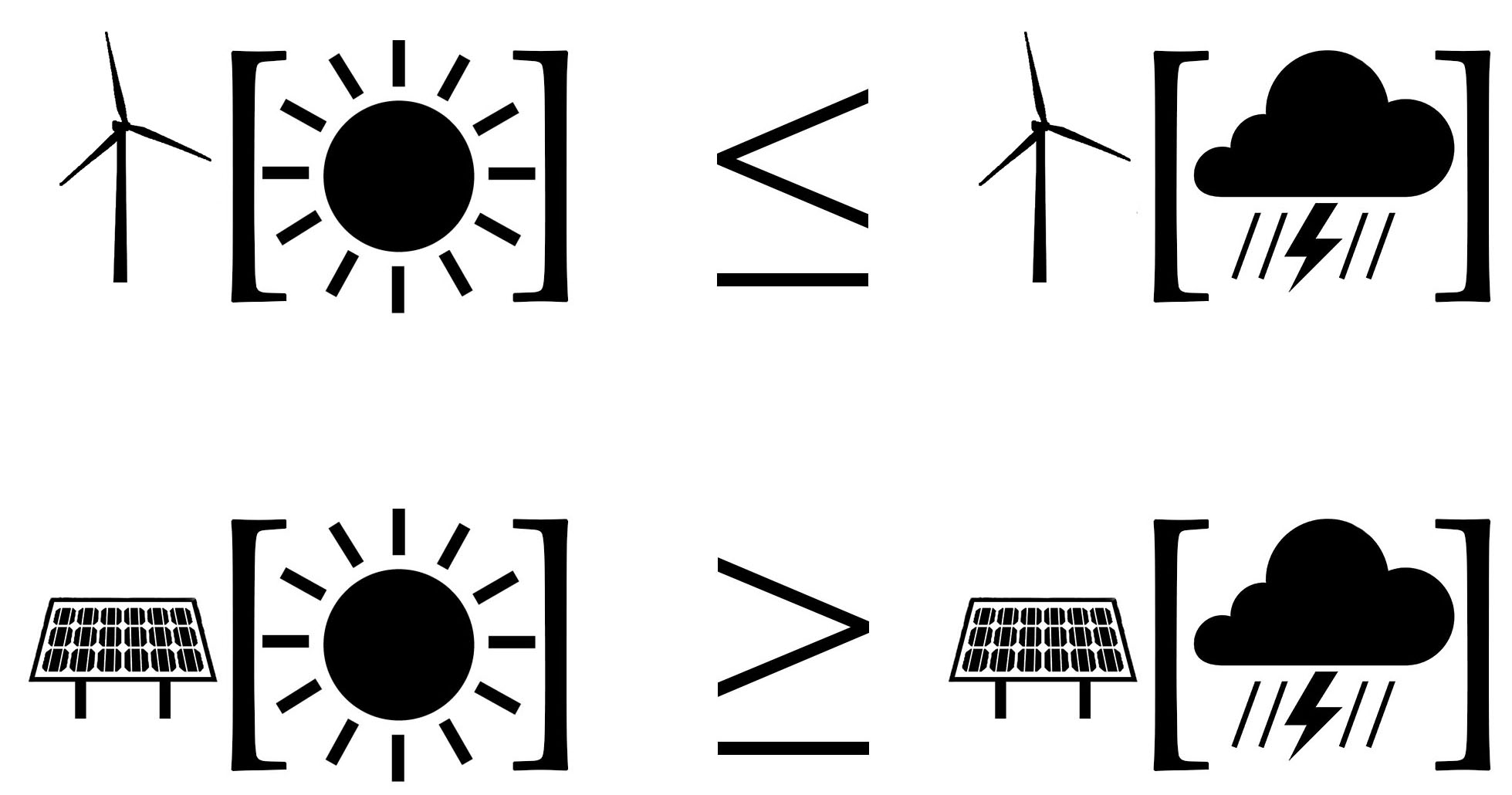}
        (b)
    \end{minipage}
\caption{(a) A restricted agent seeks to harness an uncontrolled background process (weather). They achieve this via a transformation of that process (converting energy into electricity) such that the resultant object can be utilised to perform a useful task (switching on a light-bulb). This idea can also be cast into a scenario where a client interacts with a server via a contractor as elaborated in Sec.~\ref{superprocesssection} and Sec.~\ref{communicationtheories}. (b) The utility of a background process is dependent on the capacity of an agent to harness it. If the agent can access to a wind turbine, the wind will be more useful than the sun. If that agent has access to a solar panel, the converse will be true.} \label{fig:equations}
\end{figure}

Quantum resource theories provide a framework in which the usefulness of objects for a particular task can be formally quantified. They are formulated in terms of a set of potential \emph{resources}, and a set of allowed, or free, \emph{transformations} between them. Usefulness is then determined by the set of other resources can be reached through these transformations. Quantum resource theories typically treat the set of quantum states (density operators) $\rho \in \mathcal{B}(\mathcal{H})$ of a system as the resources, with completely positive trace preserving (CPTP) maps $\mathcal{E}$ -- also known as quantum channels -- as the transformations. Common examples of quantum resource theories, with restricted subsets of CPTP maps, are thermal operations~\mbox{\cite{thermal, thermrev,rtofworkandheat, fundamentalworkcost, fundamentallimitations}}, noisy operations~\cite{rtpurity, rtinformationalnonequilibrium}, and local or separable operations~\cite{quantumentanglement, entanglementtheoryandthesecondlaw}. By exploring various properties of quantum resources, such as asymptotic conversion~\cite{reversibleframework} and rates of resource exchange~\cite{firstlaw}, we can understand which properties of states are useful under a particular set of allowed operations. For example with local operations, entanglement is useful~\cite{review}, whereas in thermodynamics, athermality is useful~\cite{thermal}, and for stabiliser computation, magic states are useful~\cite{magicstates}.

Yet, these resource theories do not capture the usefulness of more general quantum processes, which take place over several points in time and which may involve temporal correlations mediated by an inaccessible environment i.e., non-Markovian memory. In fact, it appears that non-Markovianity is the norm rather than the exception; it must be accounted for in real quantum computers~\cite{ibm}, in quantum metrology with noisy environments~\cite{quantummetrologyinnonmarkovian}, and in many realistic quantum control tasks~\cite{nmcontrol}. Furthermore, several researchers have suggested that non-Markovianity may be useful for certain tasks: it has been shown to improve the performance of quantum heat engines~\cite{nmengines, thermodynamicutilityofnonmarkovianity}, enhance quantum control~\cite{nmcontrol, thermocontrol, nmqtech, informationbackflowentanglement}, reduce decoherence~\cite{nmmemory} and allow the perfect teleportation of mixed states~\cite{tele}. However, there is no consensus on how to treat temporal correlations as a resource.

Here, we develop a framework for defining resource theories of multi-time quantum processes induced by limited experimental control, and use it to derive a family of theories in which non-Markovianity and related quantities become useful resources. While there have been numerous attempts to quantify the utility of non-Markovianity, giving rise to various resource theories~\cite{theoryofmaps, tripartite}, the main advantage of our work is that resource value is based on an operationally well defined framework for quantum processes, called the process tensor formalism~\cite{introqod, nonmarkov, nm2}. The process tensor formalism can account for the influence of (brief) experimenter interventions, does not suffer from the initial correlation problem~\cite{introqod}, and can identify temporal correlations even in CP divisible processes~\cite{divisibility}. This framework is a useful subset of the more general frameworks of quantum networks~\cite{quantumnetworks} and higher order quantum maps~\cite{higherorder}.

We begin in the next section by introducing the process tensor, a description of non-Markovian processes in terms of higher order quantum maps, and then go on to use it in Sec.~\ref{sec:ProcessResourceTheories} to show how restrictions on experimental control can lead to meaningful resource theories, in which processes themselves play the role of resources.\footnote{See Appendix~\ref{appnotationsummary} for a compact summary of the meanings of the various types of objects we use.} In Sec.~\ref{communicationtheories}, we consider the special case of restrictions on communication from past to future, demonstrating how different kinds of temporal correlations form a hierarchy of resources, before concluding in Sec.~\ref{sec:conclusion}. To start with, we will elucidate the general scenario we have in mind.

\section{Quantum Processes} \label{processessection}
Our description of quantum processes is an operational one: we explicitly account for what is within the control of some hypothetical agent, and what is not. There are two reasons why an agent may not have total control over a system $\s$ of interest. Firstly, there may be degrees of freedom $\e$ (the environment) with which the system interacts that are not directly accessible to the agent. Secondly, the agent may only have a limited capacity to influence the parts of the system that they can directly access, though they may be able to involve a separate ancillary system $\an$ in their interactions with $\s$. 

We assume a setting where the agent can effectively act on $s$ and $a$ instantaneously at a series of discrete times, between which it interacts continuously with $e$. In this case, the actions of the agent, which may potentially include any physically allowed transformation (including doing nothing at all), can be represented by completely positive (CP) trace non-increasing maps $\mathcal{A}^{\sa}$. In general, these will only be realised conditionally on some event, such as a particular outcome of a measurement, but they can always be grouped together into sets, corresponding to experimental instruments $\mathcal{J}$, such that their average, unconditional action is a CPTP map~\cite{nonmarkov}. 

A temporally ordered set of possibly restricted actions the agent is able to perform, conditional or otherwise $\{ \mathcal{A}^{\sa}_{n-1} , \dots, \mathcal{A}^{\sa}_0  \}$ is called a \textit{control sequence}; here the subscript denotes the `time-step'. However, the ancillary system $a$ can act as a quantum memory in general, allowing the agent to effectively correlate their actions on $s$ across time~\cite{stochtraj, introqod}. To represent their multi-time action on the system, we can compose these actions on the ancillary space only, denoted by $\circ_\an$\footnote{Equivalent to composing the Choi states of the maps with the link product defined in Ref.~\cite{quantumnetworks}.}; this composition implies that there are no intermediate actions on the corresponding space. The result is a higher-order quantum map~\cite{higherorder} $\mathbf{A}_{n-1:0} = {\rm tr}_\an\{\mathcal{A}^\text{\sa}_{n-1} \circ_\an \dots \circ_\an \mathcal{A}^{\sa}_0[\rho^\an_0]\}$, depicted in blue in Fig.~\ref{processtensor}, that encodes these correlations and acts on $s$ alone, albeit at multiple times (here $\rho_0^{\an}$ is the initial state of the ancilla). When the actions can be applied unconditionally, this object satisfies a hierarchy of causality conditions and is referred to as a quantum comb~\cite{quantumcircuitarch}. Unlike these actions, interactions between $s$ and $e$ are outside of the agent's influence, and need not be subject to the same limitations. Furthermore, the environment may become non-trivially correlated with $s$, leading to the breakdown of a description in terms of dynamical maps~\cite{introqod}. Between every pair of actions $\mathcal{A}^{\sa}_{j}$ and $\mathcal{A}^{\sa}_{j+1}$, there is a CPTP map $\mathcal{E}^{\se}_{j+1:j}$ acting jointly on $s$ and $e$. The collection of these maps composed only over $\e$ (denoted by $\circ_{\e}$) forms another quantum comb $\mathbf{T}_{n:0} = {\rm tr}_\e \{\mathcal{E}^\text{\se}_{n:n-1} \circ_\e \dots \circ_\e \mathcal{E}^{\se}_{1:0} \circ_\e \rho^\se_0\}$\footnote{The notation $\mathcal{E}^{\se}_{1:0} \circ_\e \rho^\se_0$ here implies that the map acts only on the $\e$ space of $\rho^\se_0$, while the $\s$ parts of both objects are left free to act or be acted on.}, known as the \textit{process tensor}. Given a potentially correlated initial $se$ state $\rho^{\se}_0$, the state after $n$ actions can be written as
\begin{equation}
\label{orderingofprocess}
\rho^{\s}_n = \mathbf{T}_{n:0}[\mathbf{A}_{n-1:0}] 
= \text{tr}_\ea \left\{ \mathcal{E}^{\se}_{n:n-1} \circ \mathcal{A}^{\sa}_{n-1}  \circ  \dots \circ \mathcal{E}^{\se}_{1:0}   \circ \mathcal{A}^{\sa}_0 [\rho^{\se}_0\otimes\rho^{\an}_0] \right\} .
\end{equation}
The second equality indicates that the interleaved $\se$ and $\sa$ dynamics can be seen as the contraction of the two higher order quantum maps, as depicted in Fig.~\ref{processtensor}.  The first is the control sequence $\mathbf{A}_{n-1:0}$, and the second is the process tensor $\mathbf{T}_{n:0}$~\cite{nm2}. When the former can only be applied conditionally on some measurement outcome, the final state $\rho^{\s}_n$ will be subnormalised. The process tensor uniquely encodes all information about a quantum process which is not under the direct control of an agent, though a consistent set of maps $\mathcal{E}^{\se}_{j+1:j}$ and state $\rho^{\se}_0$, can be non-uniquely determined by the agent, in principle, through a generalised quantum process tomography~\cite{introqod}. While we consider a version here which maps control sequences to quantum states, our results apply equally well to other quantum combs, such as those with an additional quantum state as input, which can be seen as maps from control sequences to quantum channels.

\begin{figure*}
\centering
\begin{tikzpicture}[scale=0.65]

\draw[myredfill,fill=myredfill, ultra thick,solid,rounded corners=6] (5.9,2.1) -- (5.9,0.1)  -- (4.1,0.1) -- (4.1,2.1) -- (1.9,2.1) -- (1.9,0.1) --  (0.1,0.1) --  (0.1,4.1) -- (14.6,4.1) -- (14.6,0.1) -- (12.1,0.1) -- (12.1,2.1) -- (9.9,2.1) -- (9.9,0.1) -- (8.1,0.1) -- (8.1,2.1);

\draw[myred, thick,solid,rounded corners=6] (6.5,2.1) -- (5.9,2.1) -- (5.9,0.1)  -- (4.1,0.1) -- (4.1,2.1) -- (1.9,2.1) -- (1.9,0.1) --  (0.1,0.1) --  (0.1,4.1) -- (14.6,4.1) -- (14.6,0.1) -- (12.1,0.1) -- (12.1,2.1) -- (9.9,2.1) -- (9.9,0.1) -- (8.1,0.1) -- (8.1,2.1) -- (7.5,2.1);

\draw[mybluefill,fill=mybluefill, ultra thick,solid,rounded corners=6] (6.1,1.9) -- (6.1,-0.1)  -- (3.9,-0.1) -- (3.9,1.9) -- (2.1,1.9) --  (2.1,-0.1) --  (0.1,-0.1) --  (0.1,-1.9) -- (14.6,-1.9) -- (14.6,0) -- (11.9,-0.1) -- (11.9,1.9) -- (10.1,1.9) -- (10.1,-0.1) -- (7.9,-0.1) -- (7.9,1.9);

\draw[myblue,thick,solid,rounded corners=6] (6.5,1.9) --  (6.1,1.9) -- (6.1,-0.1)  -- (3.9,-0.1) -- (3.9,1.9) -- (2.1,1.9) --  (2.1,-0.1) --  (0.1,-0.1) --  (0.1,-1.9) -- (14.6,-1.9) -- (14.6,-0.1) -- (11.9,-0.1) -- (11.9,1.9) -- (10.1,1.9) -- (10.1,-0.1) -- (7.9,-0.1) -- (7.9,1.9) -- (7.5,1.9);

\draw[black, very thick,solid] (13.7,3) -- (14.3,3);
\draw[black, very thick,solid] (14.2,2.7) -- (14.4,3.3);
\draw[black, very thick,solid] (13.7,1) -- (15,1);
\draw[black, very thick,solid] (1.7,3) -- (4.3,3);
\draw[black, very thick,solid] (9.7,3) -- (12.3,3);
\draw[black, very thick,solid] (5.7,3) -- (6.3,3);
\draw[black, very thick,solid] (1.7,1) -- (2.3,1);
\draw[black, very thick,solid] (3.7,1) -- (4.3,1);
\draw[black, very thick,solid] (9.7,1) -- (10.3,1);
\draw[black, very thick,solid] (11.7,1) -- (12.3,1);
\draw[black, very thick,dashed] (6.3,3) -- (7.7,3);
\draw[black, very thick,solid] (7.7,3) -- (8.3,3);
\draw[black, very thick,solid] (5.7,1) -- (6.3,1);
\draw[black, very thick,dashed] (6.3,1) -- (7.7,1);
\draw[black, very thick,solid] (7.7,1) -- (8.3,1);
\draw[black, very thick,solid] (1.7,-1) -- (2.3,-1);
\draw[black, very thick,solid] (3.7,-1) -- (6.3,-1);
\draw[black, very thick,dashed] (6.3,-1) -- (7.7,-1);
\draw[black, very thick,solid] (7.7,-1) -- (10.3,-1);
\draw[black, very thick,solid] (11.7,-1) -- (14.3,-1);
\draw[black, very thick,solid] (14.2,-1.3) -- (14.4,-0.7);

\draw[myblue,fill=mybluefill,ultra thick,solid,rounded corners=2] (2.3,1.7) rectangle (3.7,-1.7);
\draw[myblue,fill=mybluefill,ultra thick,solid,rounded corners=2] (10.3,1.7) rectangle (11.7,-1.7);
\draw[myred,fill=myredfill,ultra thick,solid,rounded corners=2] (4.3,3.7) rectangle (5.7,0.3);
\draw[myred,fill=myredfill,ultra thick,solid,rounded corners=2] (8.3,3.7) rectangle (9.7,0.3);
\draw[myred,fill=myredfill,ultra thick,solid,rounded corners=2] (12.3,3.7) rectangle (13.7,0.3);
\draw[mygrey,fill=mygreyfill,ultra thick,solid,rounded corners=10] (0.3,3.7) rectangle (1.7,0.3);
\draw[mygrey,fill=mygreyfill,ultra thick,solid,rounded corners=10] (0.3,-0.3) rectangle (1.7,-1.7);

\draw[] (1,-1) node {\large $\rho^{\an}_{0}$};
\draw[] (1,2) node {\large $\rho^{\se}_{0}$};
\draw[] (5,2) node {\large $\mathcal{E}_{1:0}$};
\draw[] (9,2) node[rotate=90] {\large $\mathcal{E}_{n-1:n-2}$};
\draw[] (13,2) node[rotate=90] {\large $\mathcal{E}_{n:n-1}$};
\draw[] (3,0) node { $\mathcal{A}_{0}$};
\draw[] (11,0) node { $\mathcal{A}_{n-1}$};
\draw[] (-0.3,3) node {\Large e};
\draw[] (-0.3,1) node {\Large s};
\draw[] (-0.3,-1) node {\Large a};
\draw[] (15.4,1) node {\large $\rho_{n}$};

\end{tikzpicture}
\caption{A diagrammatic representation of a quantum process, with time flowing from left to right. An initial system-environment state is evolved by a sequence of maps\protect\footnotemark, forming an interlocked pair of higher order quantum maps. These maps are grouped into two categories: those that are under the control of a hypothetical agent, and those which are not. The control operations (in blue) are the actions of this agent. Conversely, red indicates objects which are not under the control of the agent, corresponding to the process tensor. Depending on the situation, the initial state (in grey) can be considered either as a resource available to the agent or as part of the background process.} \label{processtensor}
\end{figure*}
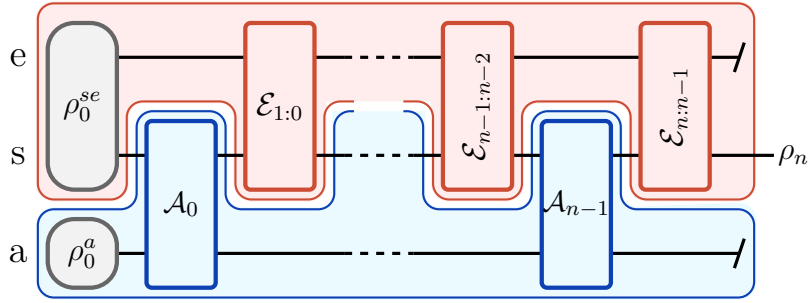
\footnotetext{In figures, where the subsystems involved are otherwise clear, we omit superscript labels from maps and states.}

It is usually convenient to represent the process tensor in the Choi form as a multipartite state instead of a multi-time-step evolution~\cite{introqod}. In this way, one can investigate its properties without being forced to specify an argument. The general form for a Choi state associated with a $n$ step process tensor is
\begin{equation} \label{processchoi}
\Upsilon_{n:0}= \text{tr}_\e \Bigg\{  \functioncomposition\limits_{j=0}^{n-1} \Big(\mathcal{E}_{j+1:j}^{\s \e} \circ S^{\s,{\mathtt{o}_j}}  \Big)   \tensorcomposition\limits_{j=0}^{n-1} \big( \psi^{{\mathtt{o}_j},{\mathtt{i}_j}}  \big) \otimes \rho_0^{{\se} } \Bigg\},
\end{equation}
where $S^{\alpha,\beta}$ is a swap operation between subsystems $\alpha$ and $\beta$, $\rho_0^{{\se} }$ is the system-environment initial state, $\psi$ is a maximally entangled bipartite state, $\mathtt{o}$ and $\mathtt{i}$ index the two halves of the maximally entangled pair by whether they correspond to an output or input of $\mathbf{T}$, and $j$ indexes the step number. In our indexing, $\s={\mathtt{o}_0}$, hence $S^{(\s,{\mathtt{o}_0})}=\mathcal{I}^{\s}$. 

The process tensor formalism is a quantum generalisation of classical stochastic processes~\cite{kolmogorov}. As such it resolves several outstanding conundrums about quantum stochastic processes. For example, it provides an unambiguous necessary and sufficient condition for Markovianity of a quantum process~\cite{nonmarkov, nm2, 1367-2630-18-6-063032}. The corresponding process tensor has the Choi state expressed in the form
\begin{equation}
    \Upsilon^{\text{Markov}}_{n:0}=  \bigotimes_{j=0}^{n-1} \left(\Lambda_{j+1:j} \right) \otimes \rho_0,\label{eq:markovprocess}
\end{equation}
where $\Lambda_{j+1:j}$ is the Choi state of a map corresponding to the $i+1$th leg of the process tensor~\cite{introqod}. More importantly, the process tensor enables the systematic exploration of the rich structure of quantum non-Markovian memory~\cite{Modi2012, Ringbauer2015, Pollock2018T, shrapnel_updating_2017, reconstruct, shrapnel_causation_2018, Costa2018, finitemarkovorder, quantummarkovorder, arXiv:1801.07418, Luchnikov2019, PhysRevLett.122.160401}. The formalism has also led to a pathway to generalise the theory of stochastic thermodynamics to quantum mechanics~\cite{thermoprocess, equilibrationonaverage, almostmarkovian, stochasticthermodynamicswitharbitraryinterventions, thermoprocess, repeatedinteractions}.

There are several other theories that share mathematical structure with the process tensor framework. Firstly, as noted above, the process tensor is a special case of the framework of quantum networks~\cite{supermaps, quantumcircuitarch, quantumnetworks}, which was originally derived as the most general representation of quantum circuit architectures. Beyond this there are \emph{causal automata/non-anticipatory channels}~\cite{kretschmann_quantum_2005, caruso_quantum_2014}, which describe quantum channels with with memory; \emph{causal boxes}~\cite{portmann_causal_2015} that enter into quantum networks with modular elements; \emph{operator tensors}~\cite{hardy_operational_2016, hardy_operator_2012} and \emph{superdensity matrices}~\cite{cotler_superdensity_2017}, employed to investigate quantum information in general relativistic space-time; and, finally, \emph{process matrices}, used for quantum causal modelling~\cite{OreshkovETAL2012, 1367-2630-18-6-063032, oreshkov_causal_2016, Milz2018}; and the $\epsilon$-transducers used within the framework of computational mechanics~\cite{barnett_computational_2015, thompson_using_2017} to describe processes with active interventions. \emph{Quantum strategies}~\cite{quantumstrategies,towardageneraltheoryofquantumgames,onameasureofdistanceforquantumstrategies,fidelityofquantumstrategieswithapplications} can also take on a similar operational structure to the process tensor when co-strategies are considered. Our results could be extended to the frameworks listed above, and any other framework for describing quantum processes as linear functionals. There have also been (much) earlier attempts to define quantum stochastic processes in terms of multi-time correlation kernels~\cite{lindblad1979, Accardi1985}. The process tensor can be shown to be equivalent to these formulations~\cite{lili-hall-wiseman}. However, the process tensor offers structural advantages because it represents processes as quantum states, whereas it is unclear how a resource theory of correlation kernels can be constructed.

The process tensor represents an uncontrolled background process, which like the weather in Fig.~\ref{fig:equations}, can represent a resource. While the agent does not have any control over the process itself, she/he can choose how to interact with it. In the next section, we will work within this structure to show that an agent's repertoire of control operations can be used to derive a resource theory of multi-time processes.

\section{Multi-Time Processes as Resources} \label{sec:ProcessResourceTheories}
Before fully exploring transformations of process tensors, we will present another type transformation that is simpler but carries many of the important features we seek.

\subsection{Preliminary Example: Supermaps and Resource Theories of Quantum Maps} \label{preliminary}

While the majority of existing results on quantum resource theories pertain to states as resources, it is becoming increasingly apparent that dynamical objects such as channels and maps can be harnessed in much the same way to perform useful tasks. Applying resource theories to dynamical objects allows an experimenter to understand how changes to their control capabilities result in differing abilities to perform a particular desired task. Recently, a framework has been developed for resource theories of quantum maps (channels in particular)~\cite{resourcetheoriesofquantumchannels, operationalresourcetheoryofquantumchannels,quantifyingoperationswithanapplication}. Prior to this, there were some specific results that look at channels as resources~\cite{comparisonofquantumchannelsbysuperchannels,entropyofaquantumchannel}, including a resource theory of memory~\cite{theoryofmaps}. Additionally, resource theories of entanglement in bipartite channels~\cite{entanglementofbipartitechannels,resourcetheoryofentanglementforbipartite}, and asymmetric channel distinguishability~\cite{resourcetheoryofasymmetricdistinguishability} have been studied. 

The central object in these theories are the so-called supermaps~\cite{supermaps} $\mathbf{S}$, which enable transformations between resources (maps). Its action on a quantum map $\mathcal{E}^{\s}$, which in turn acts on state $\rho^{\s}$ to produce $\sigma^{\s}$ is
\begin{equation}
\mathcal{E}^{\s}[\rho^{\s}]=\sigma^{\s}, \quad 
\mathbf{S}[\mathcal{E}^{\s}]=\mathcal{E}'^{\s}, 
\quad 
\mathcal{E}'^{\s}[\rho^{\s}]=\sigma'^{\s}.
\end{equation}
As depicted in Fig.~\ref{supermap}, any deterministic\footnote{In this context, deterministic means that channels are mapped to channels, as opposed to arbitrary trace non-increasing CP maps; in other words, the supermaps can be realised unconditionally.} supermap can be represented by
\begin{equation}
\mathbf{S}[\mathcal{E}^{\s}][\rho^{\s}]=\text{tr}_\an\{\mathcal{W}^{\sa} \circ \mathcal{E}^{\s} \circ \mathcal{V}^{\sa}[\rho^{\s} \otimes \rho^{\an}]\} ,
\end{equation}
where $\mathcal{E}^{\s}$ is a map on the main system $\s$ and is the argument of the supermap, while two supplementary maps $\mathcal{V}^{\sa}$ and $\mathcal{W}^{\sa}$ form a (non-unique) representation of the supermap itself, both acting on an additional ancillary subsystem $\an$. The first supplementary map $\mathcal{V}^{\sa}$ acts on $\rho^{\s}$ prior to the application of $\mathcal{E}^{\s}$, and the second $\mathcal{W}^{\sa}$ is applied to the output of $\mathcal{E}^{\s}$. These two maps can be thought of as pre- and post-manipulations in addition to the original map. 

In this setting, free maps can then be defined as those reachable through allowed pre- and post-manipulations from any other map. In this way, the range of experimental control becomes a set of transformations on the object representing the dynamical process, in this case the map $\mathcal{E}^{\s}$. For example, when an agent can only implement trace preserving supermaps (superchannels) which only communicate classical information, channels $\mathcal{E}^{\s}$ which have quantum memory become useful~\cite{theoryofmaps}. We will now derive a more general class of resource theories by taking process tensors to be resources.

\begin{figure}
\centering
\begin{tikzpicture}[scale=1.2]

\draw[mypurplefill,fill=mypurplefill, ultra thick,solid,rounded corners=6]  (6,-1) --(2,-1)  -- (2,1) -- (4,1) -- (4,0) -- (5,0) -- (5,1) --  (6.5,1) -- (6.5,-1) -- (2,-1) -- (6,-1);

\draw[mypurple,thick,solid,rounded corners=6] (6,-1) -- (2,-1)  -- (2,1) -- (4,1) -- (4,0) -- (5,0) -- (5,1) --  (6.5,1) -- (6.5,-1) -- (2,-1) -- (6,-1);

\draw[black,  thick,solid] (5.9,0.5) -- (7,0.5);

\draw[black,  thick,solid] (2.6,-0.5) -- (3.1,-0.5);
\draw[black,  thick,solid] (1.5,0.5) -- (3.1,0.5);
\draw[black,  thick,solid] (3.9,0.5) -- (4.1,0.5);
\draw[black,  thick,solid] (4.9,0.5) -- (5.1,0.5);

\draw[black,  thick,solid] (6.35,-0.35) -- (6.2,-0.65);
\draw[black,  thick,solid] (5.9,-0.5) -- (6.25,-0.5);

\draw[black,  thick,solid] (3.9,-0.5) -- (5.1,-0.5);

\draw[myred,fill=myredfill,very thick,solid,rounded corners=2] (4.1,0.9) rectangle (4.9,0.1);
\draw[mypurple,fill=mypurplefill,very thick,solid,rounded corners=2] (3.1,0.9) rectangle (3.9,-0.9);
\draw[mypurple,fill=mypurplefill,very thick,solid,rounded corners=2] (5.1,0.9) rectangle (5.9,-0.9);
\draw[mygrey,fill=mygreyfill,very thick,solid,rounded corners=2] (2.1,-0.1) rectangle (2.9,-0.9);

\draw[] (1,0.5) node {\Large $s$};
\draw[] (1,-0.5) node {\Large $a$};

\draw[] (4.5,0.5) node[rotate=0] {\Large  $\mathcal{E}$};
\draw[] (3.5,0) node[rotate=0] {\Large  $\mathcal{V}$};
\draw[] (5.5,0) node[rotate=0] {\Large  $\mathcal{W}$};
\draw[] (2.5,-0.5) node[rotate=0] {\Large  $\rho^{\an}$};

\end{tikzpicture}
\caption{A supermap (purple shaded area) acting on a map $\mathcal{E}^{\s}$. A pre- and post-manipulation are performed on the system in surplus to the original map, $\mathcal{V}^{\sa}$ and $\mathcal{W}^{\sa}$ respectively. These contain an ancillary subsystem which may enable the passage of information across $\mathcal{E}^{\s}$.} \label{supermap}
\end{figure}
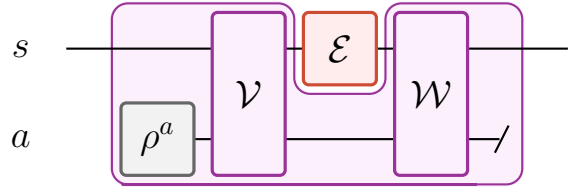

\subsection{Superprocesses} \label{superprocesssection}
If we are interested in the utility of quantum processes as resources, then we need a way of relating them through transformations that encode constraints on how they can be harnessed.
As mentioned in Sec.~\ref{processessection}, applying a control sequence to a system does not immediately appear to constitute a transformation of the process it is undergoing. Instead $\mathbf{A}_{n-1:0}$, as defined in the previous section, takes the form of a linear mapping from process tensors to quantum states. To remedy this situation, we consider the higher order maps $\mathbf{Z}_{n:0}$ that take control sequences to control sequences. Crucially, such higher order maps can also be seen as having a dual action\footnote{We will assume here that all Hilbert spaces are finite dimensional and avoid any of the potential complications with regards to dual spaces that a more general scenario would entail.}, mapping process tensors to other process tensors, fitting the role of transformations in a resource theory of multi-time quantum processes; as such, we introduce a square bra-ket style notation $\llbracket\,\cdot\, | \mathbf{Z} | \,\cdot\, \rrbracket$ to distinguish between ($\llbracket\,\cdot\, | \mathbf{Z}$) action on the process tensor and ($\mathbf{Z} | \,\cdot\, \rrbracket$) action on the control sequence. These higher order maps are a generalisation of the supermaps introduced in the previous subsection, and we will refer to them as \emph{superprocesses}. 

Just as with the supermap, there is a dilated representation of any $\mathbf{Z}_{n:0}$ in terms of reversible maps $\{\mathcal{W}_{j} \}$ and $\{\mathcal{V}_{j}\}$ on an ancillary space, as proven in Theorem 24 of \cite{causalorderingismemory}, under the assumption that there is a definite causal order between the control operations on the input and the output spaces. This allows us to rewrite the process tensor's action on the control sequence as
\begin{align}
\label{sprocdef}
&\llbracket \mathbf{T}_{n:0}|\mathbf{Z}_{n:0}|{\mathbf{A}'}_{n:0} \rrbracket= \mathbf{T}'_{n:0}[{\mathbf{A}'}_{n:0}]=\mathbf{T}_{n:0}[\mathbf{A}_{n-1:0}], 
\end{align}
where ${\mathbf{A}'}_{n:0} = {\rm tr}_\appp\{{\mathcal{A}'}^{\spap}_{n} \circ_\appp \dots \circ_\appp 
{\mathcal{A}'}^{\spap}_0[\rho^\appp_0]\}$ is a control sequence on another system $\s'$ which, in analogy to $\mathbf{A}_{n-1:0}$, can be represented in terms of maps $\mathcal{A}^{\spap}_j$ on $\s'$ and a further ancillary space $\an'$. Expanding out the objects in Eq.~\eqref{sprocdef}, we arrive at
\begin{align} \label{eq:procexpand}
\rho^{{\s}}_n &=\llbracket \notag \mathbf{T}_{n:0}|\mathbf{Z}_{n:0}|{\mathbf{A}'}_{n:0} \rrbracket   \\& = \text{tr}_{ezs'a'} \bigl\{{\mathcal{A}'}^{\spap}_{n} \circ \mathcal{M}^{\sezsp}_{n:n-1}  \circ {\mathcal{A}'}^{\spap}_{n-1} \circ \dots \circ \mathcal{M}^{\sezsp}_{1:0} \circ {\mathcal{A}'}^{\spap}_{0}  \circ \mathcal{W}^{\szsp}_{0} [\rho_0^{\sezspap}] \bigr\},
\end{align}
with 
\begin{equation}
\label{eq:mdef}
    \mathcal{M}^{\sezsp}_{\alpha+1:\alpha} = \mathcal{W}^{\szsp}_{\alpha+1} \circ \mathcal{E}^{\se}_{\alpha+1:\alpha} \circ \mathcal{V}^{\szsp}_{\alpha}.
\end{equation}
In general, ${\mathbf{A}'}_{n:0}$ can include a final measurement operation $\mathcal{A}^{\spap}_n$ that occurs after the final $\se$ map $\mathcal{E}^\se_{n:n-1}$, to allow for the case where the original control sequence cannot be implemented unconditionally~\cite{quantumnetworks}. If the control sequence $\mathbf{A}'_{n:0}$ involves conditioning on a measurement outcome, the final state $\rho^{\s}_n$ can be subnormalised. We can conclude that the diagram representing the interplay between process and control in Fig.~\ref{processtensor} is equivalent to the one in Fig.~\ref{superprocessactingonprocesstensor}. As such, the superprocess $\mathbf{Z}_{n:0}$ is itself a quantum comb with definite causal order alternating between its action on $\s$ and on $\s'$. We detail in Appendix~\ref{appchoi} that the superprocess also has a convenient representation as a many-body quantum state through the Choi isomorphism (analogous to that for the process tensor).

\begin{figure*}[bht]
\centering

\begin{tikzpicture}[scale=0.9]

\draw[myredfill,fill=myredfill, ultra thick,solid,rounded corners=3] (6,1.05) -- (4.95,1.05) -- (4.95,0.05) -- (4.05,0.05) --(4.05,1.05) -- (0.95,1.05)  --  (0.95,0.05) -- (0,0.05) --  (0,2.05) -- (15,2.05) -- (15,1.05) -- (12.95,1.05) -- (12.95,0.05) -- (12.05,0.05) -- (12.05,1.05) --  (8.95,1.05) -- (8.95,0.05) -- (8.05,0.05) -- (8.05,1.05) -- (7,1.05) -- (6,1.05) ;

\draw[myred,thick,solid,rounded corners=3] (6,1.05) -- (4.95,1.05) -- (4.95,0.05) -- (4.05,0.05) --(4.05,1.05) -- (0.95,1.05)  --  (0.95,0.05) -- (0,0.05) --  (0,2.05) -- (15,2.05) -- (15,1.05) -- (12.95,1.05) -- (12.95,0.05) -- (12.05,0.05) -- (12.05,1.05) --  (8.95,1.05) -- (8.95,0.05) -- (8.05,0.05) -- (8.05,1.05) -- (7,1.05)  ;

\draw[mypurplefill,fill=mypurplefill, very thick,solid,rounded corners=3] (6,0.95)  -- (5.05,0.95) -- (5.05,-0.05) -- (3.95,-0.05) -- (3.95,0.95) -- (3.05,0.95) -- (1.05,0.95) -- (1.05,-0.05)  -- (0,-0.05) -- (0,-1.95) -- (1.95,-1.95) -- (1.95,-0.95) -- (3.05,-0.95) -- (3.05,-1.95) -- (6,-1.95) -- (7,-1.95)--  (9.95,-1.95) -- (9.95,-0.95) -- (11.05,-0.95) -- (11.05,-1.95) -- (13.95,-1.95) -- (13.95,-0.95) -- (15,-0.95) -- (15,0.95) -- (13.05,0.95) -- (13.05,-0.05) -- (11.95,-0
.05) -- (11.95,0.95) -- (11.05,0.95) --  (9.05,0.95) -- (9.05,-0.05) -- (7.95,-0.05) -- (7.95,0.95) -- (7,0.95)  ;

\draw[mypurple, thick,solid,rounded corners=3] (6,0.95)  -- (5.05,0.95) -- (5.05,-0.05) -- (3.95,-0.05) -- (3.95,0.95) -- (3.05,0.95) -- (1.05,0.95) -- (1.05,-0.05)  -- (0,-0.05) -- (0,-1.95) -- (1.95,-1.95) -- (1.95,-0.95) -- (3.05,-0.95) -- (3.05,-1.95) -- (6,-1.95) ;

 \draw[mypurple, thick,solid,rounded corners=3] (7,-1.95)--  (9.95,-1.95) -- (9.95,-0.95) -- (11.05,-0.95) -- (11.05,-1.95) -- (13.95,-1.95) -- (13.95,-0.95) -- (15,-0.95) -- (15,0.95) -- (13.05,0.95) -- (13.05,-0.05) -- (11.95,-0.05) -- (11.95,0.95) -- (11.05,0.95) --  (9.05,0.95) -- (9.05,-0.05) -- (7.95,-0.05) -- (7.95,0.95) -- (7,0.95)  ;

\draw[mybluefill,fill=mybluefill, very thick,solid,rounded corners=3]  (6,-2.05) -- (5.05,-2.05) -- (2.95,-2.05) -- (2.95,-1.05) --  (2.05,-1.05) -- (2.05,-2.05) -- (0,-2.05) -- (0,-3) --  (2.05,-3) -- (15,-3) -- (15,-1.05)  -- (14.05,-1.05) -- (14.05,-2.05) (10.95,-2.05)  -- (10.95,-1.05) --  (10.05,-1.05) -- (10.05,-2.05) -- (7,-2.05) -- (6,-2.05) ;

\draw[myblue, thick,solid,rounded corners=3] (6,-2.05) -- (5.05,-2.05) -- (2.95,-2.05) -- (2.95,-1.05)  --  (2.05,-1.05) -- (2.05,-2.05) -- (0,-2.05) -- (0,-3)  --  (2.05,-3) -- (15,-3) -- (15,-1.05)  -- (14.05,-1.05) -- (14.05,-2.05) -- (10.95,-2.05)  -- (10.95,-1.05) --  (10.05,-1.05) -- (10.05,-2.05) -- (7,-2.05)  ;

\draw[black,  thick,solid] (12.9,1.5) -- (14.5,1.5);
\draw[black,  thick,solid] (14.4,-1.7) -- (14.6,-1.3);
\draw[black,  thick,solid] (13.9,0.5) -- (15.3,0.5);
\draw[black,  thick,solid] (0.9,1.5) -- (4.1,1.5);
\draw[black,  thick,solid] (8.9,1.5) -- (12.1,1.5);
\draw[black,  thick,solid] (4.9,1.5) -- (5.9,1.5);
\draw[black,  thick,dashed] (5.9,1.5) -- (7.1,1.5);
\draw[black,  thick,solid] (7.1,1.5) -- (8.1,1.5);

\draw[black,  thick,dashed] (5.9,0.5) -- (7.1,0.5);
\draw[black,  thick,solid] (0.9,0.5) -- (2.1,0.5);
\draw[black,  thick,solid] (2.9,0.5) -- (3.1,0.5);
\draw[black,  thick,solid] (3.9,0.5) -- (4.1,0.5);
\draw[black,  thick,solid] (4.9,0.5) -- (5.1,0.5);
\draw[black,  thick,solid] (7.9,0.5) -- (8.1,0.5);
\draw[black,  thick,solid] (8.9,0.5) -- (9.1,0.5);
\draw[black,  thick,solid] (9.9,0.5) -- (10.1,0.5);
\draw[black,  thick,solid] (10.9,0.5) -- (11.1,0.5);
\draw[black,  thick,solid] (11.9,0.5) -- (12.1,0.5);
\draw[black,  thick,solid] (12.9,0.5) -- (13.1,0.5);

\draw[black,  thick,dashed] (5.9,-0.5) -- (7.1,-0.5);
\draw[black,  thick,dashed] (5.9,-1.5) -- (7.1,-1.5);
\draw[black,  thick,dashed] (5.9,-2.5) -- (7.1,-2.5);
\draw[black,  thick,solid] (2,0.5) -- (5.9,0.5);
\draw[black,  thick,solid] (7.1,0.5) -- (14.5,0.5);

\draw[black,  thick,solid] (0.8,-0.5) -- (5.9,-0.5);
\draw[black,  thick,solid] (7.1,-0.5) -- (14.5,-0.5);
\draw[black,  thick,solid] (0.8,-1.5) -- (5.9,-1.5);
\draw[black,  thick,solid] (0.8,-2.5) -- (5.9,-2.5);
\draw[black,  thick,solid] (7.1,-1.5) -- (14.5,-1.5);
\draw[black,  thick,solid] (7.1,-2.5) -- (14.5,-2.5);
\draw[black,  thick,solid] (14.4,1.3) -- (14.6,1.7);
\draw[black,  thick,solid] (14.4,-1.7) -- (14.6,-1.3);
\draw[black,  thick,solid] (14.4,-2.7) -- (14.6,-2.3);
\draw[black,  thick,solid] (14.4,-0.7) -- (14.6,-0.3);

\draw[myblue,fill=mybluefill,very thick,solid,rounded corners=2] (14.1,-1.1) rectangle (14.9,-2.9);
\draw[myblue,fill=mybluefill,very thick,solid,rounded corners=2] (2.1,-1.1) rectangle (2.9,-2.9);
\draw[myblue,fill=mybluefill,very thick,solid,rounded corners=2] (10.1,-1.1) rectangle (10.9,-2.9);
\draw[myred,fill=myredfill,very thick,solid,rounded corners=2] (4.1,1.9) rectangle (4.9,0.1);
\draw[myred,fill=myredfill,very thick,solid,rounded corners=2] (8.1,1.9) rectangle (8.9,0.1);
\draw[myred,fill=myredfill,very thick,solid,rounded corners=2] (12.1,1.9) rectangle (12.9,0.1);
\draw[mygrey,fill=mygreyfill,very thick,solid,rounded corners=10] (0.1,1.9) rectangle (0.9,0.1);

\draw[mypurple,fill=mypurplefill,very thick,solid,rounded corners=2] (1.1,0.9) rectangle (1.9,-1.9);
\draw[mypurple,fill=mypurplefill,very thick,solid,rounded corners=2] (3.1,0.9) rectangle (3.9,-1.9);
\draw[mypurple,fill=mypurplefill,very thick,solid,rounded corners=2] (5.1,0.9) rectangle (5.9,-1.9);
\draw[mypurple,fill=mypurplefill,very thick,solid,rounded corners=2] (7.1,0.9) rectangle (7.9,-1.9);
\draw[mypurple,fill=mypurplefill,very thick,solid,rounded corners=2] (9.1,0.9) rectangle (9.9,-1.9);
\draw[mypurple,fill=mypurplefill,very thick,solid,rounded corners=2] (11.1,0.9) rectangle (11.9,-1.9);
\draw[mypurple,fill=mypurplefill,very thick,solid,rounded corners=2] (13.1,0.9) rectangle (13.9,-1.9);
\draw[mygrey,fill=mygreyfill,very thick,solid,rounded corners=10] (0.1,-0.1) rectangle (0.9,-1.9);
\draw[mygrey,fill=mygreyfill,very thick,solid,rounded corners=10] (0.1,-2.1) rectangle (0.9,-2.9);

\draw[] (0.5,1) node {\small $\rho_{0}^{\se}$};
\draw[] (4.5,1) node[rotate=90] {\small  $\mathcal{E}_{1:0}$};
\draw[] (8.5,1) node[rotate=90] {\small  $\mathcal{E}_{n-1:n-2}$};
\draw[] (12.5,1) node[rotate=90] {\small  $\mathcal{E}_{n:n-1}$};
\draw[] (2.5,-2) node[rotate=90] {\small  ${\mathcal{A}'}_{0}$};
\draw[] (10.5,-2) node[rotate=90] {\small  ${\mathcal{A}'}_{n-1}$};
\draw[] (14.5,-2) node[rotate=90] {\small  ${\mathcal{A}'}_{n}$};
\draw[] (15.6,0.5) node {\small  $\rho^{\s}_{n}$};

\draw[] (-0.3,1.5) node {$e$};
\draw[] (-0.3,0.5) node {$s$};
\draw[] (-0.3,-0.5) node {$z$};
\draw[] (-0.3,-1.5) node {$s'$};
\draw[] (-0.3,-2.5) node {$a'$};

\draw[] (0.5,-1) node {\small $\rho_{0}^{\zsp}$};
\draw[] (0.5,-2.5) node {\small $\rho_{0}^{\appp}$};
\draw[] (1.5,-0.5) node[rotate=90] {\small  $\mathcal{W}_{0}$};
\draw[] (3.5,-0.5) node[rotate=90] {\small  $\mathcal{V}_{0}$};
\draw[] (5.5,-0.5) node[rotate=90] {\small  $\mathcal{W}_{1}$};
\draw[] (7.5,-0.5
) node[rotate=90] {\small  $\mathcal{V}_{n-2}$};
\draw[] (9.5,-0.5) node[rotate=90] {\small  $\mathcal{W}_{n-1}$};
\draw[] (11.5,-0.5) node[rotate=90] {\small  $\mathcal{V}_{n-1}$};
\draw[] (13.5,-0.5) node[rotate=90] {\small  $\mathcal{W}_{n}$};

\end{tikzpicture}

\caption{The fully general control sequence-superprocess-process tensor structure of a quantum process. The dynamics can be separated into three distinct object groups: the process tensor (red), the superprocess (purple), and the reduced control sequence (blue). The process tensor is inaccessible, while the superprocess and control sequence are in principle accessible to agents. Grouping the purple with blue induces a different control sequence acting on the original background process, while grouping the purple with red corresponds to a new effective background process acted on by the control sequence in blue. If the control sequence $\mathbf{A}'_{n:0}$ can only be implemented conditionally on a measurement outcome, the final state $\rho^{\s}_n$ may be subnormalised. } \label{superprocessactingonprocesstensor}
\end{figure*}
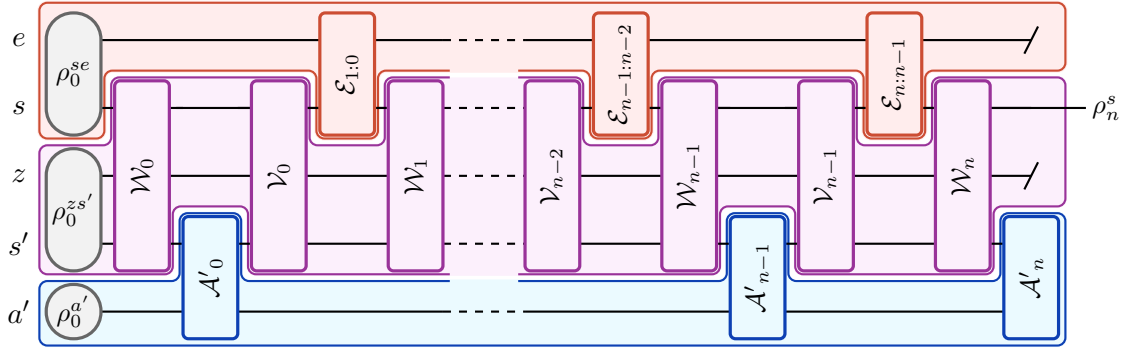

We can construct a meaningful resource theory of multi-time quantum processes by placing restrictions on an agent's control: The control sequences which the agent is capable of performing correspond to a set of superprocesses that relate them to a fiducial set of control sequences on $\s'$. These superprocesses can be regarded as free, and in turn through their dual action define the set of process tensors that can be obtained for free. Furthermore, when the process tensor is specified to not contain an initial state and only have one time-step, and the control operations are taken to be trivial, this expression reduces to the action of a supermap on a single quantum map. Resource theories of maps as in Sec.~\ref{preliminary} are a specific case of these more general theories.

An alternative picture of the superprocess is one where multiple agents with varying levels of control capabilities interact. The superprocess then represents the actions of an intermediary contracted by a client, who can only perform a limited set of operations on the $\s'$ system. This contractor interfaces directly with the process, which plays the role of a remote server, and returns the result to the client. From the client's perspective, the joint system of contractor (superprocess) and server (underlying process) can be viewed as a new process (combining red and purple in Fig.~\ref{superprocessactingonprocesstensor}): any control sequence the client  could directly apply to the underlying process, they could equally apply to the transformed process. Equally, from the perspective of the underlying process, the combined control sequence and superprocess, can be viewed as a new control sequence (combining blue and purple in Fig.~\ref{superprocessactingonprocesstensor}). 

The latter picture could be straightforwardly generalised to scenarios where the system accessed by the client is very different from that which feeds directly into the server. It could also be extended to cases where the contractor's actions are represented by superprocesses that take $n$ step processes to $m$ step ones. When $n<m$ the number of client actions can be increased by including more than one $\mathcal{V}$-$\mathcal{W}$ pair per step, and conversely when $n>m$ the main subsystem $\s$ can be joined between $\mathcal{V}$ at one step and $\mathcal{W}$ at the next step.

\subsection{Resource Theories of Multi-Time Processes} \label{sec:RTdef}
We now have all the necessary ingredients to formally define resource theories of multi-time processes: 
\begin{definition}[Resource Theory of Multi-Time Processes]
A resource theory of multi-time processes $\mathsf{R}=\{ \mathsf{T},\mathsf{Z} \}$ consists of two sets. The first set $\mathsf{T}$ consists of uncontrolled background processes which an agent might be subject to, represented by process tensors $\mathbf{T}_{n:0}$, while the second set $\mathsf{Z}$ is that of superprocesses $\mathbf{Z}_{n:0}$ which the agent is capable of implementing as transformations of background processes.
\end{definition}

This induces a structure of convertibility of processes. Free processes are those which can be obtained via application of allowed superprocesses, given any other process. While any process can reach every free process, non-free processes may be able to reach other processes too. Hence, when possible, it is sufficient for generating the full set of free processes to find a particular free transformation which links every resource to a particular free resource, and to reach all other free resources by subsequently applying other free transformations.

Many important quantum resource theories have these types of transformations which link any resource to a particular free one; for example, this is what thermalising a non-equilibrium state does. Whether such a transformation exists will depend on the restrictions placed on the set of free superprocesses. Generically, these should look like an agent applying a suitably noisy superprocess to the background process, such that the transformed background process loses its ability to carry information.

Defined in such a way, the set of free processes will be a closed set, and satisfy the `golden rule of quantum resource theories' presented in Ref.~\cite{review}.

\subsection{Monotones} \label{monotones}
We turn our attention to resolving one of the main questions underpinning this investigation: how does one quantify the utility of an uncontrolled background process? A class of monotones for multi-time process theories can be derived from the Choi representation of the process tensor. 
\begin{theorem} \label{thm:monotones}
Given a resource theory of multi-time processes, any state distance $D(\cdot,\cdot)$ measure applied to the Choi states of process tensors satisfying contractivity under the action of superprocesses forms a monotone
\begin{equation}
M = \inf_{\Upsilon_{\mathbf{T}_\text{F}}} D\big(\Upsilon_\mathbf{T},\Upsilon_{\mathbf{T}_\text{F}}\big),
\end{equation}
where $\Upsilon_{\mathbf{T}_\text{F}}$ are free process tensors in the Choi representation. 
\end{theorem}
\begin{proof}
Since $D(\cdot,\cdot)$ is contractive under application of the superprocess:
\begin{equation}
 D(\Upsilon_\mathbf{T},\Upsilon_{\mathbf{T}'})  \geq D(\Upsilon_{\llbracket \mathbf{T}|\mathbf{Z}},\Upsilon_{\llbracket \mathbf{T}'|\mathbf{Z}}).
\end{equation}
If we set $\Upsilon_{\mathbf{T}'}$ to be the closest free process to $\Upsilon_\mathbf{T}$ this expression becomes
\begin{equation}
 D(\Upsilon_\mathbf{T},\Upsilon_{\mathbf{T}_\text{F}})  \geq D(\Upsilon_{ \llbracket \mathbf{T}|\mathbf{Z}},\Upsilon_{\llbracket \mathbf{T}_\text{F}|\mathbf{Z}}).
\end{equation}
Since $\mathbf{Z}$ is a free superprocess, then $\Upsilon_{\llbracket \mathbf{T}_\text{F}|\mathbf{Z}}$ must still be a free process; however, $\Upsilon_{\llbracket \mathbf{T}_\text{F}|\mathbf{Z}}$ need not be the closest free  process to $\Upsilon_{\llbracket \mathbf{T}|\mathbf{Z}}$. Hence, there exists another free process $\Upsilon_{\mathbf{T}'_\text{F}}$ at least as close to $\Upsilon_{\llbracket \mathbf{T}|\mathbf{Z}}$
\begin{equation}
D(\Upsilon_{\llbracket \mathbf{T}|\mathbf{Z}},\Upsilon_{\llbracket \mathbf{T}_\text{F}|\mathbf{Z}}) \geq \inf_{\Upsilon_{\mathbf{T}'_\text{F}}}  D(\Upsilon_{\llbracket \mathbf{T}|\mathbf{Z}},\Upsilon_{\mathbf{T}'_\text{F}}),
\end{equation}
which gives us the final result:
\begin{equation}
M(\Upsilon_\mathbf{T}) \geq M(\Upsilon_{\llbracket \mathbf{T}|\mathbf{Z}}),
\end{equation}
for free superprocesses $\mathbf{Z}$, proving that $M$ is a monotone.
\end{proof}

The main difficulty associated with this family of monotones is specifying a $D(\cdot,\cdot)$ which is contractive under the action of superprocesses. Considering process tensors as states in the Choi representation, a superprocess is no more than a specific kind of CPTP map between states. As such, state distance measures which are contractive for all CPTP maps will automatically satisfy our needs. The trace distance and relative entropy distance on Choi states are two suitable choices in this regard~\cite{beyondstrong}. A quantum strategy approach~\cite{resourcetheoryofasymmetricdistinguishability} has yielded a similar class of multi-time measures called generalised quantum strategy divergences. These coincides with our monotones when the optimisation is restricted to control sequences whose deterministic action has a Choi state proportional to the identity operator.

Another notion of quantifying resources is \emph{robustness}~\cite{review}, which measures how much a resource can be mixed with others while remaining non-free. Robustness measures provide an avenue for linking resource theories to operational tasks such as channel discrimination~\cite{PhysRevLett.122.140402,PhysRevLett.122.140403}.
\begin{theorem} \label{thm:monotone2}
Given a resource theory of multi-time processes, the global robustness $R$ of a process $\mathbf{T} \in \mathsf{T}$ using a state definition of robustness applied to the Choi state $\Upsilon_{\mathbf{T}}$,
\begin{equation} \label{eq:robustness}
    R(\mathbf{T}) = \inf_{r \in \mathbb{R}_{\geq 0}} \left\{ r  \ : \ \exists \Upsilon_{\mathbf{\Gamma}} \in \mathsf{T} \ \text{s.t.} \ \frac{\Upsilon_{\mathbf{T}} + r \Upsilon_{\mathbf{\Gamma}}}{1+r} \in \mathsf{T}_\text{F} \right\}
\end{equation}
is a monotone. $\mathsf{T}_\text{F}$ is the set of free processes in the resource theory. For global robustness, $\mathbf{\Gamma}$ is any process in $\mathsf{T}$, but the absolute robustness can equally be defined by letting $\mathbf{\Gamma} \in \mathsf{T}_\text{F}$.
\end{theorem}
\begin{proof}
The definition of robustness for process tensor $\mathbf{T}$ means that there exists a $\Upsilon_{\mathbf{\Gamma}}$ with $r=R(\mathbf{T})$ and 
\begin{equation}
    \frac{\Upsilon_{ \mathbf{T} } + r \Upsilon_{\mathbf{\Gamma}}}{1+r} := \Upsilon_{\mathbf{\Xi}} \in \mathsf{T}_\text{F} .
\end{equation}
Robustness of a process tensor after application of a free superprocess is
\begin{equation}
    R(\llbracket \mathbf{T} | \mathbf{Z}) = \inf_{p \in \mathbb{R}_{\geq 0}} \left\{ p  \ : \ \exists \Upsilon_{\mathbf{\Omega}} \in \mathsf{T} \ \text{s.t.} \ \frac{\Upsilon_{\llbracket \mathbf{T} | \mathbf{Z}} + p \Upsilon_{\mathbf{\Omega}}}{1+p} \in \mathsf{T}_\text{F} \right\}.
\end{equation}
Let $\Upsilon_{\mathbf{\Omega}}=\Upsilon_{\llbracket \mathbf{\Gamma} | \mathbf{Z} }$ and $p=r$; the linearity of superprocesses implies that
\begin{equation}
     \frac{\Upsilon_{\llbracket \mathbf{T} | \mathbf{Z}} + r \Upsilon_{\llbracket \mathbf{\Gamma} | \mathbf{Z} }}{1+r} =\Upsilon_{\llbracket \mathbf{\Xi} | \mathbf{Z}} \in \mathsf{T}_\text{F}.
\end{equation}
Hence, the robustness of a process tensor after applying a free superprocess is not greater than beforehand
\begin{equation}
    R( \Upsilon_{\llbracket \mathbf{T} | \mathbf{Z}}) \leq  R(\Upsilon_{ \mathbf{T}}).
\end{equation}
\end{proof}

The key difference between this definition and that for the robustness of general quantum states is that -- due to the causality trace conditions a process tensor must satisfy, $\mathsf{T}$ is a smaller set than that of states~\cite{quantumnetworks}.
A quantum state $\Upsilon_{n:0} \in \mathcal{B}(\mathcal{H}_{0^{\mathtt{i}}} \otimes \mathcal{H}_{1^{\mathtt{o}}} \dots \mathcal{H}_{{n-1}^{\mathtt{i}}} \otimes \mathcal{H}_{n^{\mathtt{o}}} )$, $\Upsilon_{n:0}$ is a deterministic quantum network -- and hence the Choi state of a valid process tensor -- iff it satisfies the hierarchy of trace conditions 
\begin{equation} \label{eq:tracecondition}
\text{tr}_{j^{\mathtt{o}}} \left\{ \Upsilon_{j:0}  \right\} = \mathbbm{1}_{{j-1}^{\mathtt{i}}} \otimes \Upsilon_{j-1:0}  , \ \ \forall j \leq n.
\end{equation}
One might assume that since there is a smaller range of valid process tensors than general states to mix with $\Upsilon_{\mathbf{T}}$, a larger $r$ might be required to find a suitable $\Upsilon_{\mathbf{\Gamma}}$ in Eq.~\eqref{eq:robustness}. However, we show that this is not true: when the result must be a valid process tensor (which free processes always are), $\Upsilon_{\mathbf{\Gamma}}$ must be as well. 

\begin{lemma} \label{lem:closuretracecondition}
Given the Choi state of a process tensor $\Upsilon_{n:0} \in \mathcal{B}(\mathcal{H}_{1^{\mathtt{i}}} \otimes \mathcal{H}_{1^{\mathtt{o}}} \dots \mathcal{H}_{n^{\mathtt{i}}} \otimes \mathcal{H}_{n^{\mathtt{o}}} )$, if a convex mixture of $\Upsilon_{n:0}$ with an arbitrary quantum state $\Delta_{n:0}$ 
\begin{equation}
    \frac{\Upsilon_{n:0} + r \Delta_{n:0}}{1+r} = \Theta_{n:0}
\end{equation}
is a valid process tensor, then $\Delta_{n:0}$ satisfies Eq.~\eqref{eq:tracecondition}, implying that $\Delta_{n:0}$ is also a valid process tensor.
\end{lemma}
\begin{proof}
We consider the convex mixture 
\begin{equation}
    \frac{\Upsilon_{n:0} + r \Delta_{n:0}}{1+r} = \Theta_{n:0} \in \mathsf{T}.
\end{equation}
Applying the first trace condition to $\Theta_{n:0}$ yields
\begin{equation}
    \text{tr}_{n^{\mathtt{o}}} \left\{ \Theta_{n:0}  \right\} = \frac{1}{1+r} \left(\mathbbm{1}_{{n-1}^{\mathtt{i}}} \otimes \Upsilon_{n-1:0}  +r  \text{tr}_{n^{\mathtt{o}}} \left\{ \Delta_{n:0}  \right\} \right) = \mathbbm{1}_{{n-1}^{\mathtt{i}}} \otimes \Theta_{n-1:0}.
\end{equation}
Rearranging to collect the terms which satisfy Eq.~\eqref{eq:tracecondition} reveals
\begin{equation}
         \text{tr}_{n^{\mathtt{o}}} \left\{ \Delta_{n:0}  \right\}  = \mathbbm{1}_{{n-1}^{\mathtt{i}}} \otimes \frac{ (1+r)  \  \Theta_{n-1:0} -  \Upsilon_{n-1:0}}{r},
\end{equation}
which also satisfies Eq.~\eqref{eq:tracecondition}.

This argument can be repeated for all $j \leq n$, fulfilling the hierarchy, and implying that $\Delta_{n:0}$ is a valid process tensor.
\end{proof}
This lemma ensures that the value of the robustness of process tensors is always equal to the robustness of their corresponding Choi states.

Additionally, the robustness of a process tensor can be linked to a member of our $M$ family of monotones -- the max-relative entropy to the set of free processes -- via a similar argument to that previously employed for entanglement~\cite{maxrelativeentropyofentanglment}, combined with Lem~\ref{lem:closuretracecondition}. 
\begin{theorem} \label{thm:maxrelativeentropy}
The max-relative entropy of a process tensor $\mathbf{T}$ to the set of free processes equals its global log robustness
\begin{equation}
    \log\left(1+ R(\mathbf{T})\right) = \inf_{\Upsilon_{\mathbf{\Xi}} \in \mathsf{T}_\text{F}} D_\text{max} (\Upsilon_{\mathbf{T}}||\Upsilon_{\mathbf{\Xi}}).
\end{equation}
\end{theorem}
\begin{proof}
The membership of $\frac{\Upsilon_{\mathbf{T}} + r \Upsilon_{\mathbf{\Gamma}}}{1+p} := \Upsilon_{\mathbf{\Xi}} \in \mathsf{T}_\text{F} $ implies that $R(\mathbf{T})=r$ defines an upper bound on $\Upsilon_{\mathbf{T}}$
\begin{equation} \label{eq:robustnessupperbound}
    R(\mathbf{T}) = \inf_{r \in \mathbb{R}_{\geq 0}} \left\{ r  \ : \ \exists \Upsilon_{\mathbf{\Xi}} \in \mathsf{T}_\text{F} \ \text{s.t.} \ \Upsilon_{\mathbf{T}}  \leq (1+r) \Upsilon_{\mathbf{\Xi}} \ , \ \Upsilon_{\mathbf{\Xi}} \in \mathsf{T}_\text{F} \right\}.
\end{equation}
In Eq.~\eqref{eq:robustnessupperbound}, Lem.~\ref{lem:closuretracecondition} guarantees that there is no case where $R(\mathbf{T})>r$.

Therefore the global robustness of a process tensor can be related to to max-relative entropy,
\begin{equation}
\begin{aligned}
   \log\left(1+ R(\mathbf{T})\right) & = \inf_{r \in \mathbb{R}_{\geq 0}} \left\{ \log(1+r)  \ : \ \exists \Upsilon_{\mathbf{\Xi}} \in \mathsf{T}_\text{F} \ \text{s.t.} \ \Upsilon_{\mathbf{T}}  \leq (1+r) \Upsilon_{\mathbf{\Xi}}  \right\} \\
   & = \inf_{\Upsilon_{\mathbf{\Xi}} \in \mathsf{T}_\text{F}} \log \left(  \inf_{r \in \mathbb{R}_{\geq 0}} \left\{ 1+r  \ : \ \text{s.t.} \ \Upsilon_{\mathbf{T}}  \leq (1+r) \Upsilon_{\mathbf{\Xi}}  \right\}\right) \\
& = \inf_{\Upsilon_{\mathbf{\Xi}} \in \mathsf{T}_\text{F}} D_\text{max} (\Upsilon_{\mathbf{T}}||\Upsilon_{\mathbf{\Xi}}).
   \end{aligned}
\end{equation}
\end{proof}

Observe that the max-relative entropy to the nearest free process satisfies the contractivity condition of Thm.~\ref{thm:monotones}, and hence is an example of this $M$ type of monotone. Additionally, max-relative entropy has been shown to be predictive of success in resource dilution for entanglement~\cite{maxrelativeentropyofentanglment}. Potentially analogous results might hold for non-Markovianity.

There are many more monotones which could be investigated for resource theories of quantum processes. For channel resource theories, a common monotone involves optimising some quantity over all possible input states~\cite{approachesfortheapproximateaddativityoftheholevo,operationalresourcetheoryofquantumchannels, resourcetheoriesofquantumchannels}. The analogous family of monotones for process theories would correspond to an optimisation over control operations instead of initial states. We expect that numerous operationally important monotones would arise from different choices for the quantity which is optimised. A particularly promising avenue is to optimise the information retained about past states of the system. Such a monotone would correspond to how well an agent can preserve quantum information, given a background process with potentially non-Markovian noise.

Now that we have formally defined resource theories of multi-time processes, we are able to apply them to specific classes of operational scenarios.  For the remainder of this paper, we will consider a general set of restrictions on communication in time, and construct the corresponding family of resource theories, the monotones of which (as proposed in Thm.~\ref{thm:monotones} and Thm.~\ref{thm:monotone2}) typically are various forms of non-Markovian memory.

\section{Resource Theories of Processes with Restricted Communication}
\label{communicationtheories}
One of the most natural constraints on the structure of the superprocess is the connectivity of the ancillary subsystem, which may be restricted by constraining the information that flows between elementary maps $\mathcal{V}$ and $\mathcal{W}$ in Fig.~\ref{superprocessactingonprocesstensor}. Restricted communication within the superprocess can be manifested in numerous kinds of operational scenarios. We make use of the client-contractor-server metaphor, first introduced in Sec.~\ref{superprocesssection}, to emphasise the  important features of each scenario. Here, a \textit{client}, interacting with a \textit{server}, can only perform control operations on the system without access to a memory bearing subsystem, classical or quantum. However, the client can enlist a \textit{contractor} to perform tasks requiring memory. As such, the contractor, who facilitates communication between the client and the server, is described by a superprocess.

In this setting, the contractor can only be useful to the client when their actions are more powerful than those of the client. Hence, while there is some freedom in choosing the resource theory of the client, their allowed operations must always be a subset of those permitted for the contractor. This can always be guaranteed if one chooses to take the actions of the client as local in time and uncorrelated, which is what we specify here. The contractor can then use their operations to interact with a \textit{server}, corresponding to the process tensor. The set of free transformations on process tensors corresponds to the range of services that the contractor can provide the client. For instance, the contractor may charge a premium price for a service that involves quantum memory, and a lower price for only classical memory.

\subsection{Communication Maps} \label{commmaps}

The communication restrictions on the superprocess can be set by replacing the identity channels between $\mathcal{V}$ and $\mathcal{W}$ in Fig.~\ref{superprocessactingonprocesstensor} by channels $\mathcal{C}$ and $\mathcal{K}$, as depicted in Fig.~\ref{superprocesswithcommunication}. The particular forms of these \textit{communication maps} dictates what kind of information is able to propagate through the ancillary subsystems, i.e., systematically restrict or allow information transfer between different points in time. In this setting, the client solely acts on subsystem $\s'$ with ${\mathcal{A}'}_{\alpha}^{\s'}$. The communication maps $\mathcal{C}^{\zsp}$ and $\mathcal{K}^{\zs}$ (as well as $\mathcal{W}^{\szsp}$ and $\mathcal{V}^{\szsp}$) make up the constrained superprocess. Hence, responsibility of communication and performing joint system-ancilla operations are delegated to the superprocess (contractor). Specifying this type of resource theory for quantum processes reduces to specifying the class of allowed $\mathcal{K}^{\zs}_{\alpha}$ and $\mathcal{C}^{\zsp}_{\alpha+1:\alpha}$.

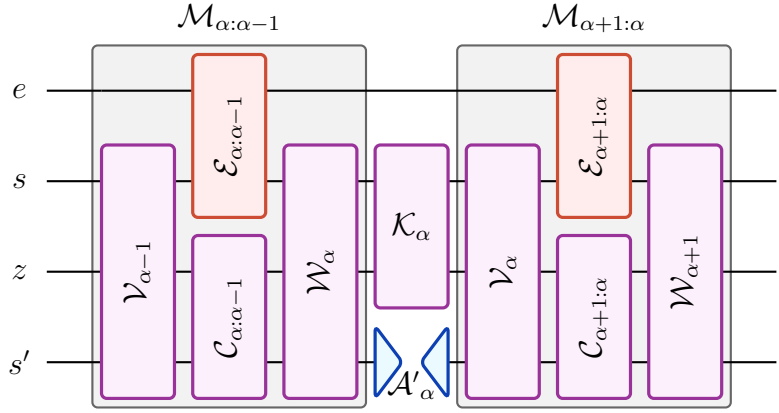
\begin{figure*} 
\centering
\begin{tikzpicture}[scale=1.2]

\draw[mygrey,fill=mygreyfill,thick,solid,rounded corners=2] (7,-2) rectangle (10,2);

\draw[mygrey,fill=mygreyfill,thick,solid,rounded corners=2] (11,-2) rectangle (14,2);

\draw[black,  thick,solid] (12.9,1.5) -- (14.5,1.5);
\draw[black,  thick,solid] (13.9,0.5) -- (14.5,0.5);
\draw[black,  thick,solid] (8.9,1.5) -- (12.1,1.5);

\draw[black,  thick,solid] (6.5,1.5) -- (7.1,1.5);
\draw[black,  thick,solid] (7.1,1.5) -- (8.1,1.5);

\draw[black,  thick,solid] (6.5,0.5) -- (7.1,0.5);
\draw[black,  thick,solid] (7.9,0.5) -- (8.1,0.5);
\draw[black,  thick,solid] (8.9,0.5) -- (9.1,0.5);
\draw[black,  thick,solid] (9.9,0.5) -- (10.1,0.5);
\draw[black,  thick,solid] (10.9,0.5) -- (11.1,0.5);
\draw[black,  thick,solid] (11.9,0.5) -- (12.1,0.5);
\draw[black,  thick,solid] (12.9,0.5) -- (13.1,0.5);

\draw[black,  thick,solid] (6.5,-0.5) -- (7.1,-0.5);
\draw[black,  thick,solid] (7.9,-0.5) -- (9.1,-0.5);
\draw[black,  thick,solid] (9.9,-0.5) -- (11.1,-0.5);
\draw[black,  thick,solid] (11.9,-0.5) -- (13.1,-0.5);
\draw[black,  thick,solid] (13.9,-0.5) -- (14.5,-0.5);

\draw[black,  thick,solid] (10.7,-1.5) -- (14.5,-1.5);
\draw[black,  thick,solid] (6.5,-1.5) -- (10.3,-1.5);

\draw[mypurple,fill=mypurplefill,very thick,solid,rounded corners=2] (8.1,-1.9) rectangle (8.9,-0.1);
\draw[mypurple,fill=mypurplefill,very thick,solid,rounded corners=2] (12.1,-1.9) rectangle (12.9,-0.1);
\draw[mypurple,fill=mypurplefill,very thick,solid,rounded corners=2] (10.1,-0.9) rectangle (10.9,0.9);
\draw[myblue,fill=mybluefill,very thick,solid,rounded corners=2] (10.1,-1.9) -- (10.1,-1.1) -- (10.4,-1.5) -- cycle;
\draw[myblue,fill=mybluefill,very thick,solid,rounded corners=2] (10.9,-1.9) -- (10.9,-1.1) -- (10.6,-1.5) -- cycle;
\draw[myred,fill=myredfill,very thick,solid,rounded corners=2] (8.1,1.9) rectangle (8.9,0.1);
\draw[myred,fill=myredfill,very thick,solid,rounded corners=2] (12.1,1.9) rectangle (12.9,0.1);

\draw[mypurple,fill=mypurplefill,very thick,solid,rounded corners=2] (7.1,0.9) rectangle (7.9,-1.9);
\draw[mypurple,fill=mypurplefill,very thick,solid,rounded corners=2] (9.1,0.9) rectangle (9.9,-1.9);
\draw[mypurple,fill=mypurplefill,very thick,solid,rounded corners=2] (11.1,0.9) rectangle (11.9,-1.9);
\draw[mypurple,fill=mypurplefill,very thick,solid,rounded corners=2] (13.1,0.9) rectangle (13.9,-1.9);

\draw[] (8.5,1) node[rotate=90] {\large  $\mathcal{E}_{\alpha:\alpha-1}$};
\draw[] (12.5,1) node[rotate=90] {\large  $\mathcal{E}_{\alpha+1:\alpha}$};
\draw[] (10.5,-1.75) node {\large  ${\mathcal{A}'}_{\alpha}$};
\draw[] (10.5,0) node {\large  $\mathcal{K}_{\alpha}$};

\draw[] (7.5,-0.5) node[rotate=90] {\large  $\mathcal{V}_{\alpha-1}$};
\draw[] (9.5,-0.5) node[rotate=90] {\large  $\mathcal{W}_{\alpha}$};
\draw[] (11.5,-0.5) node[rotate=90] {\large $\mathcal{V}_{\alpha}$};
\draw[] (13.5,-0.5) node[rotate=90] {\large $\mathcal{W}_{\alpha+1}$};

\draw[] (8.5,-1) node[rotate=90] {\large  $\mathcal{C}_{\alpha:\alpha-1}$};
\draw[] (12.5,-1) node[rotate=90] {\large  $\mathcal{C}_{\alpha+1:\alpha}$};

\draw[] (8.5,2.3) node[rotate=0] {\large  $\mathcal{M}_{\alpha:\alpha-1}$};
\draw[] (12.5,2.3) node[rotate=0] {\large  $\mathcal{M}_{\alpha+1:\alpha}$};

\draw[] (6.2,1.5) node {\large $\e$};
\draw[] (6.2,0.5) node {\large $\s$};
\draw[] (6.2,-0.5) node {\large $\z$};
\draw[] (6.2,-1.5) node {\large $\s'$};

\end{tikzpicture}
\caption{A representative subsection of the dynamics in a resource theory induced by restricting communication. Blue corresponds to the control sequence (client), purple corresponds to the superprocess (contractor), and red corresponds to the process tensor (server). In this physical scenario, the control sequence is local in time and uncorrelated (enabling the omission of $\an'$ from this diagram), while the superprocess can carry a memory correlated to $s$ via the ancillary subsystems. The utility of the superprocess at these tasks is determined by the form of the communication maps $\mathcal{C}$ and $\mathcal{K}$. $\mathcal{K}$ determines inter-step communication in parallel to a memoryless $\mathcal{A}$, while $\mathcal{C}$ accounts for intra-step communication. The grey shaded areas represent a decomposition of the dynamics into segments, as used in Sec.~\ref{Results}. The $\mathcal{K}$ type communication mediates what is passed between each $\mathcal{M}$ in the same way that the $\mathcal{C}$ type communication mediates what is passed between each $\mathcal{V}$ and $\mathcal{W}$.
}
\label{superprocesswithcommunication}
\end{figure*}

We now use communication maps to enumerate the classes of allowed operations in our resource theories. There are three classes $\{ \emptyset , \EB, \Q \}$ we will consider for each of $\mathcal{C}^{\zsp}_{\alpha+1:\alpha}$ and $\mathcal{K}^{\zs}_{\alpha}$. The choices are: $\emptyset$ no communication, $\EB$ entanglement breaking communication (slightly more general than strictly classical communication~\cite{eb}), and $\Q$ any quantum channel. These classes satisfy $\emptyset \subset \EB \subset \Q$.
We will study the resource theory induced by each combination of $\mathcal{C}^{\zsp}_{\alpha+1:\alpha}$ and $\mathcal{K}^{\zs}_{\alpha}$, resulting in an operational hierarchy of nine theories. For the sake of brevity, we will use $\mathcal{X}^x$ acting on operators on the Hilbert space $x$ as a placeholder for either $\mathcal{C}^{\zsp}_{\alpha+1:\alpha}$ or $\mathcal{K}^{\zs}_{\alpha}$ in the following cases.

\textbf{(}\textbf{$\emptyset$)} To account for the absence of communication, $\mathcal{X}^x$ consists of discarding the current state, and preparing an arbitrary fixed state. The form of $\mathcal{X}^x$ is 
\begin{equation} 
\begin{aligned}
\text{for} \quad \mathcal{X}^x \in \emptyset \quad \mathcal{X}^x[\rho^x]  =\text{tr} \{ \rho^x \} \tau^x,
\end{aligned}
\end{equation}
where $\tau^x$ is some arbitrary state which the input is erased to. We refer to this as a fixed output channel. 

\textbf{(}\textbf{$\EB$)} Entanglement breaking communication corresponds operationally to a classical agent who may only interact with the quantum system by state measurement and preparation. An entanglement breaking channel can always be represented by a positive-operator valued measure (POVM) measurement with a subsequent (possibly correlated) state re-preparation~\cite{eb}. Here, $\mathcal{X}^x$ is an entanglement breaking channel on $x$
\begin{equation} 
\begin{aligned}
\text{for} \quad \mathcal{X}^x \in \EB \quad \mathcal{X}^x[\rho^x]= \sum_k \nu^{x}_k \text{tr}\{ \Pi^{x}_k \rho^x \},
\end{aligned}
\end{equation}
where $\{\nu^{x}_k\}$ is a set of density operators and $\{\Pi^{x}_k\}$ defines a POVM~\cite{eb}. Entanglement breaking channels are more appropriate than fully classical channels in our context because, while they still only convey classical information, they are able admit quantum inputs and outputs.

\textbf{(}\textbf{$\Q$)} In the fully quantum case, $\mathcal{X}^x$ takes the form of any CPTP map acting on $x$: 
\begin{equation} 
\begin{aligned}
\text{for} \quad \mathcal{X}^x \in \Q \quad \mathcal{X}^x[\rho^x]= \sum_i X^{x}_i \rho^x {X^{x}_i}^\dagger,
\end{aligned}
\end{equation}
written in the Kraus form of a CPTP map on $x$ satisfying $ \sum_i {X^{x}_i}^\dagger X^{x}_i=1$ with the $\{X^{x}_i\}$ otherwise general~\cite{nc}. 

There may exist operational scenarios which not only allow for communication, but also pre-shared correlations. In these scenarios, in may be possible to carry types of information into the future which would be impossible with the specified type of communication alone. The most famous example of such a scenario is quantum teleportation: given pre-shared entanglement in a pair of qubits, the communication of two classical bits of information allows the state of one qubit to be teleported to the other~\cite{quantumentanglement}. If one allows for arbitrary amounts of classical communication (as in $\EB$) and supplementary entangled qubits, such an agent will be capable of any quantum communication, as in $\Q$. Consequently, this operational scenario should be classified as $\Q$. However, in the absence of any communication $\emptyset$, correlations cannot help the agent to perform new tasks.

In this subsection, we began by identifying where communication maps fit into the superprocess, and now we have enumerated three ways in which they might restrict the flow of information through time. Hence, the transformations in these resource theories can now be fully specified.

\subsection{Primitive Free Resources}
To complete the specification of our resource theories, one more ingredient is required: free resources. 

As discussed in Sec.~\ref{sec:RTdef}, a way to specify the set of free resources which elucidates their form is to identify a process which can be guaranteed to be producible for free under the operational framework, and apply the full range of allowed superprocesses. We call this starting point a \textit{primitive} free process. 

As such, in a resource theory formed by communication restrictions $(\mathcal{C},\mathcal{K})$, described by superprocess $\mathbf{Z}^{(\mathcal{C},\mathcal{K})}_{n:0}$, the set of free processes $\mathbf{T}^{(\mathcal{C},\mathcal{K})}_{n:0}$ are defined as the action of the superprocess on a primitive resource:
\begin{equation}  \label{eq:free}
 \mathbf{T}^{(\mathcal{C},\mathcal{K})}_{n:0} := \llbracket\mathbf{T}^{\text{prim}}_{n:0}|\mathbf{Z}^{(\mathcal{C},\mathcal{K})}_{n:0}.
\end{equation}
The complement of this set is the set of useful resources.

In all of the restricted communication resource theories, the natural choice for the primitive free process $\mathbf{T}^{\text{prim}}_{n:0}$ is an uncorrelated sequence of fixed output channels (plus an arbitrary initial state), whose Choi state is denoted by
\begin{equation}  \label{eq:2}
\begin{aligned}
 \mathbf{\Upsilon}^{\text{prim}}_{n:0} =  \bigotimes_{j=0}^{n-1} \left(  \tau_{j+1} \otimes \mathbbm{1}_j \right) \otimes \tau_0,
\end{aligned}
\end{equation}
where $\tau_j$ is an arbitrary state. Fixed output channels have zero information capacity, breaking any causal links in the environment between the past and the future. Thus, these processes can be considered to have no value or cost to agents.

\subsection{Enumeration of Theories} \label{Results}
We label a theory by the tuple $(\mathcal{C},\mathcal{K})$ where $\mathcal{C},\mathcal{K}\in \{ \emptyset , \EB, \Q \}$. The nine resultant theories are enumerated in Fig.~\ref{schematics}.

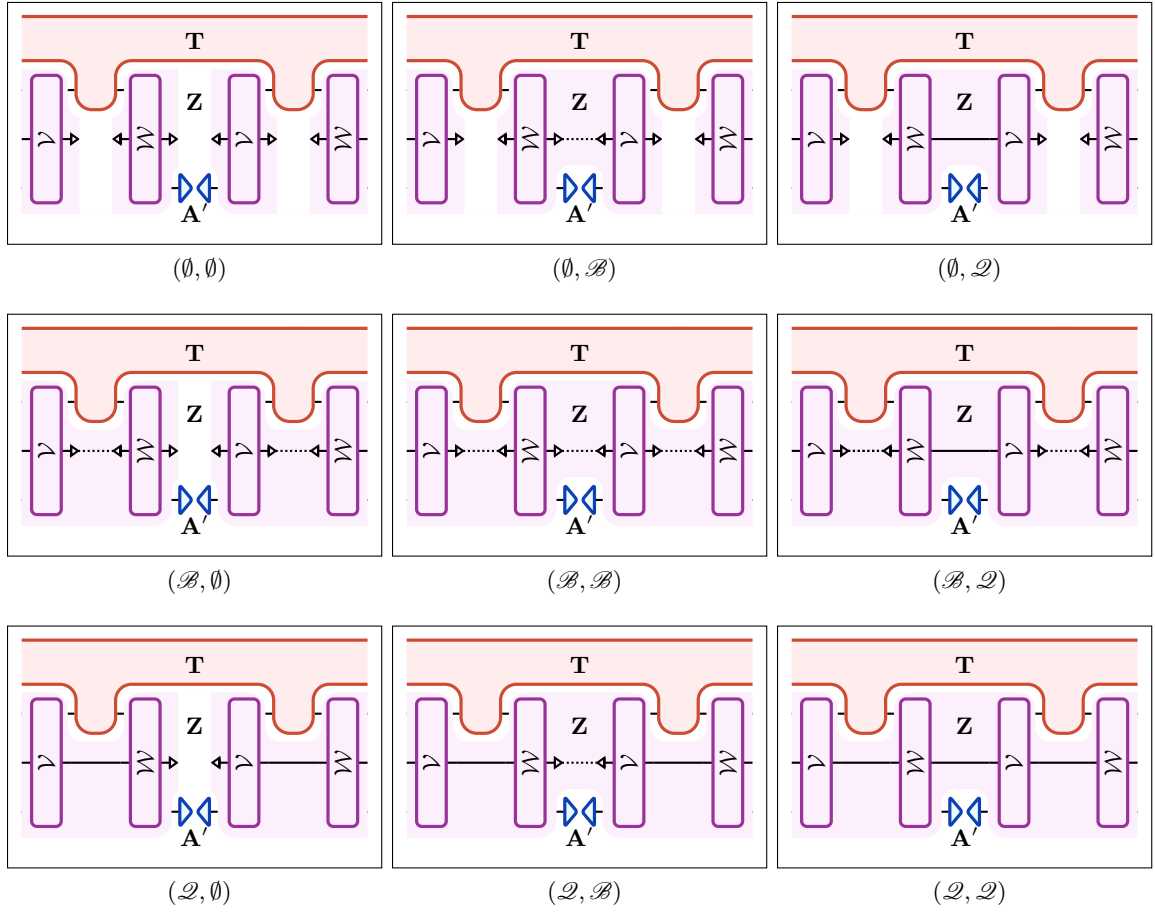
\begin{figure*} [htbp!]
\centering
\subfloat[$(\emptyset,\emptyset)$]
{\begin{tikzpicture}[scale=0.65,framed]

\draw[black,  thick,solid] (7,0.5) -- (14,0.5);
\draw[black,  thick,solid] (10.6,-1.5) -- (14,-1.5); \draw[black,  thick,solid] (7,-1.5) -- (10.4,-1.5);

\draw[myredfill,fill=myredfill, very thick,solid,rounded corners=6] (7,2) -- (12.9,2) -- (12.9,1.1) -- (12.9,0.1) -- (12.1,0.1) -- (12.1,1.1) -- (8.9,1.1) -- (8.9,0.1) -- (8.1,0.1) -- (8.1,1.1) -- (7,1.1);
\draw[myredfill,fill=myredfill, very thick,solid,rounded corners=6] (14,2) -- (12.1,2) -- (12.1,1.1) -- (14,1.1);
\draw[myred, very thick,solid,rounded corners=6] (7,2) -- (14,2);
\draw[myred, very thick,solid,rounded corners=6] (14,1.1) -- (12.9,1.1)-- (12.9,0.1) -- (12.1,0.1) -- (12.1,1.1) -- (8.9,1.1) -- (8.9,0.1) -- (8.1,0.1) -- (8.1,1.1) -- (7,1.1);

\draw[mypurplefill,fill=mypurplefill, very thick,solid,rounded corners=6] (14,0.9) -- (13.1,0.9)  -- (13.1,-0.1) -- (11.9,-0.1) -- (11.9,-2) -- (14,-2);
\draw[mypurplefill,fill=mypurplefill, very thick,solid,rounded corners=6] (7,-2) -- (10,-2) -- (10,-1) -- (11,-1) -- (11,-2) -- (13,-2) -- (13.1,-0.1) -- (11.9,-0.1) -- (11.9,0.9) -- (9.1,0.9) -- (9.1,-0.1) -- (7.9,-0.1) -- (7.9,0.9) -- (7,0.9);

\draw[black,  thick,solid] (7,-0.5) -- (8,-0.5);
\draw[black,  thick,solid] (9,-0.5) -- (10,-0.5);
\draw[black,  thick,solid] (11,-0.5) -- (12,-0.5);
\draw[black,  thick,solid] (13,-0.5) -- (14,-0.5);

\draw[black,  thick,solid] (8,-0.4) -- (8,-0.6) -- (8.14,-0.5) -- (8,-0.4) -- (8,-0.6);
\draw[black,  thick,solid] (10,-0.4) -- (10,-0.6) -- (10.14,-0.5) -- (10,-0.4) -- (10,-0.6);
\draw[black,  thick,solid] (12,-0.4) -- (12,-0.6) -- (12.14,-0.5) -- (12,-0.4) -- (12,-0.6);

\draw[black,  thick,solid] (9,-0.4) -- (9,-0.6) -- (8.86,-0.5) -- (9,-0.4) -- (9,-0.6);
\draw[black,  thick,solid] (11,-0.4) -- (11,-0.6) -- (10.86,-0.5) -- (11,-0.4) -- (11,-0.6);
\draw[black,  thick,solid] (13,-0.4) -- (13,-0.6) -- (12.86,-0.5) -- (13,-0.4) -- (13,-0.6);

\draw[myblue,fill=mybluefill,very thick,solid,rounded corners=1] (10.2,-1.8) -- (10.2,-1.2) -- (10.45,-1.5) -- cycle; \draw[myblue,fill=mybluefill,very thick,solid,rounded corners=1] (10.8,-1.8) -- (10.8,-1.2) -- (10.55,-1.5) -- cycle;

\draw[mypurple,fill=mypurplefill,very thick,solid,rounded corners=2] (7.2,0.8) rectangle (7.8,-1.8);
\draw[mypurple,fill=mypurplefill,very thick,solid,rounded corners=2] (9.2,0.8) rectangle (9.8,-1.8);
\draw[mypurple,fill=mypurplefill,very thick,solid,rounded corners=2] (11.2,0.8) rectangle (11.8,-1.8);
\draw[mypurple,fill=mypurplefill,very thick,solid,rounded corners=2] (13.2,0.8) rectangle (13.8,-1.8);
\draw[white,fill=white,very thick,solid,rounded corners=0] (10.2,0.9) rectangle (10.8,-1.1);
\draw[white,fill=white,very thick,solid,rounded corners=0] (8.2,-0.1) rectangle (8.8,-2);
\draw[white,fill=white,very thick,solid,rounded corners=0] (12.2,-0.1) rectangle (12.8,-2);

\draw[] (7.5,-0.5) node[rotate=90] {\footnotesize  $\mathcal{V}$};
\draw[] (9.5,-0.5) node[rotate=90] {\footnotesize  $\mathcal{W}$};
\draw[] (11.5,-0.5) node[rotate=90] {\footnotesize $\mathcal{V}$};
\draw[] (13.5,-0.5) node[rotate=90] {\footnotesize$\mathcal{W}$};
\draw[] (10.5,-2) node {\footnotesize  $\mathbf{A}'$};
\draw[] (10.5,1.5) node {\footnotesize  $\mathbf{T}$};
\draw[] (10.5,0.25) node {\footnotesize  $\mathbf{Z}$};

\end{tikzpicture}
}
\subfloat[$(\emptyset,\EB)$]
{\begin{tikzpicture}[scale=0.65,framed]

\draw[black,  thick,solid] (7,0.5) -- (14,0.5);
\draw[black,  thick,solid] (10.6,-1.5) -- (14,-1.5);
\draw[black,  thick,solid] (7,-1.5) -- (10.4,-1.5);

\draw[myredfill,fill=myredfill, very thick,solid,rounded corners=6] (7,2) -- (12.9,2) -- (12.9,1.1) -- (12.9,0.1) -- (12.1,0.1) -- (12.1,1.1) -- (8.9,1.1) -- (8.9,0.1) -- (8.1,0.1) -- (8.1,1.1) -- (7,1.1);
\draw[myredfill,fill=myredfill, very thick,solid,rounded corners=6] (14,2) -- (12.1,2) -- (12.1,1.1) -- (14,1.1);
\draw[myred, very thick,solid,rounded corners=6] (7,2) -- (14,2);
\draw[myred, very thick,solid,rounded corners=6] (14,1.1) -- (12.9,1.1)-- (12.9,0.1) -- (12.1,0.1) -- (12.1,1.1) -- (8.9,1.1) -- (8.9,0.1) -- (8.1,0.1) -- (8.1,1.1) -- (7,1.1);

\draw[mypurplefill,fill=mypurplefill, very thick,solid,rounded corners=6] (14,0.9) -- (13.1,0.9)  -- (13.1,-0.1) -- (11.9,-0.1) -- (11.9,-2) -- (14,-2);
\draw[mypurplefill,fill=mypurplefill, very thick,solid,rounded corners=6] (7,-2) -- (10,-2) -- (10,-1) -- (11,-1) -- (11,-2) -- (13,-2) -- (13.1,-0.1) -- (11.9,-0.1) -- (11.9,0.9) -- (9.1,0.9) -- (9.1,-0.1) -- (7.9,-0.1) -- (7.9,0.9) -- (7,0.9);

\draw[black,  thick,solid] (7,-0.5) -- (8,-0.5);
\draw[black,  thick,solid] (9,-0.5) -- (10,-0.5);
\draw[black,  thick,solid] (11,-0.5) -- (12,-0.5);
\draw[black,  thick,solid] (13,-0.5) -- (14,-0.5);

\draw[black,  thick,solid] (8,-0.4) -- (8,-0.6) -- (8.14,-0.5) -- (8,-0.4) -- (8,-0.6);
\draw[black,  thick,solid] (10,-0.4) -- (10,-0.6) -- (10.14,-0.5) -- (10,-0.4) -- (10,-0.6);
\draw[black,  thick,solid] (12,-0.4) -- (12,-0.6) -- (12.14,-0.5) -- (12,-0.4) -- (12,-0.6);

\draw[black,  thick,solid] (9,-0.4) -- (9,-0.6) -- (8.86,-0.5) -- (9,-0.4) -- (9,-0.6);
\draw[black,  thick,solid] (11,-0.4) -- (11,-0.6) -- (10.86,-0.5) -- (11,-0.4) -- (11,-0.6);
\draw[black,  thick,solid] (13,-0.4) -- (13,-0.6) -- (12.86,-0.5) -- (13,-0.4) -- (13,-0.6);

\draw[black,  thick,densely dotted] (10.14,-0.5) -- (10.86,-0.5);

\draw[myblue,fill=mybluefill,very thick,solid,rounded corners=1] (10.2,-1.8) -- (10.2,-1.2) -- (10.45,-1.5) -- cycle; \draw[myblue,fill=mybluefill,very thick,solid,rounded corners=1] (10.8,-1.8) -- (10.8,-1.2) -- (10.55,-1.5) -- cycle;

\draw[mypurple,fill=mypurplefill,very thick,solid,rounded corners=2] (7.2,0.8) rectangle (7.8,-1.8);
\draw[mypurple,fill=mypurplefill,very thick,solid,rounded corners=2] (9.2,0.8) rectangle (9.8,-1.8);
\draw[mypurple,fill=mypurplefill,very thick,solid,rounded corners=2] (11.2,0.8) rectangle (11.8,-1.8);
\draw[mypurple,fill=mypurplefill,very thick,solid,rounded corners=2] (13.2,0.8) rectangle (13.8,-1.8);
\draw[white,fill=white,very thick,solid,rounded corners=0] (8.2,-0.1) rectangle (8.8,-2);
\draw[white,fill=white,very thick,solid,rounded corners=0] (12.2,-0.1) rectangle (12.8,-2);

\draw[] (7.5,-0.5) node[rotate=90] {\footnotesize  $\mathcal{V}$};
\draw[] (9.5,-0.5) node[rotate=90] {\footnotesize  $\mathcal{W}$};
\draw[] (11.5,-0.5) node[rotate=90] {\footnotesize $\mathcal{V}$};
\draw[] (13.5,-0.5) node[rotate=90] {\footnotesize$\mathcal{W}$};
\draw[] (10.5,-2) node {\footnotesize  $\mathbf{A}'$};
\draw[] (10.5,1.5) node {\footnotesize  $\mathbf{T}$};
\draw[] (10.5,0.25) node {\footnotesize  $\mathbf{Z}$};

\end{tikzpicture}
}
\subfloat[$(\emptyset,\Q)$]
{\begin{tikzpicture}[scale=0.65,framed]

\draw[black,  thick,solid] (7,0.5) -- (14,0.5);
\draw[black,  thick,solid] (10.6,-1.5) -- (14,-1.5); \draw[black,  thick,solid] (7,-1.5) -- (10.4,-1.5);

\draw[myredfill,fill=myredfill, very thick,solid,rounded corners=6] (7,2) -- (12.9,2) -- (12.9,1.1) -- (12.9,0.1) -- (12.1,0.1) -- (12.1,1.1) -- (8.9,1.1) -- (8.9,0.1) -- (8.1,0.1) -- (8.1,1.1) -- (7,1.1);
\draw[myredfill,fill=myredfill, very thick,solid,rounded corners=6] (14,2) -- (12.1,2) -- (12.1,1.1) -- (14,1.1);
\draw[myred, very thick,solid,rounded corners=6] (7,2) -- (14,2);
\draw[myred, very thick,solid,rounded corners=6] (14,1.1) -- (12.9,1.1)-- (12.9,0.1) -- (12.1,0.1) -- (12.1,1.1) -- (8.9,1.1) -- (8.9,0.1) -- (8.1,0.1) -- (8.1,1.1) -- (7,1.1);

\draw[mypurplefill,fill=mypurplefill, very thick,solid,rounded corners=6] (14,0.9) -- (13.1,0.9)  -- (13.1,-0.1) -- (11.9,-0.1) -- (11.9,-2) -- (14,-2);
\draw[mypurplefill,fill=mypurplefill, very thick,solid,rounded corners=6] (7,-2) -- (10,-2) -- (10,-1) -- (11,-1) -- (11,-2) -- (13,-2) -- (13.1,-0.1) -- (11.9,-0.1) -- (11.9,0.9) -- (9.1,0.9) -- (9.1,-0.1) -- (7.9,-0.1) -- (7.9,0.9) -- (7,0.9);

\draw[black,  thick,solid] (7,-0.5) -- (8,-0.5);
\draw[black,  thick,solid] (9,-0.5) -- (10,-0.5);
\draw[black,  thick,solid] (11,-0.5) -- (12,-0.5);
\draw[black,  thick,solid] (13,-0.5) -- (14,-0.5);

\draw[black,  thick,solid] (8,-0.4) -- (8,-0.6) -- (8.14,-0.5) -- (8,-0.4) -- (8,-0.6);
\draw[black,  thick,solid] (12,-0.4) -- (12,-0.6) -- (12.14,-0.5) -- (12,-0.4) -- (12,-0.6);

\draw[black,  thick,solid] (9,-0.4) -- (9,-0.6) -- (8.86,-0.5) -- (9,-0.4) -- (9,-0.6);
\draw[black,  thick,solid] (13,-0.4) -- (13,-0.6) -- (12.86,-0.5) -- (13,-0.4) -- (13,-0.6);

\draw[black,  thick,solid] (10,-0.5) -- (11,-0.5);

\draw[myblue,fill=mybluefill,very thick,solid,rounded corners=1] (10.2,-1.8) -- (10.2,-1.2) -- (10.45,-1.5) -- cycle; \draw[myblue,fill=mybluefill,very thick,solid,rounded corners=1] (10.8,-1.8) -- (10.8,-1.2) -- (10.55,-1.5) -- cycle;

\draw[mypurple,fill=mypurplefill,very thick,solid,rounded corners=2] (7.2,0.8) rectangle (7.8,-1.8);
\draw[mypurple,fill=mypurplefill,very thick,solid,rounded corners=2] (9.2,0.8) rectangle (9.8,-1.8);
\draw[mypurple,fill=mypurplefill,very thick,solid,rounded corners=2] (11.2,0.8) rectangle (11.8,-1.8);
\draw[mypurple,fill=mypurplefill,very thick,solid,rounded corners=2] (13.2,0.8) rectangle (13.8,-1.8);
\draw[white,fill=white,very thick,solid,rounded corners=0] (8.2,-0.1) rectangle (8.8,-2);
\draw[white,fill=white,very thick,solid,rounded corners=0] (12.2,-0.1) rectangle (12.8,-2);

\draw[] (7.5,-0.5) node[rotate=90] {\footnotesize  $\mathcal{V}$};
\draw[] (9.5,-0.5) node[rotate=90] {\footnotesize  $\mathcal{W}$};
\draw[] (11.5,-0.5) node[rotate=90] {\footnotesize $\mathcal{V}$};
\draw[] (13.5,-0.5) node[rotate=90] {\footnotesize$\mathcal{W}$};
\draw[] (10.5,-2) node {\footnotesize  $\mathbf{A}'$};
\draw[] (10.5,1.5) node {\footnotesize  $\mathbf{T}$};
\draw[] (10.5,0.25) node {\footnotesize  $\mathbf{Z}$};

\end{tikzpicture}
}
\\
\subfloat[$(\EB,\emptyset)$]
{\begin{tikzpicture}[scale=0.65,framed]

\draw[black,  thick,solid] (7,0.5) -- (14,0.5);
\draw[black,  thick,solid] (10.6,-1.5) -- (14,-1.5); \draw[black,  thick,solid] (7,-1.5) -- (10.4,-1.5);

\draw[myredfill,fill=myredfill, very thick,solid,rounded corners=6] (7,2) -- (12.9,2) -- (12.9,1.1) -- (12.9,0.1) -- (12.1,0.1) -- (12.1,1.1) -- (8.9,1.1) -- (8.9,0.1) -- (8.1,0.1) -- (8.1,1.1) -- (7,1.1);
\draw[myredfill,fill=myredfill, very thick,solid,rounded corners=6] (14,2) -- (12.1,2) -- (12.1,1.1) -- (14,1.1);
\draw[myred, very thick,solid,rounded corners=6] (7,2) -- (14,2);
\draw[myred, very thick,solid,rounded corners=6] (14,1.1) -- (12.9,1.1)-- (12.9,0.1) -- (12.1,0.1) -- (12.1,1.1) -- (8.9,1.1) -- (8.9,0.1) -- (8.1,0.1) -- (8.1,1.1) -- (7,1.1);

\draw[mypurplefill,fill=mypurplefill, very thick,solid,rounded corners=6] (14,0.9) -- (13.1,0.9)  -- (13.1,-0.1) -- (11.9,-0.1) -- (11.9,-2) -- (14,-2);
\draw[mypurplefill,fill=mypurplefill, very thick,solid,rounded corners=6] (7,-2) -- (10,-2) -- (10,-1) -- (11,-1) -- (11,-2) -- (13,-2) -- (13.1,-0.1) -- (11.9,-0.1) -- (11.9,0.9) -- (9.1,0.9) -- (9.1,-0.1) -- (7.9,-0.1) -- (7.9,0.9) -- (7,0.9);

\draw[black,  thick,solid] (7,-0.5) -- (8,-0.5);
\draw[black,  thick,solid] (9,-0.5) -- (10,-0.5);
\draw[black,  thick,solid] (11,-0.5) -- (12,-0.5);
\draw[black,  thick,solid] (13,-0.5) -- (14,-0.5);

\draw[black,  thick,solid] (8,-0.4) -- (8,-0.6) -- (8.14,-0.5) -- (8,-0.4) -- (8,-0.6);
\draw[black,  thick,solid] (10,-0.4) -- (10,-0.6) -- (10.14,-0.5) -- (10,-0.4) -- (10,-0.6);
\draw[black,  thick,solid] (12,-0.4) -- (12,-0.6) -- (12.14,-0.5) -- (12,-0.4) -- (12,-0.6);

\draw[black,  thick,solid] (9,-0.4) -- (9,-0.6) -- (8.86,-0.5) -- (9,-0.4) -- (9,-0.6);
\draw[black,  thick,solid] (11,-0.4) -- (11,-0.6) -- (10.86,-0.5) -- (11,-0.4) -- (11,-0.6);
\draw[black,  thick,solid] (13,-0.4) -- (13,-0.6) -- (12.86,-0.5) -- (13,-0.4) -- (13,-0.6);

\draw[black,  thick,densely dotted] (8.14,-0.5) -- (8.86,-0.5);
\draw[black,  thick,densely dotted] (12.14,-0.5) -- (12.86,-0.5);

\draw[myblue,fill=mybluefill,very thick,solid,rounded corners=1] (10.2,-1.8) -- (10.2,-1.2) -- (10.45,-1.5) -- cycle; \draw[myblue,fill=mybluefill,very thick,solid,rounded corners=1] (10.8,-1.8) -- (10.8,-1.2) -- (10.55,-1.5) -- cycle;

\draw[mypurple,fill=mypurplefill,very thick,solid,rounded corners=2] (7.2,0.8) rectangle (7.8,-1.8);
\draw[mypurple,fill=mypurplefill,very thick,solid,rounded corners=2] (9.2,0.8) rectangle (9.8,-1.8);
\draw[mypurple,fill=mypurplefill,very thick,solid,rounded corners=2] (11.2,0.8) rectangle (11.8,-1.8);
\draw[mypurple,fill=mypurplefill,very thick,solid,rounded corners=2] (13.2,0.8) rectangle (13.8,-1.8);
\draw[white,fill=white,very thick,solid,rounded corners=0] (10.2,0.9) rectangle (10.8,-1.1);

\draw[] (7.5,-0.5) node[rotate=90] {\footnotesize  $\mathcal{V}$};
\draw[] (9.5,-0.5) node[rotate=90] {\footnotesize  $\mathcal{W}$};
\draw[] (11.5,-0.5) node[rotate=90] {\footnotesize $\mathcal{V}$};
\draw[] (13.5,-0.5) node[rotate=90] {\footnotesize$\mathcal{W}$};
\draw[] (10.5,-2) node {\footnotesize  $\mathbf{A}'$};
\draw[] (10.5,1.5) node {\footnotesize  $\mathbf{T}$};
\draw[] (10.5,0.25) node {\footnotesize  $\mathbf{Z}$};

\end{tikzpicture}
}
\subfloat[$(\EB,\EB)$]
{\begin{tikzpicture}[scale=0.65,framed]

\draw[black,  thick,solid] (7,0.5) -- (14,0.5);
\draw[black,  thick,solid] (10.6,-1.5) -- (14,-1.5); \draw[black,  thick,solid] (7,-1.5) -- (10.4,-1.5);

\draw[myredfill,fill=myredfill, very thick,solid,rounded corners=6] (7,2) -- (12.9,2) -- (12.9,1.1) -- (12.9,0.1) -- (12.1,0.1) -- (12.1,1.1) -- (8.9,1.1) -- (8.9,0.1) -- (8.1,0.1) -- (8.1,1.1) -- (7,1.1);
\draw[myredfill,fill=myredfill, very thick,solid,rounded corners=6] (14,2) -- (12.1,2) -- (12.1,1.1) -- (14,1.1);
\draw[myred, very thick,solid,rounded corners=6] (7,2) -- (14,2);
\draw[myred, very thick,solid,rounded corners=6] (14,1.1) -- (12.9,1.1)-- (12.9,0.1) -- (12.1,0.1) -- (12.1,1.1) -- (8.9,1.1) -- (8.9,0.1) -- (8.1,0.1) -- (8.1,1.1) -- (7,1.1);

\draw[mypurplefill,fill=mypurplefill, very thick,solid,rounded corners=6] (14,0.9) -- (13.1,0.9)  -- (13.1,-0.1) -- (11.9,-0.1) -- (11.9,-2) -- (14,-2);
\draw[mypurplefill,fill=mypurplefill, very thick,solid,rounded corners=6] (7,-2) -- (10,-2) -- (10,-1) -- (11,-1) -- (11,-2) -- (13,-2) -- (13.1,-0.1) -- (11.9,-0.1) -- (11.9,0.9) -- (9.1,0.9) -- (9.1,-0.1) -- (7.9,-0.1) -- (7.9,0.9) -- (7,0.9);

\draw[black,  thick,solid] (7,-0.5) -- (8,-0.5);
\draw[black,  thick,solid] (9,-0.5) -- (10,-0.5);
\draw[black,  thick,solid] (11,-0.5) -- (12,-0.5);
\draw[black,  thick,solid] (13,-0.5) -- (14,-0.5);

\draw[black,  thick,solid] (8,-0.4) -- (8,-0.6) -- (8.14,-0.5) -- (8,-0.4) -- (8,-0.6);
\draw[black,  thick,solid] (10,-0.4) -- (10,-0.6) -- (10.14,-0.5) -- (10,-0.4) -- (10,-0.6);
\draw[black,  thick,solid] (12,-0.4) -- (12,-0.6) -- (12.14,-0.5) -- (12,-0.4) -- (12,-0.6);

\draw[black,  thick,solid] (9,-0.4) -- (9,-0.6) -- (8.86,-0.5) -- (9,-0.4) -- (9,-0.6);
\draw[black,  thick,solid] (11,-0.4) -- (11,-0.6) -- (10.86,-0.5) -- (11,-0.4) -- (11,-0.6);
\draw[black,  thick,solid] (13,-0.4) -- (13,-0.6) -- (12.86,-0.5) -- (13,-0.4) -- (13,-0.6);

\draw[myblue,fill=mybluefill,very thick,solid,rounded corners=1] (10.2,-1.8) -- (10.2,-1.2) -- (10.45,-1.5) -- cycle; \draw[myblue,fill=mybluefill,very thick,solid,rounded corners=1] (10.8,-1.8) -- (10.8,-1.2) -- (10.55,-1.5) -- cycle;

\draw[mypurple,fill=mypurplefill,very thick,solid,rounded corners=2] (7.2,0.8) rectangle (7.8,-1.8);
\draw[mypurple,fill=mypurplefill,very thick,solid,rounded corners=2] (9.2,0.8) rectangle (9.8,-1.8);
\draw[mypurple,fill=mypurplefill,very thick,solid,rounded corners=2] (11.2,0.8) rectangle (11.8,-1.8);
\draw[mypurple,fill=mypurplefill,very thick,solid,rounded corners=2] (13.2,0.8) rectangle (13.8,-1.8);

\draw[black,  thick,densely dotted] (8.14,-0.5) -- (8.86,-0.5);
\draw[black,  thick,densely dotted] (12.14,-0.5) -- (12.86,-0.5);
\draw[black,  thick,densely dotted] (10.14,-0.5) -- (10.86,-0.5);

\draw[] (7.5,-0.5) node[rotate=90] {\footnotesize  $\mathcal{V}$};
\draw[] (9.5,-0.5) node[rotate=90] {\footnotesize  $\mathcal{W}$};
\draw[] (11.5,-0.5) node[rotate=90] {\footnotesize $\mathcal{V}$};
\draw[] (13.5,-0.5) node[rotate=90] {\footnotesize$\mathcal{W}$};
\draw[] (10.5,-2) node {\footnotesize  $\mathbf{A}'$};
\draw[] (10.5,1.5) node {\footnotesize  $\mathbf{T}$};
\draw[] (10.5,0.25) node {\footnotesize  $\mathbf{Z}$};

\end{tikzpicture}
}
\subfloat[$(\EB,\Q)$]
{\begin{tikzpicture}[scale=0.65,framed]

\draw[black,  thick,solid] (7,0.5) -- (14,0.5);
\draw[black,  thick,solid] (10.6,-1.5) -- (14,-1.5); \draw[black,  thick,solid] (7,-1.5) -- (10.4,-1.5);

\draw[myredfill,fill=myredfill, very thick,solid,rounded corners=6] (7,2) -- (12.9,2) -- (12.9,1.1) -- (12.9,0.1) -- (12.1,0.1) -- (12.1,1.1) -- (8.9,1.1) -- (8.9,0.1) -- (8.1,0.1) -- (8.1,1.1) -- (7,1.1);
\draw[myredfill,fill=myredfill, very thick,solid,rounded corners=6] (14,2) -- (12.1,2) -- (12.1,1.1) -- (14,1.1);
\draw[myred, very thick,solid,rounded corners=6] (7,2) -- (14,2);
\draw[myred, very thick,solid,rounded corners=6] (14,1.1) -- (12.9,1.1)-- (12.9,0.1) -- (12.1,0.1) -- (12.1,1.1) -- (8.9,1.1) -- (8.9,0.1) -- (8.1,0.1) -- (8.1,1.1) -- (7,1.1);

\draw[mypurplefill,fill=mypurplefill, very thick,solid,rounded corners=6] (14,0.9) -- (13.1,0.9)  -- (13.1,-0.1) -- (11.9,-0.1) -- (11.9,-2) -- (14,-2);
\draw[mypurplefill,fill=mypurplefill, very thick,solid,rounded corners=6] (7,-2) -- (10,-2) -- (10,-1) -- (11,-1) -- (11,-2) -- (13,-2) -- (13.1,-0.1) -- (11.9,-0.1) -- (11.9,0.9) -- (9.1,0.9) -- (9.1,-0.1) -- (7.9,-0.1) -- (7.9,0.9) -- (7,0.9);

\draw[black,  thick,solid] (7,-0.5) -- (8,-0.5);
\draw[black,  thick,solid] (9,-0.5) -- (10,-0.5);
\draw[black,  thick,solid] (11,-0.5) -- (12,-0.5);
\draw[black,  thick,solid] (13,-0.5) -- (14,-0.5);

\draw[black,  thick,solid] (8,-0.4) -- (8,-0.6) -- (8.14,-0.5) -- (8,-0.4) -- (8,-0.6);
\draw[black,  thick,solid] (12,-0.4) -- (12,-0.6) -- (12.14,-0.5) -- (12,-0.4) -- (12,-0.6);

\draw[black,  thick,solid] (9,-0.4) -- (9,-0.6) -- (8.86,-0.5) -- (9,-0.4) -- (9,-0.6);
\draw[black,  thick,solid] (13,-0.4) -- (13,-0.6) -- (12.86,-0.5) -- (13,-0.4) -- (13,-0.6);

\draw[black,  thick,densely dotted] (8.14,-0.5) -- (8.86,-0.5);
\draw[black,  thick,densely dotted] (12.14,-0.5) -- (12.86,-0.5);
\draw[black,  thick,densely dotted] (8.14,-0.5) -- (8.86,-0.5);
\draw[black,  thick,solid] (10,-0.5) -- (11,-0.5);

\draw[myblue,fill=mybluefill,very thick,solid,rounded corners=1] (10.2,-1.8) -- (10.2,-1.2) -- (10.45,-1.5) -- cycle; \draw[myblue,fill=mybluefill,very thick,solid,rounded corners=1] (10.8,-1.8) -- (10.8,-1.2) -- (10.55,-1.5) -- cycle;

\draw[mypurple,fill=mypurplefill,very thick,solid,rounded corners=2] (7.2,0.8) rectangle (7.8,-1.8);
\draw[mypurple,fill=mypurplefill,very thick,solid,rounded corners=2] (9.2,0.8) rectangle (9.8,-1.8);
\draw[mypurple,fill=mypurplefill,very thick,solid,rounded corners=2] (11.2,0.8) rectangle (11.8,-1.8);
\draw[mypurple,fill=mypurplefill,very thick,solid,rounded corners=2] (13.2,0.8) rectangle (13.8,-1.8);

\draw[] (7.5,-0.5) node[rotate=90] {\footnotesize  $\mathcal{V}$};
\draw[] (9.5,-0.5) node[rotate=90] {\footnotesize  $\mathcal{W}$};
\draw[] (11.5,-0.5) node[rotate=90] {\footnotesize $\mathcal{V}$};
\draw[] (13.5,-0.5) node[rotate=90] {\footnotesize$\mathcal{W}$};
\draw[] (10.5,-2) node {\footnotesize  $\mathbf{A}'$};
\draw[] (10.5,1.5) node {\footnotesize  $\mathbf{T}$};
\draw[] (10.5,0.25) node {\footnotesize  $\mathbf{Z}$};

\end{tikzpicture}
}
\\
\subfloat[$(\Q,\emptyset)$]
{\begin{tikzpicture}[scale=0.65,framed]

\draw[black,  thick,solid] (7,0.5) -- (14,0.5);
\draw[black,  thick,solid] (10.6,-1.5) -- (14,-1.5); \draw[black,  thick,solid] (7,-1.5) -- (10.4,-1.5);

\draw[myredfill,fill=myredfill, very thick,solid,rounded corners=6] (7,2) -- (12.9,2) -- (12.9,1.1) -- (12.9,0.1) -- (12.1,0.1) -- (12.1,1.1) -- (8.9,1.1) -- (8.9,0.1) -- (8.1,0.1) -- (8.1,1.1) -- (7,1.1);
\draw[myredfill,fill=myredfill, very thick,solid,rounded corners=6] (14,2) -- (12.1,2) -- (12.1,1.1) -- (14,1.1);
\draw[myred, very thick,solid,rounded corners=6] (7,2) -- (14,2);
\draw[myred, very thick,solid,rounded corners=6] (14,1.1) -- (12.9,1.1)-- (12.9,0.1) -- (12.1,0.1) -- (12.1,1.1) -- (8.9,1.1) -- (8.9,0.1) -- (8.1,0.1) -- (8.1,1.1) -- (7,1.1);

\draw[mypurplefill,fill=mypurplefill, very thick,solid,rounded corners=6] (14,0.9) -- (13.1,0.9)  -- (13.1,-0.1) -- (11.9,-0.1) -- (11.9,-2) -- (14,-2);
\draw[mypurplefill,fill=mypurplefill, very thick,solid,rounded corners=6] (7,-2) -- (10,-2) -- (10,-1) -- (11,-1) -- (11,-2) -- (13,-2) -- (13.1,-0.1) -- (11.9,-0.1) -- (11.9,0.9) -- (9.1,0.9) -- (9.1,-0.1) -- (7.9,-0.1) -- (7.9,0.9) -- (7,0.9);

\draw[black,  thick,solid] (7,-0.5) -- (8,-0.5);
\draw[black,  thick,solid] (9,-0.5) -- (10,-0.5);
\draw[black,  thick,solid] (11,-0.5) -- (12,-0.5);
\draw[black,  thick,solid] (13,-0.5) -- (14,-0.5);

\draw[black,  thick,solid] (10,-0.4) -- (10,-0.6) -- (10.14,-0.5) -- (10,-0.4) -- (10,-0.6);

\draw[black,  thick,solid] (11,-0.4) -- (11,-0.6) -- (10.86,-0.5) -- (11,-0.4) -- (11,-0.6);

\draw[black,  thick,solid] (12,-0.5) -- (13,-0.5);
\draw[black,  thick,solid] (8,-0.5) -- (9,-0.5);

\draw[myblue,fill=mybluefill,very thick,solid,rounded corners=1] (10.2,-1.8) -- (10.2,-1.2) -- (10.45,-1.5) -- cycle; \draw[myblue,fill=mybluefill,very thick,solid,rounded corners=1] (10.8,-1.8) -- (10.8,-1.2) -- (10.55,-1.5) -- cycle;

\draw[mypurple,fill=mypurplefill,very thick,solid,rounded corners=2] (7.2,0.8) rectangle (7.8,-1.8);
\draw[mypurple,fill=mypurplefill,very thick,solid,rounded corners=2] (9.2,0.8) rectangle (9.8,-1.8);
\draw[mypurple,fill=mypurplefill,very thick,solid,rounded corners=2] (11.2,0.8) rectangle (11.8,-1.8);
\draw[mypurple,fill=mypurplefill,very thick,solid,rounded corners=2] (13.2,0.8) rectangle (13.8,-1.8);
\draw[white,fill=white,very thick,solid,rounded corners=0] (10.2,0.9) rectangle (10.8,-1.1);

\draw[] (7.5,-0.5) node[rotate=90] {\footnotesize  $\mathcal{V}$};
\draw[] (9.5,-0.5) node[rotate=90] {\footnotesize  $\mathcal{W}$};
\draw[] (11.5,-0.5) node[rotate=90] {\footnotesize $\mathcal{V}$};
\draw[] (13.5,-0.5) node[rotate=90] {\footnotesize$\mathcal{W}$};
\draw[] (10.5,-2) node {\footnotesize  $\mathbf{A}'$};
\draw[] (10.5,1.5) node {\footnotesize  $\mathbf{T}$};
\draw[] (10.5,0.25) node {\footnotesize  $\mathbf{Z}$};

\end{tikzpicture}
}
\subfloat[$(\Q,\EB)$]
{\begin{tikzpicture}[scale=0.65,framed]

\draw[black,  thick,solid] (7,0.5) -- (14,0.5);
\draw[black,  thick,solid] (10.6,-1.5) -- (14,-1.5); \draw[black,  thick,solid] (7,-1.5) -- (10.4,-1.5);

\draw[myredfill,fill=myredfill, very thick,solid,rounded corners=6] (7,2) -- (12.9,2) -- (12.9,1.1) -- (12.9,0.1) -- (12.1,0.1) -- (12.1,1.1) -- (8.9,1.1) -- (8.9,0.1) -- (8.1,0.1) -- (8.1,1.1) -- (7,1.1);
\draw[myredfill,fill=myredfill, very thick,solid,rounded corners=6] (14,2) -- (12.1,2) -- (12.1,1.1) -- (14,1.1);
\draw[myred, very thick,solid,rounded corners=6] (7,2) -- (14,2);
\draw[myred, very thick,solid,rounded corners=6] (14,1.1) -- (12.9,1.1)-- (12.9,0.1) -- (12.1,0.1) -- (12.1,1.1) -- (8.9,1.1) -- (8.9,0.1) -- (8.1,0.1) -- (8.1,1.1) -- (7,1.1);

\draw[mypurplefill,fill=mypurplefill, very thick,solid,rounded corners=6] (14,0.9) -- (13.1,0.9)  -- (13.1,-0.1) -- (11.9,-0.1) -- (11.9,-2) -- (14,-2);
\draw[mypurplefill,fill=mypurplefill, very thick,solid,rounded corners=6] (7,-2) -- (10,-2) -- (10,-1) -- (11,-1) -- (11,-2) -- (13,-2) -- (13.1,-0.1) -- (11.9,-0.1) -- (11.9,0.9) -- (9.1,0.9) -- (9.1,-0.1) -- (7.9,-0.1) -- (7.9,0.9) -- (7,0.9);

\draw[black,  thick,solid] (7,-0.5) -- (8,-0.5);
\draw[black,  thick,solid] (9,-0.5) -- (10,-0.5);
\draw[black,  thick,solid] (11,-0.5) -- (12,-0.5);
\draw[black,  thick,solid] (13,-0.5) -- (14,-0.5);

\draw[black,  thick,solid] (10,-0.4) -- (10,-0.6) -- (10.14,-0.5) -- (10,-0.4) -- (10,-0.6);

\draw[black,  thick,solid] (11,-0.4) -- (11,-0.6) -- (10.86,-0.5) -- (11,-0.4) -- (11,-0.6);

\draw[black,  thick,solid] (12,-0.5) -- (13,-0.5);
\draw[black,  thick,solid] (8,-0.5) -- (9,-0.5);
\draw[black,  thick,densely dotted] (10.14,-0.5) -- (10.86,-0.5);

\draw[myblue,fill=mybluefill,very thick,solid,rounded corners=1] (10.2,-1.8) -- (10.2,-1.2) -- (10.45,-1.5) -- cycle; \draw[myblue,fill=mybluefill,very thick,solid,rounded corners=1] (10.8,-1.8) -- (10.8,-1.2) -- (10.55,-1.5) -- cycle;

\draw[mypurple,fill=mypurplefill,very thick,solid,rounded corners=2] (7.2,0.8) rectangle (7.8,-1.8);
\draw[mypurple,fill=mypurplefill,very thick,solid,rounded corners=2] (9.2,0.8) rectangle (9.8,-1.8);
\draw[mypurple,fill=mypurplefill,very thick,solid,rounded corners=2] (11.2,0.8) rectangle (11.8,-1.8);
\draw[mypurple,fill=mypurplefill,very thick,solid,rounded corners=2] (13.2,0.8) rectangle (13.8,-1.8);

\draw[] (7.5,-0.5) node[rotate=90] {\footnotesize  $\mathcal{V}$};
\draw[] (9.5,-0.5) node[rotate=90] {\footnotesize  $\mathcal{W}$};
\draw[] (11.5,-0.5) node[rotate=90] {\footnotesize $\mathcal{V}$};
\draw[] (13.5,-0.5) node[rotate=90] {\footnotesize$\mathcal{W}$};
\draw[] (10.5,-2) node {\footnotesize  $\mathbf{A}'$};
\draw[] (10.5,1.5) node {\footnotesize  $\mathbf{T}$};
\draw[] (10.5,0.25) node {\footnotesize  $\mathbf{Z}$};

\end{tikzpicture}
}
\subfloat[$(\Q,\Q)$]
{\begin{tikzpicture}[scale=0.65,framed]

\draw[black,  thick,solid] (7,0.5) -- (14,0.5);
\draw[black,  thick,solid] (10.6,-1.5) -- (14,-1.5); \draw[black,  thick,solid] (7,-1.5) -- (10.4,-1.5);

\draw[myredfill,fill=myredfill, very thick,solid,rounded corners=6] (7,2) -- (12.9,2) -- (12.9,1.1) -- (12.9,0.1) -- (12.1,0.1) -- (12.1,1.1) -- (8.9,1.1) -- (8.9,0.1) -- (8.1,0.1) -- (8.1,1.1) -- (7,1.1);
\draw[myredfill,fill=myredfill, very thick,solid,rounded corners=6] (14,2) -- (12.1,2) -- (12.1,1.1) -- (14,1.1);
\draw[myred, very thick,solid,rounded corners=6] (7,2) -- (14,2);
\draw[myred, very thick,solid,rounded corners=6] (14,1.1) -- (12.9,1.1)-- (12.9,0.1) -- (12.1,0.1) -- (12.1,1.1) -- (8.9,1.1) -- (8.9,0.1) -- (8.1,0.1) -- (8.1,1.1) -- (7,1.1);

\draw[mypurplefill,fill=mypurplefill, very thick,solid,rounded corners=6] (14,0.9) -- (13.1,0.9)  -- (13.1,-0.1) -- (11.9,-0.1) -- (11.9,-2) -- (14,-2);
\draw[mypurplefill,fill=mypurplefill, very thick,solid,rounded corners=6] (7,-2) -- (10,-2) -- (10,-1) -- (11,-1) -- (11,-2) -- (13,-2) -- (13.1,-0.1) -- (11.9,-0.1) -- (11.9,0.9) -- (9.1,0.9) -- (9.1,-0.1) -- (7.9,-0.1) -- (7.9,0.9) -- (7,0.9);

\draw[black,  thick,solid] (7,-0.5) -- (8,-0.5);
\draw[black,  thick,solid] (9,-0.5) -- (10,-0.5);
\draw[black,  thick,solid] (11,-0.5) -- (12,-0.5);
\draw[black,  thick,solid] (13,-0.5) -- (14,-0.5);

\draw[black,  thick,solid] (12,-0.5) -- (13,-0.5);
\draw[black,  thick,solid] (8,-0.5) -- (9,-0.5);
\draw[black,  thick,solid] (10,-0.5) -- (11,-0.5);

\draw[myblue,fill=mybluefill,very thick,solid,rounded corners=1] (10.2,-1.8) -- (10.2,-1.2) -- (10.45,-1.5) -- cycle; \draw[myblue,fill=mybluefill,very thick,solid,rounded corners=1] (10.8,-1.8) -- (10.8,-1.2) -- (10.55,-1.5) -- cycle;

\draw[mypurple,fill=mypurplefill,very thick,solid,rounded corners=2] (7.2,0.8) rectangle (7.8,-1.8);
\draw[mypurple,fill=mypurplefill,very thick,solid,rounded corners=2] (9.2,0.8) rectangle (9.8,-1.8);
\draw[mypurple,fill=mypurplefill,very thick,solid,rounded corners=2] (11.2,0.8) rectangle (11.8,-1.8);
\draw[mypurple,fill=mypurplefill,very thick,solid,rounded corners=2] (13.2,0.8) rectangle (13.8,-1.8);

\draw[] (7.5,-0.5) node[rotate=90] {\footnotesize  $\mathcal{V}$};
\draw[] (9.5,-0.5) node[rotate=90] {\footnotesize  $\mathcal{W}$};
\draw[] (11.5,-0.5) node[rotate=90] {\footnotesize $\mathcal{V}$};
\draw[] (13.5,-0.5) node[rotate=90] {\footnotesize$\mathcal{W}$};
\draw[] (10.5,-2) node {\footnotesize  $\mathbf{A}'$};
\draw[] (10.5,1.5) node {\footnotesize  $\mathbf{T}$};
\draw[] (10.5,0.25) node {\footnotesize  $\mathbf{Z}$};

\end{tikzpicture}
}
\caption{Schematic diagrams for each resource theory as shown by indices $(-,-)$. The first index specifies the type of communication in parallel with each leg of the process tensor. The second index indicates the type in parallel with each operation in the control sequence. Right pointing triangles signify measurement of a quantum state, left pointing triangles indicate a re-preparation of the subsystem. An empty gap between these corresponds to $\emptyset$ where the measurement and re-preparation are completely uncorrelated. A dotted line signifies classical communication between the measurement and re-preparation, as in $\EB$. A solid line without triangles means that any type of quantum communication is allowed, as per $\Q$. When the break in information flow is in parallel with the control operations, this indicates that the allowed superprocesses cannot imbue a process tensor with that kind of memory between steps. When the break is in parallel to a leg of the process tensor, the effective control sequence acting on the process tensor cannot contain the corresponding type of memory.}
\label{schematics}
\end{figure*}

In the next few subsections, we identify the full set of free resources in each theory. This is done by finding the image of our previously defined primitive free resource under all allowed superprocesses in that resource theory. As we have taken the background process to be the primitive free resource, the $e$ subsystem is no longer relevant to our analysis, and will be omitted. As explained in Sec.~\ref{commmaps}, the control operations ${\mathcal{A}'}^\s_{\alpha+1}$ are taken to be local in time and uncorrelated, implying that they can be represented by fixed output channels. The general form for evolution under the primitive free resource after the action of a superprocess is
\begin{align} \label{eq:superprocessprimitive}
  \llbracket\mathbf{T}^{(\mathcal{C},\mathcal{K})}_{n:0} | \mathbf{A}'_{n-1:0} \rrbracket &:= 
\llbracket\mathbf{T}^{\text{prim}}_{n:0}{|}  \mathbf{Z}^{(\mathcal{C},\mathcal{K})}_{n:0}|\mathbf{A}'_{n-1:0} \rrbracket
  \notag \\ &= \text{tr}_{\zsp} \!\left\{\! {\mathcal{A}'}^{\s'}_{n}\! \functioncomposition\limits_{\alpha=0}^{n-1}\! \left( \mathcal{M}^{{\szsp}}_{\alpha+1:\alpha}\!   \circ (\mathcal{K}^{\zs}_{\alpha} \otimes {\mathcal{A}'}^{\s'}_{\alpha} )\right)  \circ \mathcal{W}^{\szsp}_0[\rho^{\szsp}_0]\! \right\},
\end{align}
where now  $\mathcal{M}^{\szsp}_{\alpha+1:\alpha}=\mathcal{W}^{\szsp}_{\alpha+1} \circ (\mathcal{E}^{\s}_{\alpha+1:\alpha}\otimes \mathcal{C}^{\zsp}_{\alpha+1:\alpha}) \circ \mathcal{V}^{\szsp}_{\alpha}$ (cf. Eqs.~\eqref{eq:procexpand}~and~\eqref{eq:mdef}). See Appendix~\ref{appsteps} for more details on how $\mathcal{C}^{\zsp}_{\alpha+1:\alpha}$ determines the properties of $\mathcal{M}^{{\sezsp}}_{\alpha+1:\alpha}$, and see Appendix~\ref{app-interstep} for an outline of how, along with $\mathcal{K}^{\zs}_{\alpha}$, this specifies the free processes of the theories.

\subsubsection{\texorpdfstring{$(\emptyset, \emptyset)$}{(0,0)}, \texorpdfstring{$(\emptyset, \EB)$}{(0,EB)}, and \texorpdfstring{$(\emptyset, \Q)$}{(0,Q)}}
These three theories can be thought of as an experimenter who is performing a sequence of pre-determined maps on a system. The lack of any $\mathcal{C}$-type communication implies that the experimenter must rely on the process tensor to carry any information through time. The $(\emptyset, \Q)$ theory may be used to represent a scenario such as dynamical decoupling, where a pre-determined sequence of maps (unitary pulses) is used to preserve quantum information, given some kind of noisy environment. The $(\emptyset, \EB)$ experimenter is less capable; the quantum maps are replaced with measurements and re-preparations. This experimenter could not perform dynamical decoupling but could harness the Zeno effect. Finally, the experimenter actions in $(\emptyset, \emptyset)$ are completely destructive to all information. This theory could describe a situation like repeated scattering. 

In these theories, $\mathcal{C}^\zsp$ is a fixed output channel. For $\mathbf{Z}^{(\emptyset,\mathcal{K})}_{n:0}$ (top row of Fig.~\ref{schematics}), independent of $\mathcal{K}^{\zs}$, the right hand side of Eq.~\eqref{eq:superprocessprimitive} becomes
\begin{equation}
\text{tr}_{\zsp} \left\{ {\mathcal{A}'}^{\s'}_{n} \functioncomposition\limits_{\alpha=1}^{n} \left( \nu'^{\szsp}_{\alpha}  \text{tr}_{\szsp} \right)  \circ \mathcal{W}^{\szsp}_0[\rho^{\szsp}_0] \right\}.
\end{equation}

To get here, from Eq.~\eqref{eq:superprocessprimitive}, we observe the fact that -- in the absence of a useful process tensor -- $\mathcal{M}^{{\szsp}}_{\alpha+1:\alpha}$ destroys any links in the system between the past and future via $\mathcal{E}^{{\se} }_{\alpha+1:\alpha}$ and  $\mathcal{C}^{\zsp}_{\alpha+1:\alpha}$ which are both fixed output maps; an agent subject to these constraints must rely on a useful background process if they hope to preserve any information. Hence, $(\mathcal{E}^{{\se} }_{\alpha+1:\alpha} \otimes \mathcal{C}^{\zsp}_{\alpha+1:\alpha})[\mathcal{V}^{\szsp}_{\alpha} \circ (\mathcal{K}^{\zs}_{\alpha} \otimes \mathcal{A}'^{\s'}_{\alpha}) \dots]=\nu^{\szsp}_{\alpha+1}$, which is some arbitrary fixed state and $\nu'^{\szsp}_{\alpha}=\mathcal{W}^{\szsp}_{\alpha}[\nu^{\szsp}_{\alpha}]$. This means that any influence from $\mathcal{V}^{\szsp}_{\alpha}$, $\mathcal{K}^{\zs}_{\alpha}$, and $\mathcal{A}'^{\s'}_{\alpha}$ is erased, which is why they does not appear in the second equality. The fact that $\mathcal{K}^{\zs}_{\alpha}$ cannot influence the dynamics is the reason why these three theories have the same free processes. The primed state $\nu'^{\szsp}_{\alpha+1}$ indicates the state after all relevant maps have acted on it (just $\mathcal{W}^{\szsp}_{\alpha+1}$ in this case). The free resources in this theory are temporally uncorrelated processes. In the Choi representation, these free processes take the form of a sequence of unrelated inputs and outputs
\begin{equation}
     \Upsilon^{(\emptyset,\mathcal{K})}_{n:0}= \bigotimes_{j=0}^{n-1} \left(  \rho_{j+1} \otimes \mathbbm{1}_j \right) \otimes \rho_0, \quad \text{for} \ \mathcal{K} \in \{\emptyset,\EB,\Q\}
\end{equation}
where $\rho_{j+1}$ is an arbitrary state equal to $\text{tr}_{\zs}\{ \nu'^{\szsp}_{\alpha+1} \}$. Hence, all temporal correlations have resource value in these theories. 

While these three theories share a set of free resources, it is worth emphasising that they have different free superprocesses and therefore form distinct theories. In other words, different background processes will have different utility in each of them. When the background process is able to pass information in parallel to $\mathcal{C}^{\zsp}_{\alpha+1:\alpha}$, they will react differently to information passed by the process in parallel to $\mathcal{K}^{\zs}_{\alpha}$. For $(\emptyset, \emptyset)$, any communication is still useful, but for $(\emptyset, \EB)$ only quantum communication is valued. In the case of $(\emptyset, \Q)$, nothing further is useful.

\subsubsection{\texorpdfstring{$(\EB, \emptyset)$}{(EB,0)}} \label{sec:eb0}
For this theory, all actions of the experimenter are independent of previous steps, but the experimenter can preserve classical information independently of the background process.

An example of an experiment which can be described by this structure is ghost imaging~\cite{resolutionofquantumandclassicalghostimaging}. In ghost imaging, a pair of correlated subsystems $\s$ and $\an$ are sent into two different situations. $\s$ is sent to interact with an object -- which is analogous to the background process in the resource theory -- and then $\s$ is measured by a simple bucket detector. On the other hand, $\an$ is sent directly to be measured by a pixel array. This is repeated to produce an image of the object. Each iteration of the experimenter's actions is independent, which corresponds to our $\mathcal{K}$-type constraint, and the use of the correlated $\an$ subsystem means that classical information about the initial state of $\s$ is retained irrespective of the interaction with the object.

Now, $\mathcal{K}^\zs$ is a fixed output channel, but $\mathcal{C}^\zsp$ is entanglement breaking (Fig.~\ref{schematics} left column, middle row). The maps connecting adjacent $\mathcal{M}^{{\szsp}}$ in Eq.~\eqref{eq:superprocessprimitive} become fixed output maps whose outputs are $(\mathcal{K}^{\zs}_{\alpha} \otimes {\mathcal{A}'}^{\s'}_{\alpha}) [\mathcal{M}^{{\szsp}}_{\alpha:\alpha-1} \circ \cdots]  = \theta^{\zs}_{\alpha} \otimes \eta^{\s'}_{\alpha}$.

While each step of the process is entanglement breaking, none of the steps depend on each other. This can be expressed in the Choi representation as a tensor product of independent entanglement breaking channels as in Eq.~\eqref{eq:markovprocess}
\begin{equation}
    \Upsilon^{(\EB,\emptyset)}_{n:0}=  \bigotimes_{j=0}^{n-1}\left( \Lambda^\EB_{j+1:j} \right) \otimes \rho_0,
\end{equation}
where $\Lambda^\EB_{j+1:j}$ correspond to the Choi state of $\mathcal{M}^{{\szsp}}_{\alpha+1:\alpha}$. The free processes are Markovian but also entanglement breaking; fully quantum operations within the process are useful, as well as non-Markovianity.

\subsubsection{ \texorpdfstring{$(\Q, \emptyset)$}{(Q,0)}} \label{sec:q0}
In this theory, all actions performed by the experimenter are independent of previous steps, but the experimenter does not need to rely on the background process to preserve information between these actions. 

An example of an experimenter consistent with these constraints is ancilla assisted process tomography~\cite{ancillaassistedprocesstomography}. The system-ancilla joint state prescribed by this method must undergo state tomography, which requires a sequence of identical single-time experiments. In these experiments, one half of a maximally entangled pair $s$ is sent through the process to be characterised, while the other half $a$ is left untouched. The use of this maximally entangled ancilla $a$ ensures that information about the initial state of the system $s$ from each step is retained regardless of the influence of the process over $s$. While the actions of such an experimenter can only produce Markovian dynamics on their own, non-Markovianity might come into play via the process tensor. This would occur when the process being characterised changes as a result of the interaction with $s$. This can be seen as the quantum analogue to the ghost imaging example scenario of Sec.~\ref{sec:eb0}.

Now, $\mathcal{K}^\zs$ is a fixed output channel, but $\mathcal{C}^\zsp$ is fully quantum (Fig.~\ref{schematics} left column, bottom row). The derivation of these free states is identical to the $(\EB,\emptyset)$ case so we will not repeat it. However, since quantum communication is allowed in $\mathcal{C}^{\zsp}_{\alpha+1:\alpha}$, the Choi state of free processes becomes
\begin{equation}
    \Upsilon^{(\Q,\emptyset)}_{n:0}=  \bigotimes_{j=0}^{n-1} \left(\Lambda^\Q_{j+1:j} \right) \otimes \rho_0,
\end{equation}
where $\Lambda^\Q_{j+1:j}$ now correspond to the Choi states of $\mathcal{M}^{{\szsp}}_{\alpha+1:\alpha}$ which are now fully quantum. The free processes are still Markovian; but are now quantum. Using relative entropy for our class of monotones in Thm.~\ref{thm:monotones} produces the multipartite mutual information between all time-steps, suggesting that this is a strict quantum resource theory of non-Markovianity in the sense that the non-free processes are precisely the non-Markovian ones.

\subsubsection{\texorpdfstring{$(\EB, \EB)$}{(EB,EB)} } \label{sec:ebeb}
The scenarios this theory represents are similar to $(\emptyset, \EB)$, but the addition of the $\EB$ form of $\mathcal{C}$-type communication implies that classical information may be preserved arbitrarily well. The ghost imaging example is still pertinent in this resource theory, but here the choice of correlated pairs to send through could be adjusted during the protocol to potentially improve the efficiency of the experiment.

In this case (Fig.~\ref{schematics} centre row, centre column), $\mathcal{C}^\zsp$ and $\mathcal{K}^{\zs}$ are entanglement breaking channels. Here, Eq.~\eqref{eq:superprocessprimitive} resolves to a sum
\begin{equation}
\begin{aligned}
&\llbracket \mathbf{T}^{\text{prim}}_{n:0}{|}\mathbf{Z}^{(\EB,\EB)}_{n:0}|\mathbf{A}'_{n-1:0} \rrbracket =  \notag \\ & \sum_{\{k_\gamma\}} \text{tr}_{\s'} \Bigg\{   {\mathcal{A}'}^{\s'}_{n} \functioncomposition\limits_{\alpha=0}^{n-1} \left( \mathcal{M}^{{\szsp}}_{\alpha+1:\alpha}\circ \big(\nu^\zs_{k_{\alpha}+1} {\rm tr}_{\zs}\Pi^{\zs}_{k_{\alpha+1}} \otimes   {\mathcal{A}'}^{\s'}_{\alpha} \big) \right) \circ \mathcal{W}^{\szsp}_0[\rho^{\szsp}_0] \Bigg\}  
 \end{aligned}
 \end{equation}
where $\Pi^{\zs}_{k_{\alpha}}$ is a POVM and $\nu^\zs_{k_{\alpha}}$ is a re-preparation conditioned by that POVM, forming the entanglement breaking channel of $\mathcal{K}^{\zs}_{\alpha}$. $\{k_\gamma\}$ indexes the possible trajectories based on the set of outcomes of the POVM. If we define CP trace non-increasing maps ${\rm tr}_{\zs}\{\Pi^{\zs}_{k_{\gamma+1}} \mathcal{M}^{{\szsp}}_{\gamma+1:\gamma}[\nu^\zs_{k_{\gamma}}]\} =\mathcal{E}^{\s'}_{k_{\gamma+1},k_{\gamma}}$, which must themselves be entanglement breaking, the expression becomes
\begin{equation}
\begin{aligned}
  \sum_{\{k_\gamma\}} \text{tr}_{\s'} \left\{ \functioncomposition\limits_{\alpha=0}^{n-1} \left(  \mathcal{E}^{\s'}_{k_{\alpha+1}, k_{\alpha}} \circ {\mathcal{A}'}^{\s'}_{\alpha} \right)  \circ \mathcal{W}^{\szsp}_0[\rho^{\szsp}_0] \right\},
\end{aligned}
\end{equation}
such that the process itself (excluding the ${\mathcal{A}'}^{\s'}_{\alpha}$) can be expressed in the Choi representation as
\begin{equation}\label{eq:bb}
    \Upsilon^{(\EB,\EB)}_{n:0}= \sum_{\{k_\gamma\}}    p_{k_{0}} \tensorcomposition\limits_{j=0}^{n-1} \left( \Lambda^{\EB}_{k_{j+1},k_j} \right) \otimes \rho_{k_{0}},
\end{equation}
where ${\{k_\gamma\}}$ enumerates all possible trajectories to sum over, and $p_{k_{0}}$ is the probability of measurement outcome $k$ on the initial state. $\Lambda^{\EB}_{k_{j+1},k_j}$ is the Choi form of an entanglement breaking CP trace non-increasing map $\mathcal{E}^{\s'}_{k_{\alpha+1},k_{\alpha}}$. The full expression for the free processes corresponds to a fully general entanglement breaking process -- Markovian or otherwise. No quantum information can be retained within or between steps. Consequently, quantum entanglement \textit{and} quantum memory are resources in these theories. The set of free processes in Eq.~\eqref{eq:bb} contains all classical non-Markovian processes, see Ref.~\cite{classicalquantumstochasticprocesses} for a detailed discussion.

\subsubsection{\texorpdfstring{$(\Q, \EB)$}{(Q,EB)}} \label{ebq}
This resource theory is much like $(\Q, \emptyset)$, but the actions of the experimenter might be classically correlated between steps. The example of ancilla assisted process tomography could be extended to use an adaptive method. For this agent, the sequence of tomography runs that the experimenter performs need not be pre-determined, but can be chosen adaptively -- perhaps to more efficiently characterise the process with a smaller number of steps.

Here (Fig.~\ref{schematics} bottom row, middle column), $\mathcal{C}^\zsp$ is a fully quantum channel. However, $\mathcal{K}^\zs$ is an entanglement breaking channel. Finding the free processes in this theory is identical to $(\EB,\EB)$ and $(\EB,\Q)$ as in the Sec.~\ref{sec:ebeb}, except that $\mathcal{M}^\szsp$ is now fully quantum rather than entanglement breaking. The result differs accordingly; the Choi state of a free process in $(\Q, \EB)$ can be expressed as
\begin{equation} \label{eq:efm}
    \Upsilon^{(\Q,\EB)}_{n:0}= \sum_{\{k_\gamma\}} p_{k_{0}} \tensorcomposition\limits_{j=0}^{n-1} \left( \Lambda^{\Q}_{k_{j+1},k_j} \right) \otimes \rho_{k_{0}}.
\end{equation}
The difference to the previous cases is that now $\Lambda^{\Q}_{k_{j+1},k_j}$ is the Choi state of a fully quantum CP trace non-increasing map. We define that a process has \textit{entanglement free memory} (EFM) if its Choi state can be written as a convex combination of CP trace non-increasing maps as in Eq.~\eqref{eq:efm}

With a process that has EFM, past dynamics can only influence future dynamics if that past is able to be communicated by an entanglement breaking map. Hence, such a process will appear to be Markovian under control operations that have differing action on entangled states but not on any separable states; they will seem non-Markovian under control operations that do have differing action on separable states. Another way to think of this is that free processes in this theory are non-Markovian with respect to classical information, but Markovian with respect to information that is uniquely quantum~\cite{witnessingquantummemory}. An important subset of these processes are those arising from interactions with a classical system undergoing a stochastic process, such as a noisy magnetic field. In this case, if there is no other environment, the maps with Choi state $\Lambda^{\Q}_{k_{j+1},k_j}$ will be unitary, with trajectories weighted by pre-determined probabilities generated from the underlying classical stochastic process. However, more generally the probability of realising a sequence of $\Lambda^{\Q}_{k_{j+1},k_j}$ in Eq.~\eqref{eq:efm} will depend on the control operations applied. This kind of process lies strictly between quantum Markovianity and non-Markovianity.

\subsubsection{\texorpdfstring{$(\EB, \Q)$}{(EB,Q)}}
For the scenarios described by this theory, an experimenter would be able to perform arbitrary quantum maps, but could not guarantee that the quantum information those maps send out is preserved. Again, the ghost imaging example can be extended to describe this resource theory. The difference here is that the bucket detector is removed, and the output of $\s$ would be used as the input for the next step. Since $\an$ is measured at the end of each step, a new $\an$ would need to be correlated with this output at the beginning of each step. Hence, the quantum memory only lasts a single step.

In this case (Fig.~\ref{schematics} centre row, right column), $\mathcal{C}^\zsp$ is an entanglement breaking channel, hence $\mathcal{M}^{\szsp}$ will also be any entanglement breaking channel. However, $\mathcal{K}^{\zs}_{\alpha}$ is fully quantum. Since $\mathcal{C}^\zsp_{\alpha+1:\alpha}$ is an entanglement breaking channel, its action can be expressed as $\mathcal{C}^\zsp_{\alpha+1:\alpha}[\rho^\zsp] = \sum_{k_{\alpha+1:\alpha}} \nu^\zsp_{k_{\alpha+1:\alpha}} {\rm tr}\{\Pi^\zsp_{k_{\alpha+1:\alpha}} \rho^\zsp\}$, where $\Pi^{\zsp}_{k_{\alpha+1:\alpha}}$ is a POVM and $\nu^{\zsp}_{k_{\alpha+1:\alpha}}$ is a re-preparation conditioned by that POVM. Using this, we can express the action of a segment of the transformed process tensor as $\mathcal{M}^\szsp_{\alpha+1:\alpha}[\rho^\szsp] = \sum_{k_{\alpha+1:\alpha}} \mu^\szsp_{k_{\alpha+1:\alpha}} {\rm tr}\{\Omega^\szsp_{k_{\alpha+1:\alpha}} \rho^\szsp \}$. $\mu^\szsp_{k_{\alpha+1:\alpha}} = \mathcal{W}^{\szsp}_{\alpha+1}[\mathbbm{1}^\s\otimes \nu^\zsp_{k_{\alpha+1:\alpha}}]$ is the output state of $\mathcal{W}^{\szsp}_{\alpha+1}$ and $ {\rm tr}\{\Omega^\szsp_{k_{\alpha+1:\alpha}} \rho^\szsp \} =  {\rm tr}\{\Pi^\zsp_{k_{\alpha+1:\alpha}}\mathcal{V}^\szsp[\rho^\szsp] \}$ is the probability of the $k_{\alpha+1:\alpha}$ outcome of the (modified) POVM corresponding to $\Omega^\szsp_{k_{\alpha+1:\alpha}}$. This measurement and re-preparation is the only constraint of such an experimenter. For $\mathbf{Z}^{(\EB,\Q)}_{n:0}$, the right hand side of Eq.~\eqref{eq:superprocessprimitive}, in terms of the above objects, yields a sum of measurements and re-preparations
\begin{equation}
\sum_{\{k_{\gamma+1:\gamma}\}} \text{tr}_{\zsp} \left\{\mu^\szsp_{k_{n:n-1}} \right\} \prod_{\alpha = 1}^{n-1} (p_{k_{\alpha+1:\alpha},k_{\alpha:\alpha-1}})  p_{k_{1:0}}, 
\end{equation}
where $p_{k_{1:0}}=\text{tr}_{\szsp}\left\{\Omega^\szsp_{k_{1:0}}(\mathcal{K}^{\zs}_{0} \otimes {\mathcal{A}'}^{\s'}_{0})\circ \mathcal{W}^\szsp_0[\rho^\szsp_0] \right\}$ represents the probability of an initial evolution on trajectory $k$, while $p_{k_{\alpha+1:\alpha},k_{\alpha:\alpha-1}}=\text{tr}_{\szsp} \left\{\Omega^\szsp_{k_{\alpha+1:\alpha}}(\mathcal{K}^{\zs}_{\alpha} \otimes {\mathcal{A}'}^{\s'}_{\alpha})[\mu^\szsp_{k_{\alpha:\alpha-1}}] \right\}$ represents the probability of evolution on trajectory $k$ between steps $\alpha-1$ and $\alpha+1$. Consequently, the Choi state of a free process tensor in this theory has the form
 \begin{equation}
    \Upsilon^{(\EB,\Q)}_{n:0}= \sum_{\{k_{\gamma+1:\gamma}\}} \rho_{k_{n:n-1}} \tensorcomposition\limits_{j=1}^{n-1} \left( \Gamma^{\Q}_{k_{j+1:j},k_{j:j-1}} \right)\otimes \Gamma^{\Q}_{k_{1:0}}.
\end{equation}
Here, as in the $(\Q,\EB)$ theory, the $\Gamma^{\Q}$'s are Choi states of general CP trace non-increasing maps. However, here they have the same input and output spaces as the $\mathcal{A}^{\s'}_{\alpha}$ control operations, meaning that they cannot be interpreted in terms of trajectories of interleaved $\mathcal{E}^{\s'}$ and $\mathcal{A}^{\s'}$ maps. Instead, each $\Gamma^{\Q}_{k_{j+1:i},k_{j:j-1}}$ implies a control operation dependent probability, and the full sequence determines a distribution over final states $\rho_{k_{n:n-1}}$. In this sense there is full quantum memory between adjacent legs of the free processes, but due to the entanglement breaking nature of $\mathcal{C}$, it cannot propagate more than one step. For this reason we call this kind of memory \textit{single-step quantum memory (SSQM)}.

\subsubsection{\texorpdfstring{$(\Q, \Q)$}{(Q,Q)}}
This resource theory describes an `all-powerful' experimenter who can perform any sequence of actions, with any correlations between them. In this case (Fig.~\ref{schematics} bottom row, right column), $\mathcal{C}^\zsp$ and $\mathcal{K}^\zs$ are quantum channels. Since there are no restrictions on $\mathcal{C}^\zsp$ and $\mathcal{K}^\zs$, anything of the form of Eq.~\eqref{eq:superprocessprimitive} can be achieved for free. This resource theory is trivial since every possible process in the theory is a free process; the Choi state $\Upsilon^{(\Q, \Q)}_{n:0}$ can be that of any quantum process. Hence, nothing can be considered especially useful.

\subsection{Summary of Results} \label{sec:summaryofresults}

As we are particularly interested in non-Markovianity, which is a property of memory, we define classical memory $c$ and quantum memory $q$ as
\begin{equation}
\begin{aligned}
(c,q) := \begin{tabular}{l}\text{maximum number of steps retaining} \\ \text{(classical,quantum) information.}\end{tabular}
\end{aligned}
\end{equation}
The memory can take values $\{0,1,\infty\}$, which mean no information can be retained in time; information can be retained for at most one step; and information can be retained for the whole duration of the process. These quantities are formally related to the notion of Markov order, which has recently been generalised for quantum processes; there, the choice of instrument used for a measurement is taken into account in addition to the measurement outcomes~\cite{quantummarkovorder, finitemarkovorder}.

\begin{table*}[htbp]
\centering
\caption {Free Processes}
\begin{tabular}{cclll}
\multicolumn{2}{c}{} &
\multicolumn{3}{c}{$\mathcal{K}$}  \\
\cline{3-5}
\multicolumn{2}{c}{} &
\multicolumn{1}{c}{$\emptyset$} & 
\multicolumn{1}{c}{$\EB$} & 
\multicolumn{1}{c}{$\Q$} \\ 
\cline{3-5}
\multicolumn{1}{c|}{} & 
\multicolumn{1}{c}{$\emptyset$} & 
\multicolumn{1}{|l|}{Fixed output, $(0,0)$ }         & 
\multicolumn{1}{l|}{Fixed output, $(0,0)$ }   & 
\multicolumn{1}{l|}{Fixed output, $(0,0)$ }  \\ 
\cline{3-5} 
\multicolumn{1}{c|}{$\mathcal{C}$} &
\multicolumn{1}{c}{$\EB$}       & 
\multicolumn{1}{|l|}{Markovian EB, $(1,0)$}         & 
\multicolumn{1}{l|}{EB, $(\infty,0)$}   & 
\multicolumn{1}{l|}{SSQM, $(\infty,1)$}  \\
\cline{3-5} 
\multicolumn{1}{c|}{}          & 
\multicolumn{1}{c}{$\Q$}        & 
\multicolumn{1}{|l|}{Markovian, $(1,1)$}         & 
\multicolumn{1}{l|}{EFM, $(\infty,1)$}   & 
\multicolumn{1}{l|}{CPTP, $(\infty,\infty)$}  \\ 
\cline{3-5}\\
\end{tabular}
\caption{The sets of free resources for the nine resource theories. Each row represents a different form of $\mathcal{C}$, while each column represents a different form of $\mathcal{K}$. The classical $c$ and quantum $q$ memory length of free processes is assigned as $(c,q)$ for the free processes of each case. A memory length of $0$ implies that no information can be preserved temporally, while a memory length of $1$ means that adjacent steps are able to depend on each other. A memory length of $\infty$ denotes the case where no restrictions exist on the memory. In the top row, the inability of the allowed superprocesses to carry information through time renders the processes constructed from them no more useful than the primitive free process itself. The left column contains theories where the free processes are all Markovian but with more control going from top to bottom. The diagonal contains the most general possible processes subject to each of the three types of constraints; either fixed output, entanglement breaking (EB), or completely positive and trace preserving (CPTP). $(\Q,\EB)$ and $(\EB,\Q)$ are the most atypical theories, as they have free processes which lie in a grey area between quantum and classical non-Markovianity. Their free procecesses have entanglement free memory (EFM), and single-step quantum memory (SSQM) respectively.}
\label{ninecases}
\end{table*}

In all $(\emptyset,-)$ theories, the free processes are only those which have no temporal correlations. However, once some temporal correlations are present, each of these theories differ, as any $\mathcal{C}$-type communication can be better utilised if $\mathcal{K}$-type communication is allowed. In the $(\EB,\emptyset)$ and $(\Q,\emptyset)$ theories, all of the free processes are Markovian, although in the former case only classical operations are allowed for free. On the other hand, $(\Q,\emptyset)$ is a true resource theory of quantum non-Markovianity: there is a one-to-one correspondence between free-ness and Markovianity. The free processes in $(\EB,\EB)$ correspond to agents who can carry out any multi-time entanglement breaking process for free, while every quantum process is free in $(\Q,\Q)$. The free processes of $(\Q,\EB)$ have entanglement free memory, and $(\EB,\Q)$ has single step quantum memory. The length and quality of memory determines the type of process. In particular, classical non-Markovian processes are a strict subset of $(\EB,\EB)$, which in turn is a subset of $(\EB,\Q), \ (\Q,\EB), \ \text{and}, \ (\Q,\Q)$. The classes of free processes and the memory lengths are summarised in Table.~\ref{ninecases}. 

In our hierarchy, generally the Choi states of free processes satisfy some form of independence, while useful ones have entanglement or other correlations, implying that these monotones will take typically the form of multipartite correlation measures. In the simplest $( \emptyset, -)$ cases, the $M$ monotones of Thm.~\ref{thm:monotones} are simply the $D(\cdot,\cdot)$ measure applied between the Choi state of interest and that of the primitive free resource, e.g. the relative entropy to the maximally mixed state. In the opposite extreme $(\Q, \Q)$, the monotone is always zero since every process is free. Of particular importance, the monotone for $(\Q, \emptyset)$ is a direct measure of the non-Markovianity of a process tensor as quantified in Ref.~\cite{nonmarkov}. Furthermore, the monotone of $(\Q, \EB)$ is a measure of a stricter notion of uniquely quantum non-Markovianity. The properties of utility for each resource theory are summarised in Table.~\ref{ninecasesresources}.

\begin{table*}[htbp]
\centering
\caption {Properties of Utility}
\begin{tabular}{cclll}
\multicolumn{2}{c}{} &
\multicolumn{3}{c}{$\mathcal{K}$}  \\
\cline{3-5}
\multicolumn{2}{c}{} &
\multicolumn{1}{c}{$\emptyset$} & 
\multicolumn{1}{c}{$\EB$} & 
\multicolumn{1}{c}{$\Q$} \\ 
\cline{3-5}
\multicolumn{1}{c|}{} & 
\multicolumn{1}{c}{$\emptyset$} & 
\multicolumn{1}{|l|}{Correlations }         & 
\multicolumn{1}{l|}{Correlations }   & 
\multicolumn{1}{l|}{Correlations }  \\ 
\cline{3-5} 
\multicolumn{1}{c|}{$\mathcal{C}$} &
\multicolumn{1}{c}{$\EB$}       & 
\multicolumn{1}{|l|}{Inter-$\mathcal{E}$ corr., intra-$\mathcal{E}$ enta.}         & 
\multicolumn{1}{l|}{Entanglement}   & 
\multicolumn{1}{l|}{Intra-$\mathcal{E}$ entanglement}  \\
\cline{3-5} 
\multicolumn{1}{c|}{}          & 
\multicolumn{1}{c}{$\Q$}        & 
\multicolumn{1}{|l|}{Inter-$\mathcal{E}$ correlations}         & 
\multicolumn{1}{l|}{Inter-$\mathcal{E}$ entanglement}   & 
\multicolumn{1}{l|}{Nothing}  \\ 
\cline{3-5}\\
\end{tabular}
\caption{The properties of utility which may be exhibited by resources in each of the nine resource theories. These are expressed as properties of the Choi state of a process tensor, and can be measured with monotones as in Sec.~\ref{monotones}. In the top row, any correlations in the Choi state are useful, although differently useful in each theory. In every entry on the left column, correlations between individual partitions of the Choi state corresponding to legs of a process tensor are useful (memory between steps). In the top two entries of this column there are other properties which are also useful. However, in the bottom case only non-Markovianity is a resource, rendering it a true resource theory of quantum non-Markovianity. In the diagonal entries, any violations of the specified constraint is seen as a resource. In $(\Q,\EB)$ and $(\EB,\Q)$, only non-Markovianity that is explicitly quantum in nature can be useful.}
\label{ninecasesresources}
\end{table*}

\subsection{Combining Resources} \label{sec:combiningresources} 
\subsubsection{Convex Mixing}
There are at least three distinct notions where one could think about combining process tensor resources. The first would be the simple addition or convex `mixing' of process tensors. For two process tensors $\mathbf{T}$ and $\mathbf{S}$, another valid process tensor (in Choi form) would be $(\Upsilon_{\mathbf{T}}+r\Upsilon_{\mathbf{S}})/(1+r)$ for some $r \in \mathbb{R}_{\geq 0}$. For the five theories where either $\mathcal{C}$ or $\mathcal{K}$ forbid communication, the corresponding sets of free resources are not convex; i.e., combining resources in this way is not free. In contrast, since the set of separable states is convex, the other four resource theories in our hierarchy will have convex sets of free resources. This seems to suggest that some amount of temporal correlations would be the norm for multi-time processes, but that specifically quantum features are distinctly rarer. Numerical studies using our robustness measures may be able to add evidence to this claim.

\subsubsection{Sequential Composition} \label{sec:sequentialcomposition}
A second notion of combining process tensors applies for pairs of process tensors where at least one includes an input Hilbert space for $\s$, rather than an initial state for $\s$. These processes can be composed with each other with a new slot for an additional control operation between the final free system-environment evolution of the first process tensor, and the first free system-environment evolution of the second process tensor. The number of steps in the resultant process tensor is the sum of the number of steps for the constituents. To ensure this operation is resource non-generating, we assume the agent performing superprocess has the same capability over all steps. Similarly, the fact that the environment state is discarded at the join implies that the act of joining cannot add resource value. Furthermore, for the $(-,\emptyset)$ theories, this notion of composition reduces to the tensor product of process tensor Choi states $\Upsilon_{\mathbf{T}} \otimes \Upsilon_{\mathbf{S}}$, since the inability of the agent to carry information between steps means that these can be treated as totally independent experiments. This is not true in general, as temporal order may matter to the experimenter.

\subsubsection{Link Product}
A third notion of combining resources is to add steps to link the final system and environment states of one process tensor to another. This would be like adding steps to an existing process tensor, and can be achieved via the link product~\cite{quantumnetworks} $\mathbf{T} * \mathbf{S}$. This is different to the previous notion of combining resources because the state of the environment is not destroyed by a partial trace at the end of the initial process tensor; instead the final environment state is connected to the input of the next step. The link product is also versatile enough to account for the notion of parallel composition of multiple $s$ subsystems. However, one must be careful when using the applying the link product to resource theories, since the result may not still be free. Similarly, shortening the duration of a process tensor by removing steps is a valid operation~\cite{quantumnetworks}, but is not free in general. 

\subsubsection{Adaptivity Considerations}
In many contexts it may be possible to sidestep these additional considerations by simply pre-defining the process tensor to have more steps than is needed, and setting some of the experimenter's operations be the identity operation. This allows the number of non-trivial operations to be flexible -- and possibly even chosen adaptively -- while the total number of operations remains the same. In practice, the actual number of steps a process tensor has would be determined by the minimal time resolution of the experiment, along with the maximum length of the experiment. The generalised Kolmogorov extension theorem~\cite{kolmogorov} guarantees the existence of an infinitely fine-grained underlying process tensor that can be interacted with at any possible sequence of times.

So far, we have not specified whether the superprocess is deterministic or stochastic. This is because our formalism allows for either possibility. Stochastic strategies play an important role in many information-theoretic protocols. Moreover, there is an important distinction between arbitrary processes of any finite duration, and those of truly unbounded duration, as previously demonstrated in the case of LOCC~\cite{everythingyoualwayswantedtoknow}. For example, one may seek to conditionally end an experiment only after achieving a particular outcome -- which may not be explicitly guaranteed even in the case where the probability of the desired outcome approaches one as time goes on. This generalised Kolmogorov extension theorem also guarantees the existence of such infinite \emph{duration} process tensors.

However, clearly the guaranteed existence of a process tensor says nothing about its value in any given resource theory. As such, further work would be required to make definitive statements about the non-trivial question of how resource value is affected when resources are allowed to become temporally infinite. Establishing notions of composition which are amenable to the theory in question are a logical starting point for such an investigation.

\section{What do resource theories of non-Markovianity mean?}
Resource theories arise from preorders, which in turn are defined by convertibility under allowed resource transformations~\cite{mathres}. As such, a useful resource can be seen as one which can be converted into another which the agent cannot obtain otherwise. For example, athermal states are required to produce work in thermodynamic resource theories. Similarly, in a theory where Markovian process tensors cannot be transformed via allowed superprocesses into a desired process tensor, but non-Markovian ones can be, the latter kind of process tensor will be seen as useful.

This notion for the utility of non-Markovianity takes advantage of the fact that process tensors naturally induce a strict, operationally well motivated definition for Markovianity -- taking into account the influence of (brief) experimenter interventions, and not suffering from the initial correlation problem~\cite{introqod}. Importantly, it reduces to the classical notion in the appropriate limit (which can be obtained in our framework by restricting all tensors to preserve a fixed basis),  setting our theories apart from those based on other notions of non-Markovianity~\cite{nm2}. Since CP divisibility does not imply Markovianity~\cite{divisibility} in the sense of multi-time correlations, a resource theory of CP divisibility is not equivalent to what we have investigated. The properties each would consider to be useful are distinct: some processes which are free in a resource theory of CP divisibility would not be free in a resource theory of multi-time correlations. There may be temporal correlations in CP divisible processes that could be put to use. Likewise, since a process being CP divisible implies that it has no information backflow, a resource theory based on this notion would also not characterise the utility of all multi-time correlations.

\subsection{When is non-Markovianity useful? When is it useless?}
The next question to ask is under which physical circumstances might non-Markovianity become useful. Our results in Sec.~\ref{Results} split resource theories of multi-time processes into a handful of broad categories, leaving room for the details of the resource theory to be further specified. In these general terms, the conditions which make non-Markovianity useful are clear -- when the agent of the resource theory cannot carry information about the main subsystem from one evolution under the process tensor to another via an ancillary subsystem, forgetting the state of the main system will erase correlations between the future and the past unless the background process is non-Markovian. Non-Markovian background processes can enhance the transference of information through the process, but will only be deemed useful if the agent is not already capable of storing their own external memory of the main subsystem. Furthermore, a resource theory where non-Markovianity is the \textit{sole} resource is one where the agent can carry an arbitrarily faithful memory between control operations, but each control operation is destructive to that memory.

In our example of ancilla assisted process tomography for $(\Q,\emptyset)$, it is possible to see how non-Markovianity can become useful. An advantage of ancilla assisted process tomography is that it only requires the preparation of one particular system-ancilla joint state~\cite{ancillaassistedprocesstomography}. The multiple `steps' of this tomography method are those required for state tomography on this joint state -- which ideally would be prepared identically for every step. As one gathers the necessary number of samples to estimate the expectation values of each basis operator, they may encounter statistically anomalous measurement outcomes for a number of reasons. It may simply be a statistical fluke, or there may be a measurement error, or the joint state may have been incorrectly prepared. If the anomalous result is caused by the latter, non-Markovian processes can become advantageous in handling the associated error. This is because for non-Markovian processes this type of error in particular might ripple through into subsequent steps in the form of correlated measurement outcomes, providing a signature that they are not simply statistical flukes, and should be rejected. This type of error mitigation would be impossible for Markovian processes -- those which have a zero value for every appropriate monotone. One can think of this scenario a kind of temporal analogy to secret sharing~\cite{secretsharing}, where a minimum number of correlated time-steps are required to decipher the secret, rather than a minimal number of parties.

Monotones can easily decide when a task is impossible, but perhaps the most important unanswered question in this discussion is whether there exist non-Markovianity monotones which predict when a particular task \emph{can} be completed. We hope that such monotones can be devised for the $(\Q,\emptyset)$ theory.

$(-,\emptyset)$ theories all have non-Markovianity as a resource, and it is the sole resource in $(\Q,\emptyset)$. However, these are far from the only theories which might have non-Markovianity as a resource. Firstly, there are theories where some non-Markovian background processes are useful but not others, as seen in $(\Q,\EB)$. Furthermore, there exist theories which lie between the categories we have studied. For example, relaxing the condition of $\emptyset$ from `no transmission of information on the system or ancilla' to just `no transmission of information on the ancilla', would yield a theory with an agent which is more capable than that of $\emptyset$ but less powerful than that of $\Q$. This agent cannot carry an external memory, but their operations need not be destructive to the main subsystem; we believe this set of constraints will be relevant to many quantum control tasks. Moving toward even greater specificity, one could restrict the agent to performing only a subset of operations within the overall set induced by $\EB$ or $\Q$ -- perhaps by restricting the allowed Hamiltonians and ancilla based on experimental considerations. These more specific theories will bring resource theories of multi-time processes, including those where non-Markovianity is a resource, towards the realm of direct physical applications.

\subsection{Relation to Other Resource Theories}
Our work here is based on a general framework for quantum processes which can be probed at multiple times. As such, we expect that many `resource theories of processes' (see Ref.~\cite{review}), most of which take individual maps as resources, can be extended to encompass multi-step scenarios using our framework; this include many important recent developments in resource theories of channels~\cite{operationalresourcetheoryofquantumchannels,resourcetheoriesofquantumchannels,quantifyingoperationswithanapplication,inputoutputgames}. Another completely distinct resource theoretic description of channels utilises discrimination tasks~\cite{generalresourcetheoriesinquantummechanics}. This approach elucidates the connection between channels being non-free and being useful for a particular task. This could be extended to the multi-time case with discrimination tasks for process tensors. Resource theoretic results on the entanglement of bipartite channels~\cite{entanglementofbipartitechannels,resourcetheoryofentanglementforbipartite,fundamentallimitsonthecapacities,bipartitequantuminteractions,entanglementandsecretkey} have utilised the idea of quantum strategies~\cite{quantumstrategies,towardageneraltheoryofquantumgames,onameasureofdistanceforquantumstrategies,fidelityofquantumstrategieswithapplications}, which allows for the possibility of dynamical resources being quantum combs. Discrimination of quantum strategies~\cite{resourcetheoryofasymmetricdistinguishability} has also been investigated, which involves transformations from strategies to strategies; this is in a similar spirit to our approach of transforming process tensors with superprocesses. 

There has been a great deal of interest in the utility of non-Markovianity, giving rise a plethora of different resource theories involving non-Markovianity, as well as other related properties. The resource theory described in~\cite{theoryofmaps} captures similar behaviour to that of the $(-,\EB)$ resource theories. There have also been resource theories of divisible operations~\cite{rtnmdiv, robustness}, which have provided numerous results which are related to the $(\Q,\emptyset)$ theory in our hierarchy. This theory is formulated in terms of parameterised families of quantum maps (from density operators to density operators) rather than process tensors, which are experimentally and computationally simpler than a process tensor approach, but cannot fully account for multi-time correlations. A process matrix approach has been used to quantify the capability of processes to produce uniquely quantum effects, which has also led to a measure of non-Markovianity~\cite{quantifyingquantummechanical,quantumprocesscapability,nonmarkovianityofphotondynamics}. This method has also been used to study non-Markovian effects in real experimental data~\cite{nonmarkovianityofphotondynamics}. Finally, another class related theories non-Markovianity are explored in Refs.~\cite{tripartite, communicationcostfornonmarkovianity}; here the Markov condition is derives from the conditional quantum mutual information between subsystems of multipartite states, as opposed to explicit temporal correlations. See Refs~\cite{finitemarkovorder, quantummarkovorder} for relation between Markov chain states and quantum processes.

\section{Conclusion} \label{sec:conclusion}
At the heart of this investigation is the question: \textit{can an uncontrolled background process useful be useful to an agent?} In order to answer this question, we have presented a framework for resource theories of quantum processes. Here, the descriptor of the process, known as the process tensor, takes the role of resources, and a new construction -- the superprocess -- serves the purpose of transforming resources. Under this framework, a background process that can be simulated by agent actions on a pre-defined free process holds no value, while one that is not producible may have the potential to be used to perform tasks which were previously unavailable to the agent. We have used this framework to construct an operationally motivated hierarchy of nine theories corresponding to a realistic client-contractor-server scenario, and found the associated free processes and monotones. In many of these theories, notions of non-Markovianity are the main determinant of the value of a process. Futhermore, $(\Q,\emptyset)$ is a true theory of quantum non-Markovianity. The $(\Q,\EB)$ and $(\EB,\Q)$ free processes exhibited properties lying between classical and quantum non-Markovianity, implying that the processes of value exhibit a kind of non-Markovianity which is uniquely quantum. 

While the background process is noisy, we assume that any map contained within the superprocess can be performed perfectly. We did limit control by employing our three classes of communication, but the effectiveness of the agent at implementing the allowed processes remains unlimited. Thus, a promising direction for future work is looking into theories which have more stringent restrictions. For example, $\mathcal{V}$ and $\mathcal{W}$ could be subject to their own channel resource theories. Furthermore, recent work on resource theories of measurements~\cite{resourcetheoryofquantummeasurements} could enable our theories whic have the $\EB$ class of communication to be further dissected.

There are numerous physical settings which might make use of a resource theory of multi-time processes. For example, the client-contractor-server scenario detailed in Sec.~\ref{communicationtheories} might have applications in cloud quantum computing, where quantum memory may be restricted on the server side. Superprocesses could also be used to model untrusted devices, allowing resource-theoretic investigations within the device independent paradigm, which has recently been extended to include multi-time causal processes~\cite{quantumviolationsofaninstrumentaltest}. More generally, superprocesses provide a natural framework for investigating quantum control problems in a scenario where multiple interventions on a system are possible.

We hope that our formalism may enable known results to be applied in new contexts. One might seek construct a new multi-time process theory from state or channel theories which are already well understood. Furthermore, multi-time process resource theories are yet to be studied in the same kind of depth as state or channel theories. Properties such as resource distillation and dilution, and single shot vs asymptotic transformations are uncharted territory. We present an open invitation to sail the seas of Hilbert space using uncontrolled background processes. 

\subsubsection*{Acknowledgements}
\textit{G. B. is supported by an Australian Government Research Training Program (RTP) Scholarship. This project was made possible through the support of a grant from the John Templeton Foundation. The opinions expressed in this publication are those of the authors and do not necessarily reflect the views of the John Templeton Foundation. B. Y. acknowledges funding from the European Research Council (ERC) under the Starting Grant GQCOP (Grant no. 637352). K. M. is supported through the Australian Research Council Future Fellowship FT160100073.}

\bibliographystyle{apsrev4-1_custom}
\bibliography{refs}

\appendix \label{appendix}

\section{Superprocesses in the Choi Representation} \label{appchoi}
The Choi state of an $n \rightarrow n$ superprocess is given by
\begin{equation}
\begin{aligned}
\Psi_{n:0}= \text{tr}_\zsp \Bigg\{   \Big(  \functioncomposition\limits_{j=0}^{n-1} \big(  S^{\s,{\mathtt{i}'_{j+1}}} \circ \mathcal{W}^{\szsp}_{j+1} \circ S^{\s,{\mathtt{i}_j}} \circ \mathcal{V}^{\szsp}_{j}  \big)  \circ S^{\s,{\mathtt{i}'_0}} \circ \mathcal{W}^{\szsp}_{0} \Big)  \\
\Big(  \tensorcomposition\limits_{j=0}^{n-1}   \big(    \psi^{{\mathtt{i}'_{j+1}},{\mathtt{o}_{j+1}}} \otimes \psi^{{\mathtt{i}_{j}},{\mathtt{o}'_{j}}} \big) \otimes \psi^{\s,{\mathtt{o}_0}} \otimes  \rho_0^{\zsp} \Big)\Bigg\}.
\end{aligned}
\end{equation}
Our index system for Hilbert spaces (illustrated in Fig.~\ref{choistateindices}) has three variables. The first is the time-step, denoted by the subscript. Additionally, the unprimed indices indicate that the Hilbert spaces are shared with the process tensor, while the primed indices correspond to Hilbert spaces shared with the control sequence. Finally, each Hilbert space is labelled by whether it corresponds to an input $\mathtt{i}$ or output $\mathtt{o}$ of the respective object that the superprocess connects to on that index.

Summing over like indices from our aforementioned expression for the Choi state of the process tensor (Eq.~\eqref{processchoi}) corresponds to the action of the superprocess on the process tensor to form a new effective background process. The action of an $n \rightarrow n$ superprocess on a length $n$ process tensor is:
\begin{equation}
    \llbracket \mathbf{T}_{n:0}|\mathbf{Z}_{n:0}=\text{tr}_{\mathtt{i},\mathtt{o}} \{ (  \mathbbm{1}^{\mathtt{i}',\mathtt{o}'} \otimes \Upsilon^\text{T}_{n:0} ) \Psi_{n:0}    \},
\end{equation}
where the lack of time-step subscripts indicates that we are summing for all time-steps. $\text{T}$ indicates the partial transpose of the Choi representation of the process tensor.

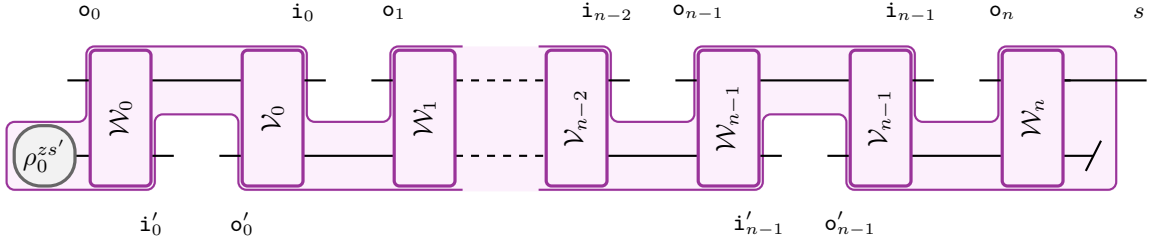
\begin{figure*}
\centering

\begin{tikzpicture}[scale=1]

\draw[mypurplefill,fill=mypurplefill, very thick,solid,rounded corners=3] (6,0.95)  -- (5.05,0.95) -- (5.05,-0.05) -- (3.95,-0.05) -- (3.95,0.95) -- (3.05,0.95) -- (1.05,0.95) -- (1.05,-0.05)  -- (0,-0.05) -- (0,-0.95) -- (1.95,-0.95) -- (1.95,0.05) -- (3.05,0.05) -- (3.05,-0.95) -- (6,-0.95) -- (7,-0.95)--  (9.95,-0.95) -- (9.95,0.05) -- (11.05,0.05) -- (11.05,-0.95) -- (14.6,-0.95) -- (14.6,0.05) -- (14.6,0.95) -- (13.05,0.95) -- (13.05,-0.05) -- (11.95,-0.05) -- (11.95,0.95) -- (11.05,0.95) --  (9.05,0.95) -- (9.05,-0.05) -- (7.95,-0.05) -- (7.95,0.95) -- (7,0.95)  ;

\draw[mypurple, thick,solid,rounded corners=3] (6,0.95)  -- (5.05,0.95) -- (5.05,-0.05) -- (3.95,-0.05) -- (3.95,0.95) -- (3.05,0.95) -- (1.05,0.95) -- (1.05,-0.05)  -- (0,-0.05) -- (0,-0.95) -- (1.95,-0.95) -- (1.95,0.05) -- (3.05,0.05) -- (3.05,-0.95) -- (6,-0.95) ;

 \draw[mypurple,thick,solid,rounded corners=3] (7,-0.95)--  (9.95,-0.95) -- (9.95,0.05) -- (11.05,0.05) -- (11.05,-0.95) -- (14.6,-0.95) -- (14.6,0.05) -- (14.6,0.95) -- (13.05,0.95) -- (13.05,-0.05) -- (11.95,-0.05) -- (11.95,0.95) -- (11.05,0.95) --  (9.05,0.95) -- (9.05,-0.05) -- (7.95,-0.05) -- (7.95,0.95) -- (7,0.95)  ;

\draw[black,  thick,solid] (13.9,0.5) -- (15,0.5);
\draw[black,  thick,dashed] (5.9,0.5) -- (7.1,0.5);
\draw[black,  thick,dashed] (5.9,-0.5) -- (7.1,-0.5);
\draw[black,  thick,solid] (14.2,-0.7) -- (14.4,-0.3);

\draw[black,  thick,solid] (0.8,0.5) -- (4.2,0.5);
\draw[black,  thick,solid] (4.8,0.5) -- (5.2,0.5);
\draw[black,  thick,solid] (7.2,0.5) -- (8.2,0.5);
\draw[black,  thick,solid] (8.8,0.5) -- (12.2,0.5);
\draw[black,  thick,solid] (12.8,0.5) -- (14,0.5);

\draw[black,  thick,solid] (0.8,-0.5) -- (2.2,-0.5);
\draw[black,  thick,solid] (2.8,-0.5) -- (5.2,-0.5);
\draw[black,  thick,solid] (7.2,-0.5) -- (10.2,-0.5);
\draw[black,  thick,solid] (10.8,-0.5) -- (14.3,-0.5);

\draw[mypurple,fill=mypurplefill,very thick,solid,rounded corners=2] (1.1,0.9) rectangle (1.9,-0.9);
\draw[mypurple,fill=mypurplefill,very thick,solid,rounded corners=2] (3.1,0.9) rectangle (3.9,-0.9);
\draw[mypurple,fill=mypurplefill,very thick,solid,rounded corners=2] (5.1,0.9) rectangle (5.9,-0.9);
\draw[mypurple,fill=mypurplefill,very thick,solid,rounded corners=2] (7.1,0.9) rectangle (7.9,-0.9);
\draw[mypurple,fill=mypurplefill,very thick,solid,rounded corners=2] (9.1,0.9) rectangle (9.9,-0.9);
\draw[mypurple,fill=mypurplefill,very thick,solid,rounded corners=2] (11.1,0.9) rectangle (11.9,-0.9);
\draw[mypurple,fill=mypurplefill,very thick,solid,rounded corners=2] (13.1,0.9) rectangle (13.9,-0.9);
\draw[mygrey,fill=mygreyfill,very thick,solid,rounded corners=10] (0.1,-0.1) rectangle (0.9,-0.9);

\draw[] (1.1,1.4) node {\footnotesize  ${\mathtt{o}_0}$};
\draw[] (3.9,1.4) node {\footnotesize  ${\mathtt{i}_0}$};
\draw[] (5.1,1.4) node {\footnotesize  ${\mathtt{o}_1}$};
\draw[] (7.9,1.4) node {\footnotesize  ${\mathtt{i}_{n-2}}$};
\draw[] (9.1 ,1.4) node {\footnotesize  ${\mathtt{o}_{n-1}}$};
\draw[] (11.9,1.4) node {\footnotesize  ${\mathtt{i}_{n-1}}$};
\draw[] (13.1,1.4) node {\footnotesize  ${\mathtt{o}_{n}}$};
\draw[] (14.9,1.4) node {\footnotesize  $\s$};

\draw[] (1.9,-1.4) node {\footnotesize  ${\mathtt{i}'_0}$};
\draw[] (3.1,-1.4) node {\footnotesize  ${\mathtt{o}'_0}$};
\draw[] (9.9,-1.4) node {\footnotesize  ${\mathtt{i}'_{n-1}}$};
\draw[] (11.1,-1.4) node {\footnotesize  ${\mathtt{o}'_{n-1}}$};

\draw[] (0.5,-0.5) node {\small $\rho_{0}^{\zsp}$};
\draw[] (1.5,0) node[rotate=90] {\small  $\mathcal{W}_{0}$};
\draw[] (3.5,0) node[rotate=90] {\small  $\mathcal{V}_{0}$};
\draw[] (5.5,0) node[rotate=90] {\small  $\mathcal{W}_{1}$};
\draw[] (7.5,0
) node[rotate=90] {\small  $\mathcal{V}_{n-2}$};
\draw[] (9.5,0) node[rotate=90] {\small  $\mathcal{W}_{n-1}$};
\draw[] (11.5,0) node[rotate=90] {\small  $\mathcal{V}_{n-1}$};
\draw[] (13.5,0) node[rotate=90] {\small  $\mathcal{W}_{n}$};

\end{tikzpicture}

\caption{Hilbert spaces of the superprocess. The top Hilbert spaces connect with the process tensor (unprimed), while the bottom ones connect to the control sequence (primed). The spaces are also labelled by whether they are incoming $\mathtt{i}$ or outgoing $\mathtt{o}$ from the perspective of the object connecting to the superprocess. There is also a subscript label for the time-step. s is the final output of the process.} \label{choistateindices}
\end{figure*}

\section{Intra-Step Behaviour} \label{appsteps}
Here, we derive the how an individual map from a process tensor is affected by a superprocess in isolation. When reduced to act on only one map $\mathcal{E}^{\se}_{\alpha+1:\alpha}$, a superprocess is just a supermap $\mathbf{S}_{\alpha+1:\alpha}$. We take $\mathcal{E}^{\se}_{\alpha+1:\alpha}$ to be the primitive free resource, which is a fixed output map, and find its image under all allowed supermaps. The result hinges on the form of $\mathcal{C}^{\zsp}_{\alpha+1:\alpha}$, which is analysed for each of the three classes of communication. In the next Sec.~\ref{app-interstep}, we will outline how these results are applied to find the free processes in Sec.~\ref{Results}.

\subsection{\textit{\texorpdfstring{$\mathcal{C}_{\alpha+1:\alpha}$}{C\_a+1:a} is a fixed output channel:}}
$\mathbf{S}_{\alpha+1:\alpha}$ is applied to $\mathcal{E}^{\se}_{\alpha+1:\alpha}$, with a fixed output communication channel explicitly accounted for via the inclusion of $\mathcal{C}^{\zsp}_{\alpha+1:\alpha}$ between $\mathcal{V}^{\szsp}_{\alpha+1:\alpha}$ and $\mathcal{W}^{\szsp}_{\alpha+1}$. $\mathcal{E}^{\se}_{\alpha+1:\alpha}$ acts on a state $\rho^{\sezsp}_\alpha$ yielding
\begin{equation} 
\begin{aligned}
& \mathbf{S}_{\alpha+1:\alpha}[\mathcal{E}^{{\se} }_{\alpha+1:\alpha}][\rho^{{\sezsp}}_\alpha]    \\ & = \mathcal{W}^{{\szsp}}_{\alpha+1} \circ (\mathcal{E}^{{\se} }_{\alpha+1:\alpha} \otimes \mathcal{C}^{\zsp}_{\alpha+1:\alpha} ) \circ \mathcal{V}^{{\szsp}}_{\alpha}[\rho^{{\sezsp}}_\alpha] \\ &     =\text{tr} \{ \mathcal{V}^{{\szsp}}_{\alpha}[\rho^{{\sezsp}}_\alpha] \} \text{tr}_\e \{\mathcal{W}^{{\szsp}}_{\alpha+1}[\sigma^{{\se}}_{\alpha+1} \otimes \tau^{\zsp}_{\alpha+1}]\} \\ & =  \delta^{\szsp}_{\alpha+1}.
\end{aligned}
\end{equation}
$\sigma^{{\se}}_{\alpha+1}$, $\tau^{\zsp}_{\alpha+1}$, and $\delta^{\szsp}_{\alpha+1}$ are arbitrary fixed states, so the image of the primitive resource under any allowed supermaps is still just other fixed output maps. 

\subsection{\textit{\texorpdfstring{$\mathcal{C}_{\alpha+1:\alpha}$}{C\_a+1:a} is any quantum channel:}} \label{app-quantumcase}
$\mathbf{S}_{\alpha+1:\alpha}$ is applied to $\mathcal{E}^{\se}_{\alpha+1:\alpha}$, which then acts on a state $\rho^{\sezsp}_{\alpha}$ yielding
\begin{equation} 
\mathbf{S}_{\alpha+1:\alpha}[\mathcal{E}^{{\se} }_{\alpha+1:\alpha}][\rho^{{\sezsp}}_{\alpha}]   =  \mathcal{W}^{{\szsp}}_{\alpha+1} \circ (\mathcal{E}^{{\se} }_{\alpha+1:\alpha} \otimes \mathcal{C}^{\zsp}_{\alpha+1:\alpha} ) \circ \mathcal{V}^{{\szsp}}_{\alpha}[\rho^{{\sezsp}}_{\alpha}] .
\end{equation}
In order to move forward with this expression, one must acknowledge that when no constraints are placed on the supermap, it is possible to convert any $\mathcal{E}^{\se}_{\alpha+1:\alpha}$ into any other map. A sufficient argument for this is below:
\begin{enumerate}
\item Within a subsystem $\s'$ of the total ancilla, create a free state $\beta$ in the underlying resource theory (via $\mathcal{R}_\beta^{{\s'}}$ which is a fixed output map to that free state), and perform a swap operation between $\beta$ and the main subsystem $\s$. This step is equivalent to the action of $\mathcal{V}^{\szsp}_{\alpha}$:
\begin{equation} 
\begin{aligned}
 (S^{{\s \s'}}_{\alpha} \circ \mathcal{R}_{\alpha,\beta}^{{\s'}})  [\rho^{{\sezsp}}_{\alpha} ] = \rho^{\e \z \s'}_{\alpha} \otimes \beta^{\s}_{\alpha}.
\end{aligned}
\end{equation}
\item $\beta$ is fed through $\mathcal{E}^{{\se} }_{\alpha+1:\alpha}$ and $\rho$ is communicated to $\mathcal{W}^{\szsp}_{\alpha+1}$ via $\mathcal{C}^{\zsp}_{\alpha+1:\alpha}$, which is an arbitrary CPTP map. The expression for this step is
\begin{equation} 
\begin{aligned}
(\mathcal{E}^{{\se} }_{\alpha+1:\alpha}\otimes \mathcal{C}^{\zsp}_{\alpha+1:\alpha}) [ \rho^{\e \z \s'}_{\alpha} \otimes \beta^{\s}_{\alpha}].
\end{aligned}
\end{equation}
\item $\s$ and $\s'$ are swapped back, then (e) and $\s'$ are discarded. This step is equivalent to the action of $\mathcal{W}^{{\sa}}_{\alpha+1}$:
\begin{equation} 
  \text{tr}_{\e,\s'} \{ S^{{\s \s'}}_{\alpha+1} \circ (\mathcal{E}^{{\se} }_{\alpha+1:\alpha}\otimes \mathcal{C}^{\zsp}_{\alpha+1:\alpha}) [ \rho^{\e \z \s'}_{\alpha} \otimes \beta^{\s}_{\alpha}] \}  =\mathcal{C}^{{\s \z}}_{\alpha+1:\alpha}[\rho^{{\s \z}}_{\alpha} ]. 
\end{equation} 
\end{enumerate}
Thus, the image an arbitrarily useless bipartite quantum channel under all supermaps with quantum communication is all bipartite quantum channels. This procedure was to remove $\mathcal{E}^{{\se} }_{\alpha+1:\alpha}$ from mathematical consideration, and transplant it with $\mathcal{C}^{{\s \z}}_{\alpha+1:\alpha}$ instead.

\begin{figure}
\centering

\begin{tikzpicture}[scale=1.1]

\draw[mypurple,fill=mypurplefill,thick,solid,rounded corners=2] (8,-1) rectangle (10,2);

\draw[mypurple,fill=mypurplefill,thick,solid,rounded corners=2] (11,-1) rectangle (13,2);

\draw[black,  thick,solid] (7.5,2.5) -- (13.5,2.5);
\draw[black,  thick,solid] (7.5,1.5) -- (13.5,1.5);
\draw[black,  thick,solid] (7.5,0.5) -- (13.5,0.5);
\draw[black,  thick,solid] (7.5,-0.5) -- (13.5,-0.5);
\draw[black,  thick,solid] (9.5,1.5) -- (9.5,-0.5);
\draw[black,  thick,solid] (11.5,1.5) -- (11.5,-0.5);
\draw[black,  thick,solid] (9.4,1.4) -- (9.6,1.6);
\draw[black,  thick,solid] (9.6,1.4) -- (9.4,1.6);
\draw[black,  thick,solid] (11.4,1.4) -- (11.6,1.6);
\draw[black,  thick,solid] (11.6,1.4) -- (11.4,1.6);
\draw[black,  thick,solid] (9.4,-0.4) -- (9.6,-0.6);
\draw[black,  thick,solid] (9.6,-0.4) -- (9.4,-0.6);
\draw[black,  thick,solid] (11.4,-0.4) -- (11.6,-0.6);
\draw[black,  thick,solid] (11.6,-0.4) -- (11.4,-0.6);

\draw[mypurple,fill=mypurplefill,very thick,solid,rounded corners=2] (8.1,-0.9) rectangle (8.9,-0.1);
\draw[myorange,fill=myorangefill,very thick,solid,rounded corners=2] (10.1,-0.9) rectangle (10.9,0.9);
\draw[myred,fill=myredfill,very thick,solid,rounded corners=2] (10.1,2.9) rectangle (10.9,1.1);

\draw[] (7,2.5) node {\Large $e$};
\draw[] (7,1.5) node {\Large $s$};
\draw[] (7,0.5) node {\Large $\z$};
\draw[] (7,-0.5) node {\Large $\s'$};

\draw[] (10.5,2) node[rotate=0] {\Large  $\mathcal{E}$};
\draw[] (10.5,0) node {\Large  $\mathcal{C}$};
\draw[] (8.5,-0.5) node {\Large   $\mathcal{R}_{\beta}$};

\draw[] (9,1) node[rotate=0] {\Large  $\mathcal{V}$};
\draw[] (12,1) node[rotate=0] {\Large  $\mathcal{W}$};

\end{tikzpicture}

\caption{A diagrammatic representation of our `transplantation' procedure. The initial system state is passed through $\mathcal{C}^{\zsp}_{\alpha+1:\alpha}$ instead of $\mathcal{E}^{\se}_{\alpha+1:\alpha}$, and the ancilla simulates the role which was previously played by the environment. Hence, $\mathcal{C}^{\zsp}_{\alpha+1:\alpha}$ is the sole determinant of what the resultant map can be.} \label{fig:3}
\end{figure}
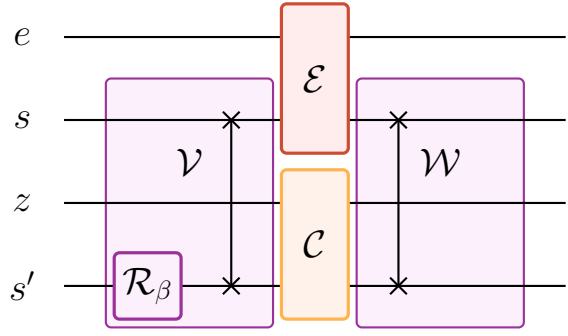

\subsection{\textit{\texorpdfstring{$\mathcal{C}_{\alpha+1:\alpha}$}{C\_a+1:a} is any entanglement breaking channel:}}
The case of entanglement breaking communication can be solved with a minor extension to what was done with quantum communication. Rather than letting $\mathcal{C}^{\zsp}_{\alpha+1:\alpha}$ be any map, we write it as a POVM measurement and subsequent re-preparation on the ancilla subsystems.
\begin{equation} 
\begin{aligned}
\mathcal{C}^{{\s \z}}_{\alpha+1:\alpha}[\rho^{{\s \z}}_{\alpha} ] =\sum_{k} \nu^{\s \z,(k)}_{\alpha+1:\alpha} \text{tr} \left\{ \Pi^{\s \z,(k)}_{\alpha+1:\alpha} \rho^{{\s \z}}_{\alpha} \right\}.
\end{aligned}
\end{equation} 
Thus, the image an arbitrarily useless bipartite quantum channel under all supermaps with entanglement breaking communication is all bipartite entanglement breaking channels.

\section{Inter-Step Behaviour} \label{app-interstep}
In the previous section we took $\mathcal{E}$ to be the fixed output map and observed how well $\mathcal{C}$ was able to bypass the blockage in information flow. Now we turn our attention to the information flow between steps. Now, $\mathcal{A}$ is taken to be the fixed output map, and $\mathcal{K}$ is used to circumvent the information blockage.

This procedure is analogous to the intra-step case. Groupings of maps $\mathcal{M}$ take the roles of $\mathcal{V}$ and $\mathcal{W}$, $\mathcal{A}$ takes on the role of $\mathcal{E}$, and $\mathcal{K}$ takes the role of $\mathcal{C}$. The rest of the mathematics is identical so we will not repeat it.

\section{Notation summary} \label{appnotationsummary}
Shown in this section is a summary of notation for the most common objects featured in this work. The type of script used for each object holds a specific meaning: calligraphic letters refer to regular quantum maps (superoperators), sans-serif letters indicate sets, boldface letters correspond to higher order maps/quantum combs with one set of inputs and outputs, while blackboard font is used to denote superprocesses -- quantum combs with two sets of inputs and outputs.
\begin{table*}[htbp] \label{tabnotationsummary}
\centering
\begin{tabular}{|l|l|}
\hline
\textbf{Object} & \textbf{Meaning}                                                              \\ \hline \hline
$\mathbf{A}$    & Control sequence                                                              \\ \hline
$\mathcal{A}$   & Control operation                                                             \\ \hline
$\mathbf{A}'$   & Client/Fiducial control sequence                                                     \\ \hline
$\mathcal{A}'$  & Client/Fiducial control operation                                                    \\ \hline
$\mathbf{T}$   & Process tensor/background process                                             \\ \hline
$\mathbf{T}'$   & Transformed process                                             \\ \hline
$\mathsf{T}$   & Set of potential process tensors                                             \\ \hline
$\mathcal{E}$   & Individual step of the background process                                     \\ \hline
$\mathbf{Z}$   & Superprocess                                                                  \\ \hline
$\mathsf{Z}$   & Set of implementable superprocesses                                                                  \\ \hline
$\mathcal{V}$   & Pre-operation in dilated form of superprocess                                               \\ \hline
$\mathcal{W}$   & Post-operation in dilated form of superprocess                                              \\ \hline
$\mathcal{C}$   & Communication map in parallel to a process tensor step                        \\ \hline
$\mathcal{K}$   & Communication map in parallel to client control operation                   \\ \hline
$\mathcal{M}$   & Individual step of transformed background process                   \\ \hline
$\llbracket \mathbf{T} | \mathbf{Z}$   & Left action, yielding effective background process                                                                  \\ \hline
$\mathbf{Z} | \mathbf{A}' \rrbracket $   & Right action, yielding full experimental control                                                                  \\ \hline
$\llbracket \mathbf{T} | \mathbf{Z} | \mathbf{A}' \rrbracket$   & Full dynamics, yielding output state                                                                  \\ \hline
\end{tabular}
\caption{Summary of notation for the most common objects featured in this work.
}
\end{table*}

\end{document}